\documentclass[sigconf,screen,nonacm]{acmart}
\usepackage{mathtools,stmaryrd,amsthm,nicefrac}
\usepackage{vwcol} 
\usepackage{multicol}
\usepackage[normalem]{ulem}
\usepackage{subcaption}
\usepackage{enumitem}
\usepackage{cancel}
\setlist[itemize]{leftmargin=*}
\usepackage{array,longtable}
\usepackage[utf8]{inputenc}
\usepackage{dashbox}
\usepackage{hyperref}
\usepackage[capitalise,noabbrev]{cleveref}
\usepackage{bm}
\usepackage{fancybox,framed}
\usepackage{blindtext}
\usepackage{titletoc}
\usepackage{thmtools,thm-restate}

\usepackage{tikzit}
\usetikzlibrary{decorations}
\usetikzlibrary{decorations.pathmorphing}
\usetikzlibrary{decorations.pathreplacing}
\usetikzlibrary{decorations.shapes}
\usetikzlibrary{decorations.text}
\usetikzlibrary{decorations.markings}
\usetikzlibrary{decorations.fractals}
\usetikzlibrary{decorations.footprints}

\tikzstyle{diamant}=[diamond, fill=couleurdefond, draw=black,inner sep=0.1em]
\tikzstyle{newe}=[rectangle, fill={gray!15}, draw=black, tikzit shape=rectangle, inner sep=0.2em]
\tikzstyle{cercle}=[circle, fill=couleurdefond, draw=black]
\tikzstyle{scercle}=[circle, fill=couleurdefond, draw=black, tikzit fill=white, inner sep=0.1em]
\tikzstyle{cartouche}=[rounded rectangle, fill=couleurdefond, draw=black,inner sep=0.2em]
\tikzstyle{neg}=[rounded rectangle, fill=couleurdefond, draw=black, execute at end node={$\neg$}]
\tikzstyle{sneg}=[rounded rectangle, fill=couleurdefond, draw=black, execute at end node={$\neg$}, scale=0.8]
\tikzstyle{negserie}=[rounded rectangle, fill=couleurdefond, draw=black, execute at end node={\footnotesize$\star\star$}]
\tikzstyle{diagrammevide}=[rectangle, fill=couleurdefond, draw=black, inner sep=1.25em, borddiagrammevide, tikzit shape=rectangle]
\tikzstyle{mdiagrammevide}=[rectangle, fill=couleurdefond, draw=black, inner sep=0.75em, sborddiagrammevide, tikzit shape=rectangle]
\tikzstyle{msdiagrammevide}=[rectangle, fill=couleurdefond, draw=black, inner sep=0.7em, msborddiagrammevide, tikzit shape=rectangle]
\tikzstyle{sdiagrammevide}=[rectangle, fill=couleurdefond, draw=black, inner sep=0.5em, sborddiagrammevide, tikzit shape=rectangle]
\tikzstyle{xsdiagrammevide}=[rectangle, fill=couleurdefond, draw=black, inner sep=0.4em, xsborddiagrammevide, tikzit shape=rectangle]
\tikzstyle{bs}=[shape=beam, fill=couleurdefond, draw, inner sep=0.25em, thick, tikzit fill=white]
\tikzstyle{sbs}=[shape=beam, fill=couleurdefond, draw, inner sep=0.2em, thick, tikzit fill=white]
\tikzstyle{npbs}=[shape=beam, horizontal fill={{npbsmoitiebasse}{npbsmoitiehaute}}, draw, inner sep=0.25em, thick, tikzit fill={rgb,255: red,128; green,128; blue,128}]
\tikzstyle{npbsalenvers}=[shape=beam, horizontal fill={{npbsmoitiehaute}{npbsmoitiebasse}}, draw, inner sep=0.25em, thick, tikzit fill={rgb,255: red,128; green,128; blue,128}]
\tikzstyle{snpbs}=[shape=beam, horizontal fill={{npbsmoitiebasse}{npbsmoitiehaute}}, draw, inner sep=0.2em, thick, tikzit fill={rgb,255: red,128; green,128; blue,128}]
\tikzstyle{snpbsalenvers}=[shape=beam, horizontal fill={{npbsmoitiehaute}{npbsmoitiebasse}}, draw, inner sep=0.2em, thick, tikzit fill={rgb,255: red,128; green,128; blue,128}]
\tikzstyle{cnot}=[shape=circle, draw, path picture={ 
\draw[black](path picture bounding box.north) -- (path picture bounding box.south) (path picture bounding box.west) -- (path picture bounding box.east);
}, tikzit fill={rgb,255: red,223; green,223; blue,223}]
\tikzstyle{thickcnot}=[shape=circle, draw, thick, path picture={ 
\draw[thick,black](path picture bounding box.north) -- (path picture bounding box.south) (path picture bounding box.west) -- (path picture bounding box.east);
}, tikzit fill={rgb,255: red,223; green,223; blue,223}]
\tikzstyle{boite22}=[fill=white, draw=black, shape=rectangle, minimum height=1cm, minimum width=0.5cm]
\tikzstyle{boite15}=[fill=white, draw=black, shape=rectangle, minimum height=0.7cm, minimum width=0.5cm]
\tikzstyle{boite2}=[fill=white, draw=black, shape=rectangle, minimum height=0cm, minimum width=0cm]
\tikzstyle{snegpotentiel}=[fill=couleurdefond, draw=black, shape=rounded rectangle, inner sep=0.25em, tikzit fill={rgb,255: red,191; green,191; blue,191}, execute at end node={\footnotesize$\star$}]
\tikzstyle{negpotentiel}=[fill=couleurdefond, draw=black, shape=rounded rectangle, tikzit fill={rgb,255: red,191; green,191; blue,191}, execute at end node={$\star$}]
\tikzstyle{token}=[fill=black, draw=black, shape=circle, inner sep=0.1em]
\tikzstyle{whitetoken}=[fill=white, draw=black, shape=circle, inner sep=0.1em]
\tikzstyle{boitePBS}=[fill=white, draw=gray, thick, shape=rectangle, rounded corners=3pt, minimum height=0.6cm, inner sep=0.1em, minimum width=0.5cm]
\tikzstyle{boitePBS2}=[fill=white, draw=gray, thick, shape=rectangle, rounded corners=3pt, minimum height=0.55cm, inner sep=0.1em, minimum width=0.5cm]
\tikzstyle{sgene}=[fill={gray!30}, draw=black, shape=rounded rectangle, rounded rectangle east arc=0pt, minimum height=0.5cm, inner sep=0em, minimum width=0cm, scale=0.8]
\tikzstyle{sdetector}=[fill={gray!30}, draw=black, shape=rounded rectangle, rounded rectangle west arc=0pt, minimum height=0.5cm, inner sep=0em, minimum width=0cm, scale=0.8]
\tikzstyle{xsgene}=[fill={gray!30}, draw=black, shape=rounded rectangle, rounded rectangle east arc=0pt, minimum height=0.5cm, inner sep=0em, minimum width=0cm, scale=0.67]
\tikzstyle{xsdetector}=[fill={gray!30}, draw=black, shape=rounded rectangle, rounded rectangle west arc=0pt, minimum height=0.5cm, inner sep=0em, minimum width=0cm, scale=0.67]
\tikzstyle{PolRot}=[fill={gray!30}, draw=black, shape=rectangle, minimum height=0.5cm, inner sep=0.1em, minimum width=0.1cm]
\tikzstyle{PhS}=[fill=white, draw=black, shape=rectangle, minimum height=0.5cm, inner sep=0.1em, minimum width=0.1cm]
\tikzstyle{gene}=[fill={gray!30}, draw=black, shape=rounded rectangle, rounded rectangle east arc=0pt, minimum height=0.5cm, inner sep=0em, minimum width=0cm]
\tikzstyle{detector}=[fill={gray!30}, draw=black, shape=rounded rectangle, rounded rectangle west arc=0pt, minimum height=0.5cm, inner sep=0em, minimum width=0cm]
\tikzstyle{cartoucherouge}=[rounded rectangle, fill={red!55!white}, draw=black, tikzit fill=red]
\tikzstyle{cartouchebleu}=[rounded rectangle, fill={blue!33!white}, draw=black, tikzit fill=blue]
\tikzstyle{diamantrouge}=[diamond, fill={rgb,255: red,255; green,115; blue,115}, draw=black]
\tikzstyle{diamantbleu}=[diamond, fill={rgb,255: red,171; green,171; blue,255}, draw=black]
\tikzstyle{control}=[fill=black, draw=black, shape=circle, scale=0.35]
\tikzstyle{wcontrol}=[fill=white, draw=black, shape=circle, scale=0.35]
\tikzstyle{boite3qubits}=[fill=white, draw=black, shape=rectangle, minimum height=2.5cm]
\tikzstyle{boite2qubits}=[fill=white, draw=black, shape=rectangle, minimum height=1.75cm]

\tikzstyle{new}=[-, tikzit draw=magenta]
\tikzstyle{tirets}=[-, draw=black, dashed]
\tikzstyle{noire}=[-, draw=black, tikzit draw=magenta]
\tikzstyle{ep}=[-, draw=black, tikzit draw=magenta]
\tikzstyle{longdashed}=[-, dash pattern=on 5pt off 5pt]
\tikzstyle{pointilles}=[-, draw=black, dotted]
\tikzstyle{trait}=[-, draw=black, thick,dashed]
\tikzstyle{boxed}=[-, draw=gray, thick,dashed]
\tikzstyle{grise}=[-, draw={rgb,255: red,191; green,191; blue,191}]
\tikzstyle{rouge}=[-, draw=red]
\tikzstyle{bleue}=[-, draw=bleu, tikzit draw=blue]
\tikzstyle{verte}=[-, draw={rgb,255: red,0; green,230; blue,0}]
\tikzstyle{borddiagrammevide}=[-, dash pattern=on 0.5em off 0.5em on 0.5em off 0.5em on 0.5em off 0em]
\tikzstyle{msborddiagrammevide}=[-, dash pattern=on 0.28em off 0.28em on 0.28em off 0.28em on 0.28em off 0em]
\tikzstyle{sborddiagrammevide}=[-, dash pattern=on 0.2em off 0.2em on 0.2em off 0.2em on 0.2em off 0em]
\tikzstyle{xsborddiagrammevide}=[-, dash pattern=on 0.16em off 0.16em on 0.16em off 0.16em on 0.16em off 0em]
\tikzstyle{mediumdash}=[-, dash pattern=on 2pt off 2pt]
\tikzstyle{rougefonce}=[-, draw={red!50!black}, tikzit draw={rgb,255: red,136; green,0; blue,0}]

\input{figures/styles-pbs.tikzdefs}

\tikzstyle{gate}=[fill=white, draw=black, shape=rectangle, minimum height=0.5cm, minimum width=0.1cm, inner sep=0.1em]
\tikzstyle{control}=[fill=black, draw=black, shape=circle, scale=0.35]
\tikzstyle{not}=[shape=circle, path picture={ 
\draw[black](path picture bounding box.north) -- (path picture bounding box.south) (path picture bounding box.west) -- (path picture bounding box.east);
}, draw=black]
\tikzstyle{wcontrol}=[fill=white, draw=black, shape=circle, scale=0.35]
\tikzstyle{empty}=[fill=white, draw=black, shape=rectangle, inner sep=0.4em, emptyborder]
\tikzstyle{globalphase}=[fill=white, draw=black, inner sep=0.15em, shape=rounded rectangle]
\tikzstyle{ancilla}=[fill=black, draw=black, shape=rectangle, minimum width=0.01cm, minimum height=0.25cm, inner sep=0.01em]
\tikzstyle{ground}=[fill=white, path picture={\draw[black](-1.5mm,0)--(-0.6mm,0);\draw[black,thick](-0.6mm,-1.75mm)--(-0.6mm,1.75mm) (0mm,-0.9mm)--(0mm,0.9mm) (0.6mm,-0.5mm)--(0.6mm,0.5mm);}, minimum width=0.1mm, draw=none, outer sep=0pt]
\tikzstyle{gate22}=[fill=white, draw=black, shape=rectangle, minimum height=0.7cm, minimum width=0.5cm]
\tikzstyle{void}=[shape=rectangle, minimum height=0.5cm]
\tikzstyle{circuit}=[fill=white, draw=black, shape=rectangle]

\tikzstyle{emptyborder}=[-, dash pattern=on 0.16em off 0.16em on 0.16em off 0.16em on 0.16em off 0em]
\tikzstyle{etc}=[-, draw=black, dashed, thick]
\tikzstyle{dots}=[-, dotted, draw=black]
\tikzstyle{shortdashed}=[-, draw=black, dash pattern=on 1pt off 1pt]
\tikzstyle{shortdashed2}=[-, draw=black, dash pattern=on 2pt off 2pt]

\input{figures/bwcontrol.tikzstyles}
\tikzstyle{every node}=[font=\LARGE]

\usepackage{macros}

\newcommand{\interp}[1]{\left\llbracket #1 \right\rrbracket}
\newcommand{\interpD}[2]{\left\llbracket #1 \right\rrbracket^\sharp_{#2}}
\newcommand{\interpEq}[2]{\left\llbracket #1 \right\rrbracket_{#2}}
\newcommand{\CPTP}[1]{\llparenthesis #1 \rrparenthesis}

\newcommand{\eqeqref}[1]{\overset{\eqref{#1}}{=}}
\newcommand{\eqdeuxeqref}[2]{\overset{\eqref{#1}\eqref{#2}}{=}}
\newcommand{\eqtroiseqref}[3]{\overset{\eqref{#1}\eqref{#2}\eqref{#3}}{=}}
\newcommand{\eqquatreeqref}[4]{\overset{\eqref{#1}\eqref{#2}\eqref{#3}\eqref{#4}}{=}}

\newcommand{\eqspace}{\ =\ }

\newcounter{eqnt}

\newcounter{eqnabc}
\newenvironment{eqnabc}{\equation\refstepcounter{eqnabc}}{\tag{\alph{eqnabc}}\endequation}
\newcounter{eqnexpr}

\newcolumntype{b}{X}
\newcolumntype{s}{>{\hsize=.5\hsize}X}
\newcolumntype{C}{>{$}c<{$}}  
\newcolumntype{R}{>{$}r<{$}}  
\newcolumntype{L}{>{$}l<{$}}  
\makeatletter
\DeclareRobustCommand{\crefnosort}[1]{%
  \begingroup\@cref@sortfalse\cref{#1}\endgroup
}
\DeclareRobustCommand{\crefnocompress}[1]{%
  \begingroup\@cref@compressfalse\cref{#1}\endgroup
}
\DeclareRobustCommand{\crefnosortnocompress}[1]{%
  \begingroup\@cref@sortfalse\@cref@compressfalse\cref{#1}\endgroup
}
\DeclareRobustCommand{\crefornocompress}[1]{%
  \begingroup\@cref@compressfalse\crefname{equation}{Equation}{Equation}\renewcommand{\crefpairconjunction}{ or~}\renewcommand{\creflastconjunction}{ or~}\cref{#1}\endgroup
}
\DeclareRobustCommand{\crefornosortnocompress}[1]{%
  \begingroup\@cref@sortfalse\@cref@compressfalse\crefname{equation}{Equation}{Equation}\renewcommand{\crefpairconjunction}{ or~}\renewcommand{\creflastconjunction}{ or~}\cref{#1}\endgroup
}
\makeatother

\makeatletter
\newcommand*\vspacebeforeline[1]{
    \ifvmode 
        \vskip #1
        \vskip \z@skip
    \else
        \@bsphack
        \vadjust pre {%
            \@restorepar
            \vskip #1
            \vskip \z@skip
        }%
        \@esphack
    \fi
}
\makeatother

\newcommand{\smallskipbeforeline}{\vspacebeforeline{3pt plus 1pt minus 1pt}}

\let\oldscalebox\scalebox
\renewcommand{\scalebox}[2]{\raisebox{2.5pt-#1pt*5/2}{\oldscalebox{#1}{#2}}}

\title{Minimal Equational Theories for Quantum Circuits}

\author{Alexandre Cl\'ement}
\affiliation{
  \institution{Universit\'e Paris-Saclay, ENS Paris-Saclay, CNRS, Inria, LMF,}
  \city{91190, Gif-sur-Yvette}
  \country{France}}
\email{alexandre.clement@inria.fr}

\author{No\'e Delorme}
\affiliation{
  \institution{Universit\'e de Lorraine, \\CNRS, Inria, LORIA,}
  \city{F-54000 Nancy}
  \country{France}}
\email{noe.delorme@inria.fr}

\author{Simon Perdrix}
\affiliation{
  \institution{Universit\'e de Lorraine, \\CNRS, Inria, LORIA,}
  \city{F-54000 Nancy}
  \country{France}}
\email{simon.perdrix@loria.fr}

\copyrightyear{2024}
\acmYear{2024}
\setcopyright{rightsretained}
\acmConference[LICS '24]{39th Annual ACM/IEEE Symposium on Logic in Computer Science}{July 8--11, 2024}{Tallinn, Estonia}
\acmBooktitle{39th Annual ACM/IEEE Symposium on Logic in Computer Science (LICS '24), July 8--11, 2024, Tallinn, Estonia}\acmDOI{10.1145/3661814.3662088}
\acmISBN{979-8-4007-0660-8/24/07}

\begin{document}
\allowdisplaybreaks[1] 

\begin{abstract}
  We introduce the first minimal and complete equational theory for quantum circuits. Hence, we show that any true equation on quantum circuits can be derived from simple rules, all of them being standard except a novel but intuitive one which states that a multi-control $2\pi$ rotation is nothing but the identity. Our work improves on the recent complete equational theories for quantum circuits, by getting rid of several rules including a fairly impractical one. One of our main contributions is to prove the minimality of the equational theory, i.e. none of the rules can be derived from the other ones. More generally, we demonstrate that any complete equational theory on quantum circuits (when all gates are unitary) requires rules acting on an unbounded number of qubits. Finally, we also simplify the complete equational theories for quantum circuits with ancillary qubits and/or qubit discarding.
\end{abstract}

\maketitle

\section{Introduction}

Since its introduction in the 1980s, the quantum circuit model  has become ubiquitous in quantum computing, and a cornerstone of quantum software. Transforming quantum circuits is a central task in the development of the quantum computer, for circuit optimisation  (such as minimizing the number of gates, T-gates, and CNot gates, as well as reducing circuit depth), for ensuring compatibility with hardware constraints, and for enabling fault-tolerant computations.  Such transformations can be achieved through the use of an equational theory, i.e. roughly speaking, a set of rules which allow one to replace a piece of circuit with an equivalent one. An equational theory is complete when any true equation can be derived. In other words, if two circuits represent the same unitary transformation, completeness ensures that they can be transformed into each other using only the rules of the equational theory. 

From a  foundational perspective, a complete equational theory can be regarded as a set of axioms or principles that govern the behaviour of quantum circuits. Therefore, it is important to establish concise and meaningful rules that accurately capture the properties of quantum circuits.

Several equational theories have been demonstrated to be complete for specific, non-universal fragments of quantum circuits, like Clifford circuits \cite{Selinger13}, 1- and 2-qubit Clifford+T circuits  \cite{bian2022generators}, 3-qubit Clifford+CS \cite{bian2023generators}, CNOT-dihedral circuits \cite{Amy_2018}. Recently, a first complete equational theory for arbitrary quantum circuits was introduced \cite{CHMPV}. Derived from a complete equational theory for photonic circuits \cite{Clement2022lov} through an elaborate completion procedure, this original complete equational theory   for quantum circuits  contained a few cumbersome rules, some of them have been shown to be derivable in \cite{CDPV}. There remained however a family of equations acting on an unbounded number of qubits, and involving a dozen parameters with non trivial relations among them (see \Cref{euler3d} in \cref{fig:qcold-axioms}), leaving open the question of whether such a family of intricate equations is necessary. 

We introduce a concise and meaningful equational theory for quantum circuits, presented in its simplest form in \cref{fig:qcsf-axioms}, where quantum circuits are considered up to global phases. The main novelty, in terms of rules, is that the intricate  rule \eqref{euler3d} is replaced by the following equation:
\begin{align*}\tag{I}
  \tikzfigM{./qc-axioms/mctrlP2pi-}\;=\;\tikzfigM{./qc-axioms/multi-id}
\end{align*}

\begin{figure*}[h]
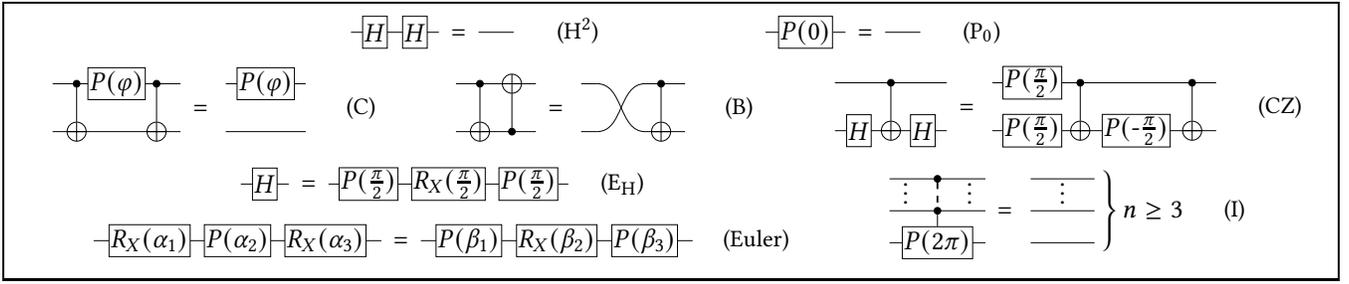

  \fbox{\begin{minipage}{.985\textwidth}\begin{center}
    \vspace{-.5em}
    \hspace{-1.5em}\begin{subfigure}{0.22\textwidth}
      \begin{align}\label{HH-ugp}\tag{H$^2$}\tikzfigM{./qc-axioms/HH}=\tikzfigM{./qc-axioms/Id}\end{align}
    \end{subfigure}\hspace{5em}
    \begin{subfigure}{0.21\textwidth}
      \begin{align}\label{P0-ugp}\tag{P$_0$}\tikzfigM{./qc-axioms/P0}=\tikzfigM{./qc-axioms/Id}\end{align}
    \end{subfigure}

    \hspace{-1.5em}\begin{subfigure}{0.28\textwidth}
      \begin{align}\label{CNOTPCNOT-ugp}\tag{C}\tikzfigM{./qc-axioms/CNOTPphiCNOT}=\tikzfigM{./qc-axioms/PphiId}\end{align}
    \end{subfigure}\hspace{1.2em}
    \begin{subfigure}{0.26\textwidth}
      \begin{align}\label{bigebre-ugp}\tag{B}\tikzfigM{./qc-axioms/CNOTNOTC}=\tikzfigM{./qc-axioms/SWAPCNOT}\end{align}
    \end{subfigure}\hspace{1.2em}
    \begin{subfigure}{0.39\textwidth}
      \begin{align}\label{CZ-ugp}\tag{CZ}\tikzfigM{./qc-axioms/H2CNOTH2}=\tikzfigM{./qc-axioms/CZ}\end{align}
    \end{subfigure}

    \hspace{-1.5em}\begin{subfigure}{0.55\textwidth}
      \begin{center}\begin{subfigure}{0.60\textwidth}\begin{align}\label{eulerH-ugp}\tag{E$_{\textup{H}}$}\tikzfigM{./qc-axioms/H}=\tikzfigM{./qc-axioms/eulerH}\end{align}\end{subfigure}\end{center}
      \vspace{0em}\begin{align}\label{euler-ugp}\tag{Euler}\tikzfigM{./qc-axioms/euler-left}=\tikzfigM{./qc-axioms/euler-right-}\end{align}
    \end{subfigure}\hspace{1.2em}
    \begin{subfigure}{0.32\textwidth}
      \begin{align}\label{ctrl2pi-ugp}\tag{I}\tikzfigM{./qc-axioms/mctrlP2pi-}=\tikzfigM{./qc-axioms/multi-id-braket}\end{align}
    \end{subfigure}

    \vspace{.5em}
  \end{center}\end{minipage}}
  \caption{\normalfont Minimal and complete equational theory for vanilla quantum circuits up to global phases.\label{fig:qcsf-axioms}}
\end{figure*}
 
Semantically, \eqref{ctrl2pi-ugp} is trivial: a $2\pi$ $Z$-rotation is nothing but the identity, hence its controlled version is also the identity whatever the number of control qubits is. Syntactically, the multi-control gate is defined inductively (see \Cref{mctrlPdef} in \cref{fig:shortcutcircuits}).

Each equation of \cref{fig:qcsf-axioms} has a simple and meaningful interpretation: \eqref{HH-ugp} means that $H$ is self inverse; \eqref{P0-ugp} that a rotation of angle $0$ is the identity; \eqref{CNOTPCNOT-ugp} that CNot is self inverse (when $\varphi=0$), and also that $P(\varphi)$ and CNot commute on the control qubit; \eqref{bigebre-ugp} essentially that composing 3 CNots is a swap; \eqref{CZ-ugp} that a Control-Z can be implemented in two ways using either one or two CNots; \eqref{eulerH-ugp} is the Euler decomposition of $H$; and finally \eqref{euler-ugp} relates two possible Euler decompositions into $Z$- and $X$-rotations. 

Our main result is to prove that the equational theory we introduce in \cref{fig:qcsf-axioms} is complete and minimal. While completeness ensures that any valid equation involving quantum circuits can be derived from these rules, minimality guarantees that none of the equations in \cref{fig:qcsf-axioms} can be derived from the other equations.

In particular, for any $n_0\ge2$, the instance of \eqref{ctrl2pi-ugp} with $n_0$ control qubits cannot be derived from the other instances of \eqref{ctrl2pi-ugp}  together with the other rules of \cref{fig:qcsf-axioms}. Beyond the minimality of this particular equational theory, one of our main contributions is to show that there is no complete equational theories acting on a bounded number of qubits for vanilla\footnote{By \emph{vanilla} quantum circuits we mean that all gates are unitary, in particular there is no qubit initialisations, ancillary qubits or discarding.}  quantum circuits.  

Minimality does not imply uniqueness. Depending on the context, it might be relevant to consider alternative equational theories. For instance, we show that one can replace the Equations \eqref{eulerH-ugp} and \eqref{euler-ugp} with Equations \eqref{Paddition-prime} and \eqref{eulerprime} leading to an equational theory which is also complete and minimal. Equation \eqref{Paddition-prime} stands for $P(\varphi_2)\circ P(\varphi_1)= P(\varphi_1+\varphi_2)$ and \eqref{eulerprime} -- introduced for the first time, up to our knowledge, in the context of the ZX-calculus \cite{coecke2018zx,vilmart2018nearoptimal} -- is an alternative formulation of the Euler decomposition with only two parameters on the LHS of the equation. 

We provide in this paper the first \emph{minimal} equational theory for quantum circuits. Indeed, the question of the minimality is still open for the complete equational theories equipping non-universal fragments of quantum circuits, like Clifford \cite{Selinger13}, 2-qubit Clifford+T \cite{bian2022generators}, and 3-qubit Clifford+CS \cite{bian2023generators}. More broadly, this is one of the first minimality results for a graphical language for quantum computing. Indeed, whereas the first completeness  results for universal graphical quantum languages have been obtained through the ZX-calculus \cite{DBLP:conf/lics/JeandelPV18,jeandel2018diagrammatic,hadzihasanovic2018two,Jeandel2020completeness}, and despite a great effort and significant progresses in the recent years, only \emph{nearly minimal}\footnote{Here \emph{nearly minimal} means that  a majority of the rules, but not all,  have been proved to be underivable form the other ones.} equational theories have been introduced for the ZX-calculus \cite{BackensPW16,BackensPW20,shi2018towards,vilmart2018nearoptimal} and its variants like the ZH-calculus \cite{van2019completeness,van2019investigating}.  In all these cases, the minimality of the provided complete equational theories is still open. Only the PBS-calculus is equipped with complete and minimal equational theories \cite{clement:hal-02929291,clement2022minimising}, notice however that the PBS-calculus focuses on coherent control and can only represent some specific oracle-based evolutions called superpositions of linear maps.  In other words, the PBS-calculus can be seen as a construction to provide coherent control capabilities to arbitrary (graphical) quantum language.

Beyond vanilla quantum circuits, one can define equational theories for quantum circuits with ancilla and/or qubit discarding (or trace out). Various constructions exist to transport a complete equational theory for vanilla quantum circuits to these settings \cite{Huot2019universal,Carette2021completeness}. In \cite{CDPV}, it has been shown that the intricate \Cref{euler3d} can be derived from its 2-qubit case in the presence of ancillary qubits and/or discarding. We show that  \Cref{euler3d} is not necessary at all. Indeed we derive from the equational theory of \cref{fig:qcsf-axioms} simple complete equational theories, acting on a bounded number of qubits (namely at most $3$), for quantum circuits with ancillary qubits and/or discarding.   

The paper is structured as follows. First,  we provide in \cref{sec:prop} the basic definitions of quantum circuits in the prop formalism. In \cref{sec:complete}, we introduce our simplified complete equational theory for vanilla quantum circuits including global phases. In \cref{sec:min}, we show that each rule of the equational theory cannot be derived from the other rules, leading to the minimality of the equational theory. An alternative equational theory is proved to be complete and minimal in \Cref{sec:alt}. And finally, we consider various extensions of vanilla quantum circuits in \cref{sec:extensions}: we introduce a complete and minimal equational theory for vanilla quantum circuits up to global phases; and complete equational theories acting on a bounded number of qubits for quantum circuits with ancilla and for quantum circuits with qubit discarding.

\section{Quantum circuits as a graphical language}\label{sec:prop}

Quantum circuits can be represented rigorously in the combinatorial structure called \emph{prop} \cite{maclane1965categorical,zanasi2018interacting,carette2020recipe}. 

\begin{definition}\label{def:prop}
A prop $\mathbf C$ of circuits  is a collection of sets $\mathbf C[n,m]$ indexed by  $n,m \in \mathbb N$. An element $C\in \mathbf C[n,m]$, denoted $C:n\to m$, is a circuit with $n$ input and $m$ output wires and is depicted:
\begin{equation*}
  \tikzfigS{./category/C}
\end{equation*}

\noindent A prop is equipped with two compositions: two circuits $C_1:n_1\to m_1$ and $C_2:n_2\to m_2$ can be composed
\begin{itemize}
\item sequentially with the associative operator $\circ$ when $m_1=n_2$, leading to $C_2 \circ C_1:n_1\to m_2$, which is graphically depicted as
\begin{gather*}
  \tikzfigS{./category/C2}\circ\tikzfigS{./category/C1}=\tikzfigS{./category/C1C2}
\end{gather*}
\item in parallel with the associative operator $\otimes$, leading to $C_1\otimes C_2:n_1+n_2\to m_1+m_2$, which is graphically depicted as
\begin{gather*}
    \tikzfigS{./category/C1}\otimes\tikzfigS{./category/C2}=\tikzfigS{./category/C1-C2}
\end{gather*}
\end{itemize}

\noindent Moreover, the compositions satisfy $$(C_3\otimes C_4)\circ (C_1\otimes C_2) = (C_3\circ C_1)\otimes (C_4\circ C_2)$$ whenever $C_3\circ C_1$ and $C_4\circ C_2$ are defined. This equation, together with the associativity of  the compositions, guarantee that the graphical representation is not ambiguous, indeed both sides of the equation are graphically depicted as $$\tikzfigS{./category/C1-C2-C3}$$

\noindent Finally, a prop also comes with three particular circuits: 
\begin{itemize}
\item The empty circuit $\gempty:0\to 0$ that satisfies $C\otimes\gempty=C=\gempty\otimes C$ for any circuit $C:n\to m$.

\item The identity circuit $\gI:1\to 1$, that can be extended to the identity  $\gI^{\otimes n}:n\to n$ on any number $n$ of wires, inductively defined as $\gI^{\otimes 0}\defeq\gempty$ and $\gI^{\otimes k+1}\defeq\gI\otimes\gI^{\otimes k}$, and that satisfies $C\circ\gI^{\otimes n}=C=\gI^{\otimes m}\circ C$ for any circuit $C:n\to m$.

\item The swap circuit $\gSWAP:2\to 2$ that satisfies $\scalebox{.69}{\begin{tikzpicture}[scale=0.7]
	\begin{pgfonlayer}{nodelayer}
		\node [style=none] (1) at (-0.75, 0.5) {};
		\node [style=none] (4) at (-0.5, -0.5) {};
		\node [style=none] (5) at (-0.75, -0.5) {};
		\node [style=none] (6) at (-0.5, 0.5) {};
		\node [style=none] (11) at (0.75, 0.5) {};
		\node [style=none] (13) at (0.75, -0.5) {};
		\node [style=none] (16) at (1.25, -0.5) {};
		\node [style=none] (18) at (1.25, 0.5) {};
		\node [style=none] (19) at (2.5, 0.5) {};
		\node [style=none] (20) at (2.75, -0.5) {};
		\node [style=none] (21) at (2.5, -0.5) {};
		\node [style=none] (22) at (2.75, 0.5) {};
	\end{pgfonlayer}
	\begin{pgfonlayer}{edgelayer}
		\draw (6.center) to (1.center);
		\draw (4.center) to (5.center);
		\draw [in=180, out=0] (6.center) to (13.center);
		\draw [in=180, out=0, looseness=0.75] (4.center) to (11.center);
		\draw (22.center) to (19.center);
		\draw (20.center) to (21.center);
		\draw [in=180, out=0] (18.center) to (21.center);
		\draw [in=180, out=0, looseness=0.75] (16.center) to (19.center);
		\draw (11.center) to (18.center);
		\draw (13.center) to (16.center);
	\end{pgfonlayer}
\end{tikzpicture}}=\scalebox{.69}{\begin{tikzpicture}[scale=0.7]
	\begin{pgfonlayer}{nodelayer}
		\node [style=none] (1) at (-0.75, 0.5) {};
		\node [style=none] (5) at (-0.75, -0.5) {};
		\node [style=none] (12) at (1, -0.5) {};
		\node [style=none] (14) at (1, 0.5) {};
	\end{pgfonlayer}
	\begin{pgfonlayer}{edgelayer}
		\draw (1.center) to (14.center);
		\draw (5.center) to (12.center);
	\end{pgfonlayer}
\end{tikzpicture}
}$, and from which we can inductively build circuits $\sigma_{k}$ that swap several wires with $\sigma_0\defeq\gI$ and $\sigma_{k+1}\defeq(\gSWAP\otimes\gI^{\otimes k})\circ(\gI\otimes\sigma_{k})$ and which satisfy $\sigma_m\circ(C\otimes\gI)=(\gI\otimes C)\circ\sigma_n$ for any $C:n\to m$. Graphically,
\begin{gather*}
   \tikzfigS{./category/naturality-left}=\tikzfigS{./category/naturality-right}
\end{gather*}
This last equation implies that the circuits are defined up to deformation. 
\end{itemize}
\end{definition}

In the language of category theory, a prop is a strict symmetric monoidal category whose objects are generated by a single object, or equivalently, generated by $(\N,+)$ as a monoid of objects.

We consider the vanilla model of quantum circuits generated by Hadamard, Phase gates, and CNot, together with global phases.

\begin{definition}[Vanilla quantum circuits]
  Let $\propQC$ be the prop of vanilla quantum circuits generated by $\gH:1\to 1$, $\gP:1\to 1$, $\gCNOT:2\to 2$ and $\gs:0\to 0$ for any $\varphi\in\R$.
\end{definition}

At this point,  circuits are essentially graphical structures. Any circuit $C:n\to n$ can be interpreted as a unitary evolution $\interp C$ on the 
Hilbert space $\C^{\{0,1\}^n}=span(\ket x, x\in \{0,1\}^n)$.

\begin{definition}[Semantics]\label{def:QCsem}
  For any $n$-qubit $\propQC$-circuit $C$, let $\interp{C}: \C^{\{0,1\}^n} \to \C^{\{0,1\}^n}$ be the \emph{semantics} of $C$ inductively defined as
  \begin{align*}
    \interp{C_2\circ C_1} &= \interp{C_2}\circ\interp{C_1}\\
    \interp{C_1\otimes C_2} &= \interp{C_1}\otimes\interp{C_2}\\
    \interp{\gempty} &= 1\mapsto 1\\
    \interp{\gI} &= \ket{x}\mapsto \ket{x}\\
    \interp{\!\gSWAP} &= \ket{x,y}\mapsto \ket{y,x}\\
    \interp{\gs} &= 1\mapsto e^{i\varphi}\\
    \interp{\gH} &= \ket x \mapsto \frac{\ket{0}+(-1)^x\ket{1}}{\sqrt{2}}\\
    \interp{\gP} &= \ket x\mapsto e^{ix\varphi}\ket{x}\\
    \interp{\!\gCNOT} &= \ket{x,y}\mapsto  \ket{x,x\oplus y}
  \end{align*}
\end{definition}

Notice that the interpretation $\interp.$ preserves the structure of the prop, i.e. if two circuits are equivalent according to the equations of props (\cref{def:prop}), they have the same interpretation\footnote{In category theoretic terms, $\interp{\cdot}$ has to be a prop functor.}. 

It is well-known that any unitary map can be represented by such a vanilla quantum circuit:
\begin{proposition}[Universality~\cite{Barenco1995gates}]\label{prop:universalityQC}
  $\propQC$ is universal, i.e.~for any unitary $U:\C^{\{0,1\}^n} \to \C^{\{0,1\}^n}$ there exists a $\propQC$-circuit $C:n\to n$ such that $\interp{C}=U$.
\end{proposition}

Quantum circuits, as defined above, only have four different kinds of generators, however, it is often convenient to use other gates that can be defined by combining them. For instance, following \cite{CHMPV,CDPV}, Pauli gates, X-rotations, and multi-controlled gates are defined in \Cref{fig:shortcutcircuits}.

Moreover, although the CNot gate acts on two consecutive qubits, it is useful to define shortcut notations for CNot gates acting on any pair of qubits as
\begin{gather*}
  \tikzfigS{./shortcut/bigCNOT-left}\defeq\tikzfigS{./shortcut/bigCNOT-right}
  \quad \textup{and} \quad \tikzfigS{./shortcut/bigNOTC-left}\defeq\tikzfigS{./shortcut/bigNOTC-right}
\end{gather*}

We use the standard bullet-based notation for multi-controlled gates. For instance 
$$\tikzfigS{./shortcut/3-Pphi}$$
denotes the application of a phase gate $\gP$ on the third qubit controlled by the first two qubits. With a slight abuse of notations, we use dashed lines for arbitrary number of control qubits, e.g. 
$$\tikzfigS{./shortcut/0+Pphi}:n+1\to n+1 \quad\textup{or simply}\quad \tikzfigS{./shortcut/0+1Pphi}:n+1\to n+1$$
have $n\ge 0$ control qubits (possibly zero), whereas 
$$\tikzfigS{./shortcut/1+Pphi}:n+2\to n+2 \quad\textup{and}\quad \tikzfigS{./shortcut/+1Pphi}:1+n+1\to 1+n+1$$
have at least one control qubit. Finally, it is often useful to define shortcut notations for multi-controlled gates with controls possibly below the target qubit.

\begin{gather*}
  \tikzfigS{./shortcut/0+Pphi+0+0}\defeq\tikzfigS{./shortcut/swaps0+0+0+Pphiswaps}
\end{gather*}

An equational theory $\Gamma$ for a prop is a collection of equations over circuits. We write $\Gamma\vdash C_1=C_2$ when $C_1$ can be transformed into $C_2$ using only the equations of $\Gamma$ together with the deformation rules of the prop. We say that an equational theory $\Gamma$ is \emph{sound} if any derivable equation preserves the interpretation, i.e.~for any circuits $C_1,C_2$, if $\Gamma\vdash C_1=C_2$ then $\interp{C_1}=\interp{C_2}$. We say that an equational theory $\Gamma$ is \emph{complete} if any true equation can be derived, i.e.~for any circuits $C_1,C_2$, if $\interp{C_1}=\interp{C_2}$ then $\Gamma\vdash C_1=C_2$.

\section{Completeness}\label{sec:complete}

The first complete equational theory for quantum circuits has been introduced in \cite{CHMPV}. It contained some simple and commonly used equations such that
\begin{gather*}
  \tikzfigS{./qc-axioms/HH}=\tikzfigS{./qc-axioms/Id} \qquad \textup{or} \qquad \tikzfigS{./qcoriginal-axioms/SWAP_00}=\tikzfigS{./qcoriginal-axioms/SWAP_04}
\end{gather*}
together with some more intricate rules. It has then been proved in \cite{CDPV} that some of those rules can actually be derived from the others, leading to a simpler equational theory $\QCold$ depicted in \Cref{fig:qcold-axioms}. Nevertheless, $\QCold$ still contains the most intricate rule, namely \Cref{euler3d}. This equation is not desirable for mainly two reasons: (1) it involves many parameters, all linked together with non-linear relations; (2) it acts on an unbounded number of qubits, i.e.~$\QCold$ contains an instance of the equation acting on $n$ qubits for any $n\ge2$.

In this section, we simplify the equational theory again, essentially by replacing \Cref{euler3d} by
\begin{equation}\tag{I}\tikzfigS{./qc-axioms/mctrlP2pi}\;=\;\tikzfigS{./qc-axioms/Id}^{\otimes n}\;,\;n\ge 3\end{equation}

This new equation still acts on an unbounded number of qubits but is significantly simpler as it does not involve any variable parameter and it is very intuitive: it states that the multi-controlled phase gate with parameter $2\pi$ is the identity. This leads to the sound equational theory $\QC$ depicted in \Cref{fig:qc-axioms}. Notice that $\QC$ has other differences compared to $\QCold$: \Cref{bigebrebis} is replaced by \Cref{bigebre}, \Cref{S0} is replaced by \Cref{gphaseempty}, \Cref{3CNOTstarget} is removed and \Cref{eulerconditioned} is replaced by \Cref{euler}. Notice that this last modification only changes the relations between the parameters $\alpha_i$ and the parameters $\beta_j$: \Cref{eulerconditioned} is presented with uniqueness conditions (cf.~caption of \Cref{fig:qcold-axioms}), whereas \Cref{euler} is equipped with functions to compute the parameters $\beta_j$ from the parameters $\alpha_i$ (cf.~\Cref{sec:discussionEuler} and caption of \Cref{fig:qc-axioms}).

To prove the completeness of the new equational theory $\QC$, we simply derive all the rules of the former complete equational theory $\QCold$.

\begin{figure*}
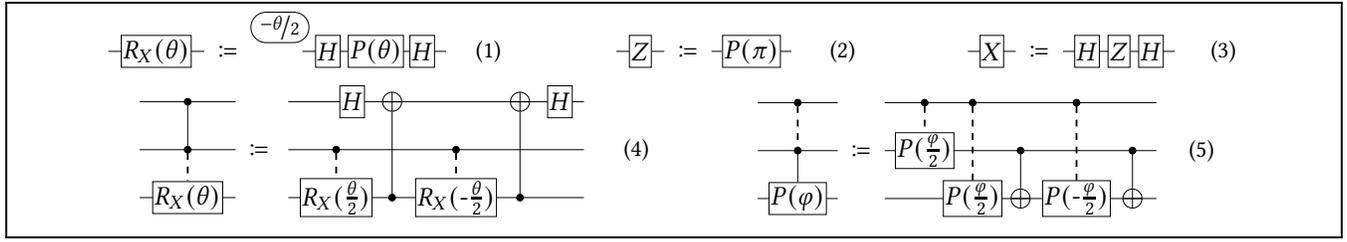

  \fbox{\begin{minipage}{.985\textwidth}\begin{center}
    \vspace{-.5em}
    \hspace{-1.5em}\begin{subfigure}{0.32\textwidth}
      \begin{align}\label{RXdef}\tikzfigM{./gates/RXtheta}\defeq\tikzfigM{./shortcut/HPthetaH}\end{align}
    \end{subfigure}\hspace{3em}
    \begin{subfigure}{0.21\textwidth}
      \begin{align}\label{Zdef}\tikzfigM{./gates/Z}\defeq\tikzfigM{./shortcut/Ppi}\end{align}
    \end{subfigure}\hspace{3em}
    \begin{subfigure}{0.23\textwidth}
      \begin{align}\label{Xdef}\tikzfigM{./gates/X}\defeq\tikzfigM{./shortcut/HZH}\end{align}
    \end{subfigure}

    \hspace{-1.5em}\begin{subfigure}{0.42\textwidth}
      \begin{align}\label{mctrlRXdef}\tikzfigM{./shortcut/mctrlRXtheta}\defeq\tikzfigM{./shortcut/mctrlRXthetadef}\end{align}
    \end{subfigure}\hspace{3em}
    \begin{subfigure}{0.37\textwidth}
      \begin{align}\label{mctrlPdef}\tikzfigM{./shortcut/mctrlPphi-}\defeq\tikzfigM{./shortcut/mctrlPphidef-}\end{align}
    \end{subfigure}
    \vspace{.5em}
  \end{center}\end{minipage}}
  \caption{\normalfont Shortcut notations for usual gates defined for any $\varphi,\theta\in\R$. \Cref{RXdef} defines $X$-rotations while Equations \eqref{Zdef} and \eqref{Xdef} define Pauli gates. Equations \eqref{mctrlRXdef} and \eqref{mctrlPdef} are inductive definitions of multi-controlled gates.\label{fig:shortcutcircuits}}
\end{figure*}

\begin{figure*}
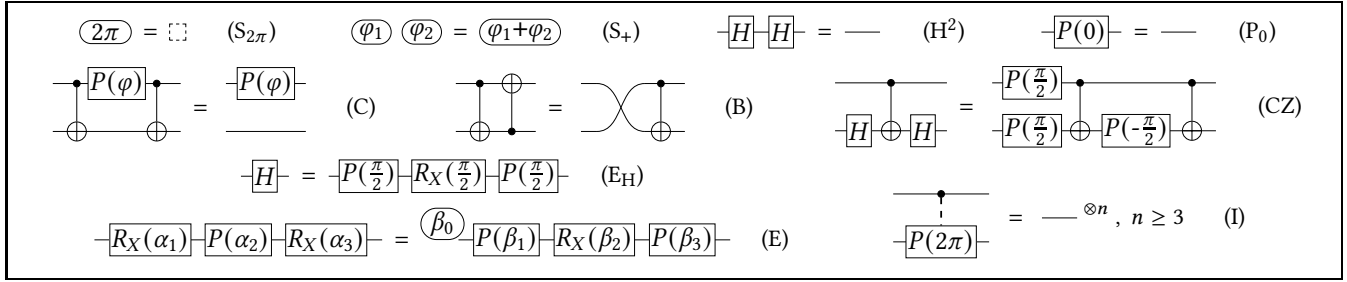

  \fbox{\begin{minipage}{.985\textwidth}\begin{center}
    \vspace{-.5em}
    \hspace{-1.5em}\begin{subfigure}{0.18\textwidth}
    \begin{align}\label{gphaseempty}\tag{S$_{2\pi}$}\tikzfigM{./qc-axioms/s2pi}=\tikzfigM{./qc-axioms/empty}\end{align}
    \end{subfigure}\hspace{1.2em}
    \begin{subfigure}{0.25\textwidth}
      \begin{align}\label{gphaseaddition}\tag{S$_+$}\tikzfigM{./qc-axioms/sphi1}\tikzfigM{./qc-axioms/sphi2}=\tikzfigM{./qc-axioms/sphi1plusphi2}\end{align}
    \end{subfigure}\hspace{1.2em}
    \begin{subfigure}{0.22\textwidth}
      \begin{align}\label{HH}\tag{H$^2$}\tikzfigM{./qc-axioms/HH}=\tikzfigM{./qc-axioms/Id}\end{align}
    \end{subfigure}\hspace{1.2em}
    \begin{subfigure}{0.21\textwidth}
      \begin{align}\label{P0}\tag{P$_0$}\tikzfigM{./qc-axioms/P0}=\tikzfigM{./qc-axioms/Id}\end{align}
    \end{subfigure}

    \hspace{-1.5em}\begin{subfigure}{0.28\textwidth}
      \begin{align}\label{CNOTPCNOT}\tag{C}\tikzfigM{./qc-axioms/CNOTPphiCNOT}=\tikzfigM{./qc-axioms/PphiId}\end{align}
    \end{subfigure}\hspace{1.2em}
    \begin{subfigure}{0.26\textwidth}
      \begin{align}\label{bigebre}\tag{B}\tikzfigM{./qc-axioms/CNOTNOTC}=\tikzfigM{./qc-axioms/SWAPCNOT}\end{align}
    \end{subfigure}\hspace{1.2em}
    \begin{subfigure}{0.39\textwidth}
      \begin{align}\label{CZ}\tag{CZ}\tikzfigM{./qc-axioms/H2CNOTH2}=\tikzfigM{./qc-axioms/CZ}\end{align}
    \end{subfigure}
    \vspace{-.3em}

    \hspace{-1.5em}\begin{subfigure}{0.55\textwidth}
      \begin{center}\begin{subfigure}{0.60\textwidth}\begin{align}\label{eulerH}\tag{E$_{\textup{H}}$}\tikzfigM{./qc-axioms/H}=\tikzfigM{./qc-axioms/eulerH}\end{align}\end{subfigure}\end{center}
      \vspace{-.2em}\begin{align}\label{euler}\tag{E}\tikzfigM{./qc-axioms/euler-left}=\tikzfigM{./qc-axioms/euler-right}\end{align}
    \end{subfigure}\hspace{1.2em}
    \begin{subfigure}{0.32\textwidth}
      \begin{align}\label{ctrl2pi}\tag{I}\tikzfigM{./qc-axioms/mctrlP2pi}\;=\;\tikzfigM{./qc-axioms/Id}^{\otimes n}\;,\;n\ge 3\end{align}
    \end{subfigure}
    \vspace{.5em}
  \end{center}\end{minipage}}
  \caption{\normalfont Equational theory $\QC$. Equations \eqref{gphaseaddition} and \eqref{CNOTPCNOT} are defined for any $\varphi,\varphi_1,\varphi_2\in\R$. In Equation \eqref{euler} the angles $\alpha_1,\alpha_2,\alpha_3\in\R$ are arbitrary parameters, whereas the angles $\beta_0,\beta_1,\beta_2,\beta_3$ are restricted to $[0,2\pi)$ and are computed as follows: $z\defeq \cos\left(\frac{\alpha_2}{2}\right)\cos\left(\frac{\alpha_1+\alpha_3}{2}\right)+i\sin\left(\frac{\alpha_2}{2}\right)\cos\left(\frac{\alpha_1-\alpha_3}{2}\right)$, $z'\defeq \cos\left(\frac{\alpha_2}{2}\right)\sin\left(\frac{\alpha_1+\alpha_3}{2}\right)-i\sin\left(\frac{\alpha_2}{2}\right)\sin\left(\frac{\alpha_1-\alpha_3}{2}\right)$. If $z'=0$ then $\beta_0\defeq\frac{\alpha_2}{2}-\arg(z)$, $\beta_1\defeq2\arg(z)$, $\beta_2\defeq0$ and $\beta_3\defeq0$. If $z=0$ then $\beta_0\defeq\frac{\alpha_2}{2}-\arg(z')$, $\beta_1\defeq2\arg(z')$, $\beta_2\defeq\pi$ and $\beta_3\defeq0$. Otherwise $\beta_0\defeq\frac{\alpha_2}{2}-\arg(z)$, $\beta_1\defeq\arg(z)+\arg(z')$, $\beta_2\defeq2\arg\left(i+\left\lvert\frac{z}{z'}\right\rvert\right)$ and $\beta_3\defeq\arg(z)-\arg(z')$. \Cref{ctrl2pi} states that the control phase gate of angle $2\pi$ is the identity for any $n\ge 3$ qubits. The equation is actually sound for any number of qubits, but derivable from the other rules for 1 qubit and 2 qubits. \label{fig:qc-axioms}}
\end{figure*}

\begin{figure*}
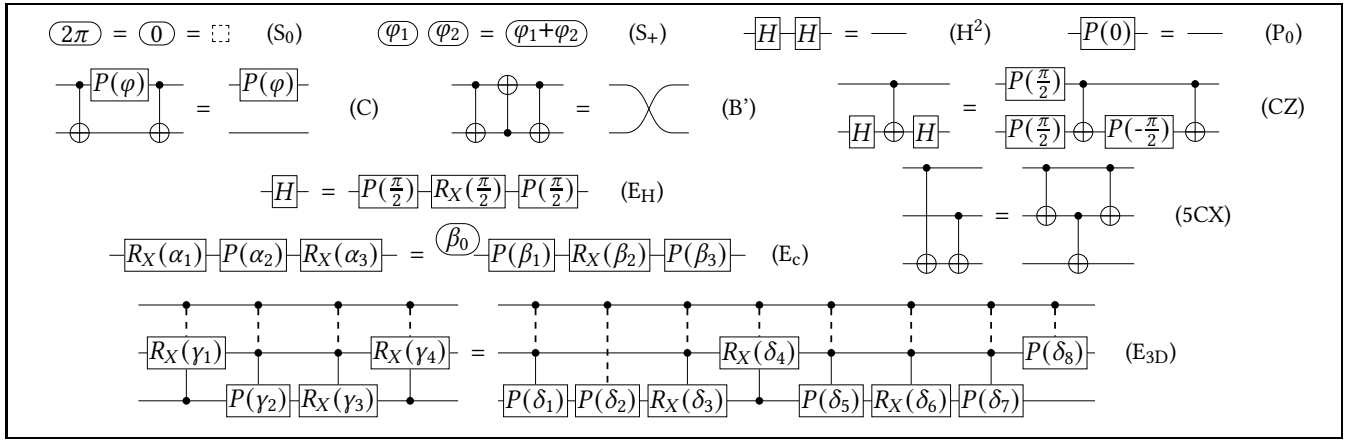

  \fbox{\begin{minipage}{.985\textwidth}\begin{center}
    \vspace{-.5em}
    \hspace{-1.5em}\begin{subfigure}{0.22\textwidth}
    \begin{align}\label{S0}\tag{S$_0$}\tikzfigM{./qc-axioms/s2pi}=\tikzfigM{./qcoriginal-axioms/s0}=\tikzfigM{./qc-axioms/empty}\end{align}
    \end{subfigure}\hspace{1.2em}
    \begin{subfigure}{0.25\textwidth}
      \begin{align}\tag{S$_+$}\tikzfigM{./qc-axioms/sphi1}\tikzfigM{./qc-axioms/sphi2}=\tikzfigM{./qc-axioms/sphi1plusphi2}\end{align}
    \end{subfigure}\hspace{1.2em}
    \begin{subfigure}{0.22\textwidth}
      \begin{align}\tag{H$^2$}\tikzfigM{./qc-axioms/HH}=\tikzfigM{./qc-axioms/Id}\end{align}
    \end{subfigure}\hspace{1.2em}
    \begin{subfigure}{0.21\textwidth}
      \begin{align}\tag{P$_0$}\tikzfigM{./qc-axioms/P0}=\tikzfigM{./qc-axioms/Id}\end{align}
    \end{subfigure}

    \hspace{-1.5em}\begin{subfigure}{0.28\textwidth}
      \begin{align}\tag{C}\tikzfigM{./qc-axioms/CNOTPphiCNOT}=\tikzfigM{./qc-axioms/PphiId}\end{align}
    \end{subfigure}\hspace{1.2em}
    \begin{subfigure}{0.26\textwidth}
      \begin{align}\label{bigebrebis}\tag{B'}\tikzfigM{./qcoriginal-axioms/SWAP_00}=\tikzfigM{./qcoriginal-axioms/SWAP_04}\end{align}
    \end{subfigure}\hspace{1.2em}
    \begin{subfigure}{0.39\textwidth}
      \begin{align}\tag{CZ}\tikzfigM{./qc-axioms/H2CNOTH2}=\tikzfigM{./qc-axioms/CZ}\end{align}
    \end{subfigure}

    \hspace{-1.5em}\begin{subfigure}{0.55\textwidth}
      \begin{center}\begin{subfigure}{0.60\textwidth}\begin{align}\tag{E$_{\textup{H}}$}\tikzfigM{./qc-axioms/H}=\tikzfigM{./qc-axioms/eulerH}\end{align}\end{subfigure}\end{center}
      \vspace{-.2em}\begin{align}\label{eulerconditioned}\tag{E$_\textup{c}$}\tikzfigM{./qc-axioms/euler-left}=\tikzfigM{./qc-axioms/euler-right}\end{align}
    \end{subfigure}\hspace{2em}
    \begin{subfigure}{0.28\textwidth}
      \begin{align}\label{3CNOTstarget}\tag{5CX}\tikzfigM{./qcoriginal-axioms/3CNOTstarget-step-0}=\tikzfigM{./qcoriginal-axioms/3CNOTstarget-step-8}\end{align}
    \end{subfigure}

    \vspace{.5em}
    \hspace{-2.5em}\begin{subfigure}{.81\textwidth}
      \begin{align}\label{euler3d}\tag{E$_{\textup{3D}}$}\tikzfigM{./qcoriginal-axioms/euler3d-left}=\tikzfigM{./qcoriginal-axioms/euler3d-right}\end{align}
    \end{subfigure}
    \vspace{.5em}
  \end{center}\end{minipage}}
  \caption{\normalfont Old complete equational theory $\QCold$ from \cite{CDPV} for vanilla quantum circuits. Equations \eqref{gphaseaddition} and \eqref{CNOTPCNOT} are defined for any $\varphi,\varphi_1,\varphi_2\in\R$. In Equations \eqref{eulerconditioned} and \eqref{euler3d}, the LHS circuit has arbitrary parameters which uniquely determine the parameters of the RHS circuit. Thus, for any $\alpha_i\in \mathbb R$, there exist $\beta_j\in \mathbb R$ such that \Cref{eulerconditioned} is sound. We make the angles $\beta_j$ unique by assuming that $\beta_1 \in [0,\pi)$, $\beta_0,\beta_2,\beta_3\in[0,2\pi)$ and if $\beta_2\in\{0,\pi\}$ then $\beta_1=0$. \Cref{euler3d} reads as follows: the equation is defined for any $n\ge 2$ input qubits, in such a way that all gates are controlled by the first $n-2$ qubits. Similarly to \Cref{eulerconditioned}, for any $\gamma_i\in\mathbb R$, there exist $\delta_j\in\mathbb R$ such that \Cref{euler3d} is sound. We ensure that the angles $\delta_j$ are uniquely determined by assuming that $\delta_1,\delta_2,\delta_5\in[0,\pi)$, $\delta_3,\delta_6,\delta_7,\delta_8\in[0,2\pi)$, $\delta_4\in[0,4\pi)$, if $\delta_3=0$ and $\delta_6\neq0$ then $\delta_2=0$, if $\delta_3=\pi$ then $\delta_1=0$, if $\delta_4\in\{0,2\pi\}$ then $\delta_1=\delta_3=0$, if $\delta_4\in\{\pi,3\pi\}$ then $\delta_2=0$, if $\delta_4\in\{\pi,3\pi\}$ and $\delta_3=0$ then $\delta_1=0$, and if $\delta_6\in\{0, \pi\}$ then $\delta_5=0$.\label{fig:qcold-axioms}}
\end{figure*}

\subsection{Euler-decomposition rule and  $1$-qubit completeness}\label{sec:discussionEuler}

In this subsection, we discuss one particular rule of the equational theory, namely \Cref{euler}, which is a cornerstone of the completeness result for quantum circuits.
\begin{equation}\tag{E}
  \tikzfigS{./qc-axioms/euler-left}=\tikzfigS{./qc-axioms/euler-right}
\end{equation}

\Cref{euler} follows from the well-known Euler-decomposition  which states that any unitary can be decomposed, up to a global phase, into basic X- and Z-rotations. Note that, in its previous version \cite{CHMPV,CDPV}, the Euler-decomposition rule was presented with conditions ensuring the uniqueness of the angles $\beta_0,\beta_1,\beta_2,\beta_3$ (see \Cref{eulerconditioned}). In the present paper, following \cite{vilmart2018nearoptimal,coecke2018zx}, we provide explicit functions: the angles $\alpha_1,\alpha_2,\alpha_3\in\R$ are arbitrary parameters, whereas the angles $\beta_0,\beta_1,\beta_2,\beta_3$ are restricted to $[0,2\pi)$ and are computed as follows, using intermediate complex numbers $z,z'$: 
\begin{gather*}
  z\defeq \cos\left(\frac{\alpha_2}{2}\right)\cos\left(\frac{\alpha_1+\alpha_3}{2}\right)+i\sin\left(\frac{\alpha_2}{2}\right)\cos\left(\frac{\alpha_1-\alpha_3}{2}\right) \\
  z'\defeq \cos\left(\frac{\alpha_2}{2}\right)\sin\left(\frac{\alpha_1+\alpha_3}{2}\right)-i\sin\left(\frac{\alpha_2}{2}\right)\sin\left(\frac{\alpha_1-\alpha_3}{2}\right)
\end{gather*}
\begin{itemize}
  \item If $z'=0$ then $\beta_0\defeq\frac{\alpha_2}{2}-\arg(z)$, $\beta_1\defeq2\arg(z)$, $\beta_2\defeq0$ and $\beta_3\defeq0$.
  \item If $z=0$ then $\beta_0\defeq\frac{\alpha_2}{2}-\arg(z')$, $\beta_1\defeq2\arg(z')$, $\beta_2\defeq\pi$ and $\beta_3\defeq0$.
  \item Otherwise $\beta_0\defeq\frac{\alpha_2}{2}-\arg(z)$, $\beta_1\defeq\arg(z)+\arg(z')$, $\beta_2\defeq2\arg\left(i+\left\lvert\frac{z}{z'}\right\rvert\right)$ and $\beta_3\defeq\arg(z)-\arg(z')$.
\end{itemize}
In the above definition of the $\beta_i$, we use implicit modulo $2\pi$ so that every $\beta_i$ is  in  $[0, 2\pi)$, note that $\beta_2$ is then actually in $[0,\pi]$.

We illustrate how $\QC$ can be used to derive basic but essential properties by showing, in particular, that the phase gates form a group and their parameters are $2\pi$-periodic:

\begin{restatable}{proposition}{simplederivableequations}\label{prop:simplederivableequations}
  For any $\varphi_1,\varphi_2\in \mathbb R$, the following equations can be derived in $\QC$: 
 \begin{gather}
    \label{P2pi}\tag{P$_{2\pi}$}\tikzfigS{./qcprime-completeness/P2pi_00}=\tikzfigS{./qcprime-completeness/P2pi_04}\\
    \label{Paddition}\tag{P$_+$}\tikzfigS{./qcoriginal-axioms/Pphi1Pphi2}\;=\;\tikzfigS{./qcoriginal-axioms/Pphi1phi2}\\
    \label{XPX}\tag{P$_-$}\tikzfigS{./qcprime-completeness/XPX_00}=\tikzfigS{./qcprime-completeness/XPX_06}
  \end{gather}
\end{restatable}
\begin{proof}
  The idea of the derivations is to apply \cref{euler} by introducing new gates. The details are given in \cref{appendix:proofsimplederivableequations}.
\end{proof}

We introduce an Euler-based normal form for 1-qubit circuits, which corresponds to the RHS of \eqref{euler}:

\begin{definition}[Normal form]\label{def:normalform}
  A normal form is a $\propQC$-circuit of the form $$\tikzfigS{qc-axioms/euler-right}$$ with parameters satisfying $\beta_0,\beta_1,\beta_3\in[0,2\pi)$, $\beta_2\in[0,\pi]$, and if $\beta_2\in\{0,\pi\}$ then $\beta_3=0$.
\end{definition}
 
Any $1$-qubit circuit can be turned into an Euler-based normal form.

\begin{restatable}{lemma}{procedureNormalFormEuler}\label{lem:procedureNormalFormEuler}
  Any $1$-qubit $\propQC$-circuit can be  put in normal form using the axioms of $\QC$.
\end{restatable}
\begin{proof} We introduce a simple procedure that turns any $1$-qubit circuit into a circuit in normal form. Starting with a circuit made of $\gH$, $\gP$, and $\gs$, we  remove all $\gH$ using \eqref{eulerH}, leading to a  circuit that can be described using only $\gRX$, $\gP$ and $\gs$, and satisfying the following properties: (i) each $R_X$ gate has a parameter in $[0,\pi]$, and (ii) each pair of $R_X$ gates is separated with at least one phase gate. We repeatedly  use  \eqref{euler} and \eqref{Paddition} to decrease respectively the number of $R_X$ gates and the total number of gates, while preserving  properties (i) and (ii). This leads to a circuit of the form $\tikzfigS{qc-axioms/euler-right--}$ up to global phases (and possibly some missing gates that can be trivially added with a $0$ parameter). If $\gamma_2=0$, one can remove the $R_X$ gate, merge the two phase gates using \eqref{Paddition} and put back a $R_X(0)$ gate and an $P(0)$ gate. If $\gamma_2=\pi$, one can introduce a $R_X(-\pi)$ and $R_X(\pi)$ gates at the end of the circuit (which is possible using the axioms of $\QC$) and apply $\eqref{euler}$ with parameters $\alpha_1 = \pi$, $\alpha_2 = \gamma_3$ and $\alpha_3=-\pi$. In this particular case $z'=0$ so $\beta_2=0$ and $\beta_3=0$. To finish the proof it is sufficient to apply  \eqref{Paddition} and  \eqref{P2pi} to satisfy $\beta_1, \beta_3 \in [0,2\pi)$ ; and  \eqref{gphaseaddition} and \eqref{gphaseempty} to get a single global phase parametrised with $\beta_0\in [0,2\pi)$. Notice that the numbers of steps of the entire procedure is linear in the size of the original circuit except the last stage which consists in putting the parameters $\beta_0, \beta_1$ and $\beta_3$ in the appropriate interval which depends on the actual value of the parameters. This last stage can actually be done in a constant number of steps using \eqref{euler}.
\end{proof}

\begin{restatable}{lemma}{unicityNormalFormEuler}\label{lem:unicityNormalFormEuler}
  For any $1$-qubit unitary there exists a unique circuit $C$ in normal form which implements $U$, i.e. $\interp C = U$.
\end{restatable}
\begin{proof}
  The existence can be seen as a direct consequence of  universality (\Cref{prop:universalityQC}) together with \Cref{lem:procedureNormalFormEuler}. It remains to prove the uniqueness. Let $\beta_0,\beta_1,\beta_3, \beta'_0,\beta'_1,\beta'_3,\in [0,2\pi)$ and $\beta_2,\beta_2'\in [0,\pi]$ satisfying the conditions of the normal form and such that $$\interp{\tikzfigS{./qc-axioms/euler-right}}=\interp{\tikzfigS{./qcprime-axioms/eulerprimebis-right}}$$
  The corresponding unitaries are of the form
  \begin{equation*}
    e^{i\beta_0}\begin{pmatrix}
      \cos(\frac{\beta_2}{2}) & -ie^{i\beta_1}\sin(\frac{\beta_2}{2}) \\
      -ie^{i\beta_3}\sin(\frac{\beta_2}{2}) & e^{i(\beta_1+\beta_3)}\cos(\frac{\beta_2}{2})
    \end{pmatrix}
  \end{equation*}
  which implies the equations
  \begin{equation*}
    \begin{cases}
      e^{i\beta_0}\cos(\frac{\beta_2}{2})=e^{i\beta_0'}\cos(\frac{\beta_2'}{2})\\
      e^{i(\beta_0+\beta_1)}\sin(\frac{\beta_2}{2})=e^{i(\beta_0'+\beta_1')}\sin(\frac{\beta_2'}{2})\\
      e^{i(\beta_0+\beta_3)}\sin(\frac{\beta_2}{2})=e^{i(\beta_0'+\beta_3')}\sin(\frac{\beta_2'}{2})
    \end{cases}
  \end{equation*}
  The first equation gives $\cos(\frac{\beta_2}{2})=\pm\cos(\frac{\beta_2'}{2})$. By hypothesis, $\beta_2,\beta_2'\in[0,\pi]$, which implies that $\cos(\frac{\beta_2}{2}),\cos(\frac{\beta_2'}{2})\geq0$, so that $\cos(\frac{\beta_2}{2})=\cos(\frac{\beta_2'}{2})$. Therefore, $\frac{\beta_2}{2}=\pm\frac{\beta_2'}{2}\bmod{2\pi}$, and since $\beta_2,\beta_2'\in[0,\pi]$, it must be the case that $\beta_2=\beta_2'$.
  Again, the first equation gives $\beta_0=\beta_0'\bmod{2\pi}$ and because $\beta_0,\beta_0'\in[0,2\pi)$ it must be the case that $\beta_0=\beta_0'$. Similarly, the second and third equations give $\beta_1=\beta_1'$ and $\beta_3=\beta_3'$ respectively.
\end{proof}

The existence and uniqueness of the normal form ensure the completeness for $1$-qubit $\propQC$-circuits:
\begin{proposition}[1-qubit Completeness]
  For any $1$-qubit circuits $C_1, C_2$ s.t. $\interp {C_1} = \interp{C_2}$, we have $\QC\vdash C_1 = C_2$. 
\end{proposition}
\begin{proof}
  Given some $1$-qubit circuits $C_1, C_2$ s.t. $\interp {C_1} = \interp{C_2}$, both circuits can be put in normal forms $C_1'$ and $C_2'$ respectively. By soundness $\interp {C'_1} = \interp{C'_2}$, so, according to \Cref{lem:unicityNormalFormEuler}, $C'_1=C'_2$. 
\end{proof}

Thanks to the 1-qubit completeness, we can  prove in particular  the following \Cref{eulerstar} which is a generalisation of both \Cref{eulerconditioned} and \Cref{euler} where the parameters $\beta_j$ are not restricted by functions (like \Cref{euler}) neither by conditions (like \Cref{eulerconditioned}). In practice, \Cref{eulerstar} is  more convenient to use than \Cref{euler} as one does not have to verify some relations between the parameters.

\begin{corollary}
  The following equation
  \begin{equation}\label{eulerstar}\tag{E$^*$}
    \tikzfigS{./qc-axioms/euler-left}=\tikzfigS{./qc-axioms/euler-right}
  \end{equation}
  is derivable in $\QC$ whenever the parameters $\alpha_i$ and $\beta_j$ make the equation sound.
\end{corollary}

\subsection{Proving Equation \eqref{euler3d}}\label{sec:proofEuler3D}

In order to prove \Cref{euler3d}, we need a significant amount of intermediate circuit identities. For instance, the following (derivable) equations
\begin{equation*}
  \tikzfigS{./identities/HHCNOTHH_00}=\tikzfigS{./identities/NOTC} \quad \textup{and} \quad
  \tikzfigS{./identities/Pphasegadget-step-0}=\tikzfigS{./identities/Pphasegadget-step-3}
\end{equation*}
are not axioms of the equational theory $\QC$, but it is very useful to use them as intermediate steps in a bigger derivation without having to derive them each time. Moreover, it is useful to have some equations involving multi-controlled gates. For instance,
\begin{gather*}
  \tikzfigS{./examples/cPc}=\tikzfigS{./examples/ccP}=\tikzfigS{./examples/SWAPccPSWAP}
\end{gather*}
or
\begin{gather*} \tikzfigS{./examples/mctrlPphi1mctrlPphi2}=\tikzfigS{./examples/mctrlPphi1phi2}
\end{gather*}

Most of the required identities have already been proved in \cite{CHMPV,CDPV} with the former equational theories. To avoid having to prove again each of these identities in $\QC$ individually, we state the following lemma, where $\QC_k$ denotes the equational theory composed of all the equations of $\QC$ acting on at most $k$ qubits.

\begin{restatable}{lemma}{allAxiomExeptEulerthreeD}\label{lem:allAxiomExeptEuler3D}
  All the equations of $\QCold$, except \Cref{euler3d}, are derivable in $\QC_3$.
\end{restatable}
\begin{proof}
  The proof of \eqref{eulerconditioned} is discussed in \Cref{sec:discussionEuler}. It is then sufficient to derive Equations \eqref{S0}, \eqref{bigebrebis} and \eqref{3CNOTstarget} in $\QC$ (the other equations are directly in $\QC$). The details are given in \Cref{appendix:proofallAxiomExeptEuler3D}.
\end{proof}

As a consequence of this lemma, all the results that have been proved with $\QCold$ in \cite{CDPV} without using \Cref{euler3d} are also provable with $\QC$. Moreover, all the axioms of the original complete equational theory introduced in \cite{CHMPV}, except \Cref{euler3d}\footnote{Actually, the original complete equational theory introduced in \cite{CHMPV} does not contain \Cref{euler3d} but a slightly different equation with one more parameter. It has then been proved in \cite{CDPV} that this parameter can be removed. For simplicity, we denote both version of this intricate equation by \Cref{euler3d}.} have been proved with $\QCold$ in \cite{CDPV} without using \Cref{euler3d}. Hence, the results that have been proved in \cite{CHMPV} without using \Cref{euler3d} are also provable in $\QC$. We state and sum up all these results in \Cref{appendix:intermediateresults} together with some new useful equations.

To prove \Cref{euler3d}, a key idea is to break down this equation into two simpler (although still non-trivial) equations that, when considered together, enable us to find a derivation of \Cref{euler3d}. \Cref{Euler2dmulticontrolled} is a generalisation of the Euler \Cref{eulerstar} to multi-controlled gates acting on any number of qubits. \Cref{Euler3dsansphases-multicontrolled} is a variant of \Cref{euler3d} where the multi-controlled phase gates have been removed.

\begin{lemma}\label{propEuler2dmulticontrolled}
  The following equation is a consequence of the axioms of $\QC$:
  \begin{equation}\label{Euler2dmulticontrolled}\tag{\ref*{euler}$^*_n$}\tikzfigS{./Preuve-Euler2d/Euler2dleft-multicontrolled-compact}=\tikzfigS{./Preuve-Euler2d/Euler2dright-multicontrolled-compact}\end{equation}
  where the angles are as in \Cref{eulerstar}, $n\geq1$ is the number of qubits, and $\tikzfigS{./Preuve-Euler2d/mctrlphasebeta0}$ denotes either $\tikzfigS{./Preuve-Euler2d/phasebeta0}$ on $0$ qubits or $\tikzfigS{./Preuve-Euler2d/mctrlPbeta0}$ on one or more qubits.
\end{lemma}
\begin{proof}
  The proof is by induction on the number of qubits. Intuitively, after unfolding the definition of the multi-controlled gates given by \cref{mctrlRXdef,mctrlPdef}, we use the induction hypothesis, together with elementary auxiliary equations, to put both sides of \cref{Euler2dmulticontrolled} in a kind of normal form. Then, using the equality of the semantics of the two sides, we show that the two circuits obtained are equal up to simple transformations. The details are given in~\Cref{preuveEuler2dmulticontrolled}.
\end{proof}

\begin{lemma}\label{propEuler3dsansphases-multicontrolled}
  The following equation is a consequence of the axioms of $\QC$:
  \begin{equation}\label{Euler3dsansphases-multicontrolled}\tag{$\textup{E}^{\textup{phase-free}}_{\textup{3D}}$}
    \tikzfigS{./Preuve-Euler-sans-phases/Euler3dsansphasesleft-multicontrolled}=\tikzfigS{./Preuve-Euler-sans-phases/Euler3dsansphasesright-multicontrolled}
  \end{equation}
  where $\gamma_1=2\alpha_1$, $\gamma_3=-2\alpha_2$, $\gamma_4=2\alpha_3$, $\delta_3=-2\beta_1$, $\delta_4=2\beta_2$, and $\delta_6=-2\beta_3$, with $\alpha_i$ and $\beta_i$ being parameters of any instance of \Cref{eulerstar}.
\end{lemma}
\begin{proof}
  Intuitively, after unfolding the definition of the multi-controlled gates, we can do elementary transformations in order to be able to apply Equation~\textup{(\hyperref[Euler2dmulticontrolled]{\ref*{euler}$^*_{n-1}$})}. The full derivation is given in~\Cref{preuveEuler3dsansphases-multicontrolled}.
\end{proof}

\begin{proposition}\label{prop:proofeuler3d}
  \Cref{euler3d} is derivable in $\QC$.
\end{proposition}
\begin{proof}
  Intuitively, the key is to introduce a multi-controlled $R_X$ gate together with its inverse, with a well-chosen angle so that after applying \cref{Euler3dsansphases-multicontrolled}, one of the gates has parameter $\pi$, which makes it essentially commute with multi-controlled $P$ gates. This allows us to derive \cref{euler3d} from \cref{Euler3dsansphases-multicontrolled} together with \cref{Euler2dmulticontrolled}. The details are given in~\Cref{preuveeuler3d}.
\end{proof}

\begin{theorem}[Completeness]\label{completenessQC}
  The equational theory $\QC$ is complete for $\propQC$-circuits, i.e.~for any $\propQC$-circuits $C_1,C_2$, if $\interp{C_1}=\interp{C_2}$ then $\QC\vdash C_1=C_2$.
\end{theorem}
\begin{proof}
    Thanks to \Cref{lem:allAxiomExeptEuler3D} it only remains to prove \Cref{euler3d}, which has been done in \Cref{prop:proofeuler3d}.
\end{proof}

\section{Minimality}\label{sec:min}

We say that an axiom of an equational theory is necessary if it cannot be derived from the other axioms. To demonstrate the necessity of an axiom, a common technique is to define an alternative interpretation of circuits which satisfies all the axioms except the one we want to prove necessary. Note that in the framework of vanilla quantum circuits, as all generators are unitary and thus preserve the number of qubits, the derivation of an equation on $n$-qubit circuits only involves axioms acting on $n$ qubits or fewer. Thus, it is sufficient to find an alternative interpretation that is sound for all other axioms of the equational theory acting on at most $n$ qubits (and possibly not sound for axioms acting on strictly more that $n$ qubits). Moreover, such an alternative interpretation should also preserve the monoidal structure of quantum circuits. Technically speaking, it has to be a functor.

We  illustrate with  \Cref{HH} how an alternative interpretation can be used to prove the necessity of an axiom: Let  $\interp{\cdot}_{H^2}$ be inductively defined as follows:
\begin{gather*}
  \interp{C_2\circ C_1}_{H^2}= \interp{C_1\otimes C_2}_{H^2}= max(\interp{C_2}_{H^2},\interp{C_1}_{H^2})\\
  \interp{\gH}_{H^2}= 1; \quad\textup{and } \interp{g}_{H^2}= 0 \textup{ for any other generator }g
\end{gather*}
Note that every axiom of $\QC$ acting on at most $1$ qubit is sound according to this new interpretation except \Cref{HH}. Intuitively, \Cref{HH} is the only rule acting on at most $1$ qubit that transforms a circuit containing some $\gH$ gates into one that contains none. Notice that \Cref{eulerH} and \Cref{euler} contain some $\gH$ gates as $\gRX$ is a shortcut notation that contains some (see \Cref{RXdef}).

For each axiom we introduce an alternative interpretation which demonstrates its necessity (see \Cref{appendix:necessityproofs}). 
We explain here the intuition for each axiom and then provide more details  regarding the  necessity of \Cref{euler} and \Cref{ctrl2pi}.
\begin{itemize}
  \item \Cref{gphaseempty} is the only rule acting on $0$ qubits that transforms a circuit containing some global phase gates into one that contains none.
  \item At least one instance of \Cref{gphaseaddition} with parameter $\psi$ is necessary for each $\psi\in\R\setminus\{2\pi\}$, otherwise there is no equation on $0$ qubits that transforms a circuit containing some $\minitikzfig{gates/spsi}$ into one that contains none. 
  \item \Cref{HH} is the only rule acting on at most $1$ qubit that transforms a circuit containing some $\gH$ gates into one that contains none.
  \item \Cref{P0} is the only rule acting on at most $1$ qubit that does not preserve the parity of the number of $\gH$ and $\gP$ gates.
  \item \Cref{CNOTPCNOT} is the only rule acting on at most $2$ qubits that transforms a circuit containing some $\gCNOT$ gates into one that contains none.
  \item \Cref{bigebre} is the only rule that does not preserve
  the parity of the number of $\gSWAP$ gates.
  \item \Cref{CZ} is the only rule that does not preserve the parity of the number of $\gCNOT$ and $\gSWAP$ gates.
  \item \Cref{eulerH} is the only rule that does not preserve the parity of the number of $\gH$ gates.
\end{itemize}

\Cref{euler} represents, in fact, infinitely many equations. It is then of interest to prove the necessity of infinitely many instances of \Cref{euler}. Moreover, \Cref{euler} is the only equation involving non-linear computations on the angles. In the context of the ZX-calculus \cite{vilmart2018nearoptimal}, the non-linearity was taken as an argument in itself for the necessity of the rule. With the following lemma, which can be seen as an attempt to formalise this argument, we prove that, on the one hand, uncountably many instances of \cref{euler} are needed to have completeness, and on the other hand, adding the linear \Cref{Paddition,XPX} does not reduce the number (in terms of cardinality) of instances of \cref{euler} needed.

\begin{restatable}{lemma}{EulercardinalityofR}\label{lem:EulercardinalityofR}
  Let $A$ be a set of equations consisting of
  \begin{itemize}
  \item any set of instances of \Cref{eulerstar} of cardinality strictly less than $2^{\aleph_0}$ (where $2^{\aleph_0}$ denotes the cardinality of $\R$),
  \item all instances of the other equations of $\QC\setminus E$,
  \item all instances of \Cref{Paddition,XPX}.
  \end{itemize}
  Then there exists an instance of \Cref{euler} which is not a consequence of the equations of~$A$.
\end{restatable}
\begin{proof}
  Intuitively, we consider the alternative interpretation consisting of the sum of all the parameters of the $P$ gates of a circuit. We show by direct calculation that, among all possible instances of \cref{euler}, the difference in this quantity between the LHS and the RHS of \cref{euler} takes a continuum of values, which cannot be recovered from a set of instances of \cref{eulerstar} of cardinality strictly less than that of the continuum. The other equations of $\QC\setminus E$, as well as \cref{Paddition}, only change the alternative interpretation by multiples of $\frac\pi2$, hence the argument still holds in their presence if we take the alternative interpretation modulo $\frac\pi2$. In order to manage \cref{XPX}, we use a slightly more involved version of the alternative interpretation where we put a minus sign to the parameters of some gates. The details are given in~\Cref{appendix:proofEulercardinalityofR}.
\end{proof}

There is also infinitely many instances of \Cref{ctrl2pi}. We then need to prove the necessity for every $n\ge3$. Note that the necessity of all instances of \Cref{ctrl2pi} implies that there is no equational theory for vanilla quantum circuits acting on a bounded number of qubits (we discuss this later).

\begin{definition}[Alternative Interpretation]\label{def:QCinterp}
  For  any $k\in \mathbb N$, for any $\propQC$-circuit $C$, let $\interpD{C}k\in [0,2\pi)$ be  inductively defined as 
  \begin{align*}
    \interpD{C_2\circ C_1}k &= \interpD{C_2}k+\interpD{C_1}k\bmod 2\pi\\
    \interpD{C_1\otimes C_2}k &= \interpD{C_1}{k}+\interpD{C_2}{k}\bmod 2\pi\\
    \interpD{\gs}k &= 2^k\varphi \bmod 2\pi\\
    \interpD{\gH}k &= 2^{k-1}\pi\bmod 2\pi\\
    \interpD{\gP}k &= 2^{k-1}\varphi \bmod 2\pi\\
    \interpD{\!\gCNOT}k &= 2^{k-2}\pi \bmod 2\pi\\
    \interpD{\gempty}k &= 0\\
    \interpD{\gI}k &= 0\\
    \interpD{\!\gSWAP}k &= 2^{k-2}\pi\bmod 2\pi
  \end{align*}
\end{definition}

The alternative interpretation is strongly related to the determinant of the unitary map represented by the circuit, in particular, for any $n$-qubit circuit $C$, $\interpD C n$ is the argument of the determinant of $\interp C$. More generally the following property is satisfied: 

\begin{proposition}\label{prop:det}
  For any $n$-qubit circuit $C$ and any $k\ge n$, \[\det(\interp C)^{2^{k-n}} = e^{i\interpD C k}\]
\end{proposition}
\begin{proof}
  By induction on $C$. The base cases can be checked easily, and the induction step is given by
  \begin{align*}
    \det(\interp{C_2\circ C_1})^{2^{k-n}}
    & = \det(\interp{C_2}\circ \interp{C_1})^{2^{k-n}}\\
    & = \det(\interp{C_2})^{2^{k-n}} \det(\interp {C_1})^{2^{k-n}}\\
    & = e^{i\interpD {C_2} k}e^{i\interpD {C_1} k}\\
    & = e^{i\interpD {C_2\circ C_1} k}\\[0.25cm]
    \det(\interp{C_1 \otimes C_2})^{2^{k-n}} &= \det(\interp{C_1} \otimes \interp{C_2})^{2^{k-n}}\\
    & = (\det(\interp{C_1})^{2^{n_2}} \det(\interp {C_2})^{2^{n_1}})^{2^{k-n}}\\
    &= \det(\interp{C_1})^{2^{k-n_1}} \det(\interp {C_2})^{2^{k-n_2}}\\
    &=  e^{i\interpD {C_1} k}e^{i\interpD {C_2} k}\\&= e^{i\interpD {C_1\otimes  C_2} k}
  \end{align*}
\end{proof}

\begin{corollary}\label{prop:alt_soundness}
  For any $n$-qubit circuits $C_1,C_2$ and any $k\geq n$, if $\interp {C_1} = \interp{C_2}$ then $\interpD{C_1}k=\interpD{C_2}k$. 
\end{corollary}

Notice that \Cref{prop:alt_soundness} does not apply when $k<n$, in particular \[\interpD{\scalebox{0.6}{\begin{tikzpicture}
	\begin{pgfonlayer}{nodelayer}
		\node [style=control] (1) at (0, 0.625) {};
		\node [style=gate] (3) at (0, -0.625) {$P(2\pi)$};
		\node [style=none] (7) at (1.5, 0.625) {};
		\node [style=none] (8) at (-1.5, 0.625) {};
		\node [style=none] (11) at (-1.5, -0.625) {};
		\node [style=none] (12) at (1.5, -0.625) {};
	\end{pgfonlayer}
	\begin{pgfonlayer}{edgelayer}
		\draw (8.center) to (7.center);
		\draw (11.center) to (12.center);
		\draw [style=etc] (1) to (3);
	\end{pgfonlayer}
\end{tikzpicture}}\colon n\to n}{n-1}=~\pi~\neq~ 0 = ~\interpD{\gI^{\otimes n}}{n-1}\]
Hence, it provides an alternative interpretation to demonstrate that,  roughly speaking, completeness requires rules acting on an unbounded number of qubits:

\begin{theorem}\label{thm:unboundedness}
  There is no complete equational theory for $\propQC$-circuits made of equations acting on a bounded number of qubits. More precisely any complete equational theory for $n$-qubit $\propQC$-circuits has at least one rule acting on $n$ qubits. 
\end{theorem}
\begin{proof}
  Any equational theory $A$ (sound with respect to $\interp.$) made of rules acting on at most $n-1$ qubits is sound with respect to $\interpD.{n-1}$, thanks to \Cref{prop:alt_soundness}. However, we have 
  \begin{equation*}
    \interpD{\scalebox{0.6}{\begin{tikzpicture}
	\begin{pgfonlayer}{nodelayer}
		\node [style=control] (1) at (0, 0.625) {};
		\node [style=gate] (3) at (0, -0.625) {$P(2\pi)$};
		\node [style=none] (7) at (1.5, 0.625) {};
		\node [style=none] (8) at (-1.5, 0.625) {};
		\node [style=none] (11) at (-1.5, -0.625) {};
		\node [style=none] (12) at (1.5, -0.625) {};
	\end{pgfonlayer}
	\begin{pgfonlayer}{edgelayer}
		\draw (8.center) to (7.center);
		\draw (11.center) to (12.center);
		\draw [style=etc] (1) to (3);
	\end{pgfonlayer}
\end{tikzpicture}}\colon n\to n}{n-1}=\pi
    \qquad\textup{and}\qquad
    \interpD{\gI^{\otimes n}}{n-1}=0
  \end{equation*}
  so the equation $\scalebox{0.6}{\begin{tikzpicture}
	\begin{pgfonlayer}{nodelayer}
		\node [style=control] (1) at (0, 0.625) {};
		\node [style=gate] (3) at (0, -0.625) {$P(2\pi)$};
		\node [style=none] (7) at (1.5, 0.625) {};
		\node [style=none] (8) at (-1.5, 0.625) {};
		\node [style=none] (11) at (-1.5, -0.625) {};
		\node [style=none] (12) at (1.5, -0.625) {};
	\end{pgfonlayer}
	\begin{pgfonlayer}{edgelayer}
		\draw (8.center) to (7.center);
		\draw (11.center) to (12.center);
		\draw [style=etc] (1) to (3);
	\end{pgfonlayer}
\end{tikzpicture}} = \gI^{\otimes n}$ cannot be derived from $A$. Thus $A$ is not complete. 
\end{proof}

\Cref{thm:unboundedness} demonstrates that any complete equational theory for $\propQC$-circuits acts on an unbounded number of qubits. Note that this central result does not depend on the particular generators used to define $\propQC$-circuits, and applies to any universal gate set made of unitaries acting on a bounded number of qubits:
\begin{corollary}
  Given a universal language for quantum circuits generated by unitary gates acting on a bounded number of qubits, there is no complete equational theory made of equations acting on a bounded number of qubits.
\end{corollary}
\begin{proof}
  Let $\propQC^\star$ be a prop generated by gates acting on a bounded number of qubits, and $\interp{\cdot}^\star: \C^{\{0,1\}^n} \to \C^{\{0,1\}^n}$ its semantics defining a universal language for unitary quantum circuits.
  By universality, for any generator $g$ of $\propQC$ there exists a $\propQC^\star$-circuit $E(g)$ such that $\interp{g}=\interp{E(g)}^\star$. Similarly, for any generator $g$ of $\propQC^\star$ there exists a $\propQC$-circuit $D(g)$ such that $\interp{g}^\star=\interp{D(g)}$. $E$ and $D$ can be straightforwardly extended to prop functors between $\propQC$-circuits and $\propQC^\star$-circuits.
  Suppose for the sake of contradiction that there is a complete equational theory $\QC^\star$ made of equations acting on a bounded number of qubits for $\propQC^\star$-circuits.
  We construct an equational theory $\QC^\bot$ for $\propQC$-circuits in two steps:
  (i) for any generator $g$ of $\propQC$ we add an axiom $D(E(g))=g$, which implies $\QC^\bot\vdash D(E(C))=C$ for any $\propQC$-circuit $C$; and
  (ii) for any axiom $C_1=C_2$ of $\QC^\star$ we add an axiom $D(C_1)=D(C_2)$, which implies $\QC^\bot\vdash D(C_1)=D(C_2)$ whenever $\interp{C_1}^\star=\interp{C_2}^\star$.
  This construction ensures that $\QC^\bot$ is complete for $\propQC$-circuits. Indeed, given two $\propQC$-circuits $C_1,C_2$  such that $\interp{C_1}=\interp{C_2}$ (and so $\interp{E(C_1)}^\star=\interp{E(C_2)}^\star$) we can then derive $C_1$ into $C_2$ using $\QC^\bot\vdash C_1=D(E(C_1))=D(E(C_2))=C_2$.
  Hence, $\QC^\bot$ is a complete equational theory for vanilla quantum circuits made of equations acting on a bounded number of qubits, which contradicts \Cref{thm:unboundedness}.
\end{proof}

\Cref{thm:unboundedness} directly implies the necessity of \Cref{ctrl2pi} for any $n\ge 3$.

Notice that, as \Cref{ctrl2pi} for $n=4$ involves Clifford+T circuits (all phase gates have parameter multiple of $\nicefrac{\pi}{4}$), an interesting consequence of \Cref{prop:alt_soundness} is that a complete equational theory for Clifford+T quantum circuits must contain at least one equation acting on $4$ qubits. This implies that the equational theory for 2-qubit Clifford+T quantum circuits from \cite{bian2022generators} is not complete for Clifford+T quantum circuits acting on 3 or more qubits. More generally, we have the following corollary.
\begin{corollary}
  Any complete equational theory for the fragment of vanilla quantum circuits where all phase gates have parameter multiple of $\nicefrac{\pi}{2^n}$, contains at least one equation acting on $n+2$ qubits.
\end{corollary}
\begin{proof}
  In Equation $\scalebox{0.6}{\begin{tikzpicture}
	\begin{pgfonlayer}{nodelayer}
		\node [style=control] (1) at (0, 0.625) {};
		\node [style=gate] (3) at (0, -0.625) {$P(2\pi)$};
		\node [style=none] (7) at (1.5, 0.625) {};
		\node [style=none] (8) at (-1.5, 0.625) {};
		\node [style=none] (11) at (-1.5, -0.625) {};
		\node [style=none] (12) at (1.5, -0.625) {};
	\end{pgfonlayer}
	\begin{pgfonlayer}{edgelayer}
		\draw (8.center) to (7.center);
		\draw (11.center) to (12.center);
		\draw [style=etc] (1) to (3);
	\end{pgfonlayer}
\end{tikzpicture}} = \gI^{\otimes n+2}$, all phase gates have parameter $\pm\nicefrac{\pi}{2^n}$. This equation is not sound according to $\interpD{\cdot}{n+1}$. Thanks to \Cref{prop:alt_soundness}, all sound equations acting on strictly less than $n+2$ qubits are also sound according to $\interpD{\cdot}{n+1}$. Hence $\scalebox{0.6}{\begin{tikzpicture}
	\begin{pgfonlayer}{nodelayer}
		\node [style=control] (1) at (0, 0.625) {};
		\node [style=gate] (3) at (0, -0.625) {$P(2\pi)$};
		\node [style=none] (7) at (1.5, 0.625) {};
		\node [style=none] (8) at (-1.5, 0.625) {};
		\node [style=none] (11) at (-1.5, -0.625) {};
		\node [style=none] (12) at (1.5, -0.625) {};
	\end{pgfonlayer}
	\begin{pgfonlayer}{edgelayer}
		\draw (8.center) to (7.center);
		\draw (11.center) to (12.center);
		\draw [style=etc] (1) to (3);
	\end{pgfonlayer}
\end{tikzpicture}} = \gI^{\otimes n+2}$ is not derivable using equations acting on strictly less than $n+2$ qubits.
\end{proof}

\begin{theorem}[Minimality]\label{thm:minimalityQC}
  The equational theory $\QC$ is minimal for $\propQC$-circuits, i.e.~none of its equations can be derived from the others.
\end{theorem}
\begin{proof}
  The necessity of \Cref{euler} is a direct consequence of \Cref{lem:EulercardinalityofR}. The necessity of \Cref{ctrl2pi} is stated as \Cref{thm:unboundedness}. The necessity of the remaining axioms is proved in \Cref{appendix:necessityproofs}.
\end{proof}

\section{An alternative minimal equational theory}\label{sec:alt}

In \cite{vilmart2018nearoptimal}, the author introduced in the context of the ZX-calculus a variant of the Euler decomposition rule, which can be stated as \Cref{eulerprime} depicted below in the quantum circuit formalism. He showed,  in the context of the ZX-calculus, that the standard Euler decomposition rule as well as the Euler decomposition of Hadamard can be derived from this new rule. Following this approach, we consider the alternative equational theory $\QCprime$ for vanilla quantum circuits where the two equations
\begin{gather}
  \tikzfigS{./qc-axioms/H}=\tikzfigS{./qc-axioms/eulerH}\tag{E$_{\textup{H}}$}\\
  \tikzfigS{./qc-axioms/euler-left}=\tikzfigS{./qc-axioms/euler-right}\tag{E}
\end{gather}
of $\QC$ are replaced by
\begin{gather}
  \tikzfigS{./qcoriginal-axioms/Pphi1Pphi2}=\tikzfigS{./qcoriginal-axioms/Pphi1phi2}\label{Paddition-prime}\tag{P$_+$}\\
  \tikzfigS{./qcprime-axioms/eulerprimebis-left}=\tikzfigS{./qcprime-axioms/eulerprimebis-right}\label{eulerprime}\tag{E'}
\end{gather}
and show that this also leads to a complete and minimal equational theory for vanilla quantum circuits.

Interestingly, the new version of the Euler-decomposition rule (\Cref{eulerprime}) involves one less parameter, and thus the angle calculation is simpler.

For any angles $\alpha_1',\alpha_3'\in\R$ the angles $\beta_0',\beta_1',\beta_2',\beta_3'$ are restricted to $[0,2\pi)$ and are computed by taking $\alpha_1=\alpha_1'+\frac{\pi}{2}$, $\alpha_2=\frac{\pi}{2}$, $\alpha_3=\alpha_3'+\frac{\pi}{2}$, $\beta_i'=\beta_i$ for $i\in\{1,2,3\}$ and $\beta_0'=\beta_0+\frac{\pi}{4}$ in \Cref{euler}, which leads to
\begin{gather*}
  z\defeq -\sin\left(\frac{\alpha_1'+\alpha_3'}{2}\right)+i\cos\left(\frac{\alpha_1'-\alpha_3'}{2}\right) \\
  z'\defeq \cos\left(\frac{\alpha_1'+\alpha_3'}{2}\right)-i\sin\left(\frac{\alpha_1'-\alpha_3'}{2}\right)
\end{gather*}
\begin{itemize}
  \item If $z'=0$ then $\beta_0'\defeq\frac{\pi}{2}-\arg(z)$, $\beta_1'\defeq2\arg(z)$, $\beta_2'\defeq0$ and $\beta_3'\defeq0$.
  \item If $z=0$ then $\beta_0'\defeq\frac{\pi}{2}-\arg(z')$, $\beta_1'\defeq2\arg(z')$, $\beta_2'\defeq\pi$ and $\beta_3'\defeq0$.
  \item Otherwise $\beta_0'\defeq\frac{\pi}{2}-\arg(z)$, $\beta_1'\defeq\arg(z)+\arg(z')$, $\beta_2'\defeq2\arg\left(i+\left\lvert\frac{z}{z'}\right\rvert\right)$ and $\beta_3'\defeq\arg(z)-\arg(z')$.
\end{itemize}

We show  that the Clifford fragment (restricting the parameters of the phase gates to be multiples of $\frac \pi 2$) is closed under \Cref{eulerprime}, which is not the case for \Cref{euler}.

\begin{restatable}{proposition}{eulerprimeclifford}\label{prop:eulerprimeclifford}
  The Clifford fragment is closed under \Cref{eulerprime}, i.e.~$\alpha'_i=0\pmod{\frac{\pi}{2}}$ for $i\in\{1,3\}$ if and only if $\beta_i'=0\pmod{\frac{\pi}{2}}$ for $i\in\{1,2,3\}$.
\end{restatable}
\begin{proof}
  The proof is given in \Cref{appendix:minimalityQCprime}.
\end{proof}

To prove the completeness of $\QCprime$ for vanilla quantum circuits, we simply derive all equations of $\QC$.

\begin{theorem}\label{thm:completenessQCprime}
  The equational theory $\QCprime$ is complete for $\propQC$-circuits.
\end{theorem}
\begin{proof}
  It is sufficient to prove \Cref{eulerH} and \Cref{euler} in $\QCprime$. The other equations of $\QC$ are also in $\QCprime$. \Cref{eulerH} can be easily proved by taking $\alpha_1'=\alpha_3'=0$ in \Cref{eulerprime}. The proof of \Cref{euler} needs a more precise analysis, which was first done in \cite{vilmart2018nearoptimal}. The details are given in \Cref{appendix:completenessQCprime}.
\end{proof}

\begin{proposition}\label{prop:necessityeulerprime}
  An infinite number of instances of \Cref{eulerprime} are necessary in $\QCprime$.
\end{proposition}
\begin{proof}
Similarly to \Cref{euler} in $\QC$, one needs a set of instances of \Cref{eulerprime} of the cardinality of $\R$ to have completeness. Indeed, this is a consequence of \Cref{lem:EulercardinalityofR}, and of the fact that each instance of \Cref{eulerprime} can be derived from a single instance of \Cref{euler} (by taking $\alpha_1=\alpha_1'+\nicefrac{\pi}{2}$, $\alpha_2=\nicefrac{\pi}{2}$ and $\alpha_3=\alpha_3'+\nicefrac{\pi}{2}$).
\end{proof}

This gives the necessity of \Cref{eulerprime}. All the other axioms can also be proved necessary, leading to the minimality of $\QCprime$.
\begin{theorem}[Minimality]\label{thm:minimalityQCprime}
  The equational theory $\QCprime$ is minimal for $\propQC$-circuits.
\end{theorem}
\begin{proof}
  The necessity arguments for Equations \eqref{gphaseempty}, \eqref{gphaseaddition}, \eqref{HH}, \eqref{CNOTPCNOT}, \eqref{bigebre}, \eqref{CZ} and \eqref{ctrl2pi} in $\QC$ remain true in $\QCprime$. \Cref{P0} is the only axiom acting on at most one qubit that does not preserve the presence of phase gates. \Cref{Paddition-prime} is the only rule of $\QCprime$ that does not remains in the Clifford fragment (see \Cref{prop:eulerprimeclifford}). The necessity of \Cref{eulerprime} is proved in \Cref{prop:necessityeulerprime}.
\end{proof}

\section{Quantum circuits extensions}\label{sec:extensions}

\subsection{Vanilla quantum circuits up to global phases}
The proof of the necessity of axioms acting on an unbounded number of qubits (\Cref{thm:unboundedness}) in $\QC$ is closely related to the determinant of $\interp{C}$, which is sensitive to the global phase. So, one can wonder whether a complete equational theory acting on a bounded number of qubits exists for  quantum circuits up to global phases, a standard model of quantum computing. 
In this section, we answer this question in the negative by strengthening \Cref{thm:unboundedness}.

\begin{definition}[Vanilla quantum circuits up to global phases]
  Let $\propQCugp$ be the prop of quantum circuits up to global phases generated by $\gH:1\to 1$, $\gP:1\to 1$ for any $\varphi\in\R$ and $\gCNOT:2\to 2$.
\end{definition}

One could interpret $\propQCugp$-circuits as unitaries up to global phases by defining equivalence classes. Instead, we equivalently  interpret $\propQCugp$-circuits 
as CPTP maps, which are not sensitive to global phases.
  
\begin{definition}[Semantics]
  For any $n$-qubit $\propQCugp$-circuit $C$, let $\CPTP{C}: \mathcal{M}_{2^n}(\C) \to  \mathcal{M}_{2^n}(\C)$ be the \emph{semantics} of $C$ inductively defined as the linear map $\CPTP{C_2\circ C_1} = \CPTP{C_2}\circ\CPTP{C_1}$; $\CPTP{C_1\otimes C_2} = \CPTP{C_1}\otimes\CPTP{C_2}$; and $\CPTP{g}=\rho\mapsto\interp{g}\rho\interp{g}^\dagger$ for every generator $g$, where $M^\dagger$ is the adjoint of the matrix $M$.
\end{definition}

We consider the equational theory $\QCugp$, depicted in \Cref{fig:qcsf-axioms}, obtained  from 
$\QC$ by applying $(\cdot)^{\textup{gp-free}}$ which consists of removing the phase gates ($\gs$), from every circuit involved. Notice that \Cref{gphaseempty,gphaseaddition} are then trivialised. With a slight abuse of notation we use $\gRX$ to denote  $(\gRX)^{\textup{gp-free}}=\scalebox{0.8}{\tikzfigM{./shortcut/HPthetaH-}}$. The equational theory $\QCugp$ is sound with respect to $\CPTP{\cdot}$. 

\begin{corollary}[Completeness]
  The equational theory $\QCugp$ is complete for $\propQCugp$-circuits.
\end{corollary}
\begin{proof}
  Let $C_1,C_2$ be $\propQCugp$-circuits such that $\CPTP{C_1}=\CPTP{C_2}$. Then, there exists $\varphi\in\R$ such that $\interp{C_1}=\interp{\gs\otimes C_2}$, considering $C_1$ and $C_2$ as $\propQC$-circuits. By completeness of $\QC$ we have a derivation of $C_1=\gs\otimes C_2$ in $\QC$. We turn it into a derivation of $C_1=C_2$  in $\QCugp$ by replacing every application of a rule of $\QC$ by its corresponding one in $\QCugp$ and by removing the trivialised applications of Equations \eqref{gphaseempty} and \eqref{gphaseaddition}.
\end{proof}

\begin{corollary}[Minimality]
  The equational theory $\QCugp$ is minimal for $\propQCugp$-circuits.
\end{corollary}
\begin{proof}
  Suppose for the sake of contradiction that there is an axiom $(a)$ of the form $C_1=C_2$ in $\QC\setminus \{\eqref{gphaseempty}, \eqref{gphaseaddition}\} $ s.t. the corresponding axiom $(a_{\textup{gp-free}})$ $C_1^{\textup{gp-free}} = C_2^{\textup{gp-free}}$ is derivable in $\QCugp\setminus (a_{\textup{gp-free}})$. We construct from this derivation a derivation of $C_1=C_2$ in $\QC\setminus (a)$, by applying the corresponding axioms when it is possible. Otherwise we introduce, thanks to \eqref{gphaseempty} and \eqref{gphaseaddition}, the missing phases of the form $\gs$ together with their dual \scalebox{0.7}{\begin{tikzpicture}
	\begin{pgfonlayer}{nodelayer}
		\node [style=globalphase] (0) at (0, 0) {-$\varphi$};
	\end{pgfonlayer}
\end{tikzpicture}
}. There might be some extra phases in the final circuit that can be removed by completeness of  \eqref{gphaseempty} and \eqref{gphaseaddition} for $0$-qubit circuits. 
\end{proof}

\begin{proposition}
  There is no complete equational theory for quantum circuits up to global phases made of equations acting on a bounded number of qubits.
\end{proposition}
\begin{proof}
  Suppose for the sake of contradiction that there is a complete equational theory $\QCugp^\bot$ for $\propQCugp$-circuits made of equations acting on a bounded number of qubits. Then we could construct one for $\propQC$-circuits. The construction goes as follows: for each axiom $C_1=C_2$ of $\QCugp^\bot$, make the equation sound for $\interp{\cdot}$ by adding a global phase $\gs$ to $C_2$, leading to axiom $C_1=\gs\otimes C_2$. Then we add Equations \eqref{gphaseempty} and \eqref{gphaseaddition}, which are complete for $\propQC$-circuits acting on zero qubit. This would lead to a complete equational theory for $\propQC$-circuits made of equations acting on a bounded number of qubits, which is impossible according to \Cref{thm:unboundedness}.
\end{proof}

\subsection{Quantum circuits with ancillae}

\begin{figure*}
  \fbox{\begin{minipage}{.985\textwidth}\begin{center}
    \vspace{-.5em}
    \hspace{-1.5em}\begin{subfigure}{0.18\textwidth}
      \begin{align}\label{gphaseempty-ancilla}\tag{S$_{2\pi}$}\tikzfigM{./qc-axioms/s2pi}=\tikzfigM{./qc-axioms/empty}\end{align}
    \end{subfigure}\hspace{1.2em}
    \begin{subfigure}{0.22\textwidth}
      \begin{align}\label{HH-ancilla}\tag{H$^2$}\tikzfigM{./qc-axioms/HH}=\tikzfigM{./qc-axioms/Id}\end{align}
    \end{subfigure}\hspace{1.2em}
    \begin{subfigure}{0.22\textwidth}
      \begin{align}\label{initP-ancilla}\tag{AP}\tikzfigM{./qcancilla-axioms/initPphi}=\tikzfigM{./qcancilla-axioms/initId}\end{align}
    \end{subfigure}\hspace{1.2em}
    \begin{subfigure}{0.18\textwidth}
      \begin{align}\label{initdest-ancilla}\tag{A}\tikzfigM{./qcancilla-axioms/initdest}=\tikzfigM{./qcancilla-axioms/empty}\end{align}
    \end{subfigure}

    \hspace{-1.5em}\begin{subfigure}{0.23\textwidth}
      \begin{align}\label{initCNOT-ancilla}\tag{ACX}\tikzfigM{./qcancilla-axioms/initCNOT}=\tikzfigM{./qcancilla-axioms/initIdId}\end{align}
    \end{subfigure}\hspace{5em}
    \begin{subfigure}{0.28\textwidth}
      \begin{align}\label{5CX-ancilla}\tag{5CX}\tikzfigM{./qcancilla-axioms/3CNOT-left}=\tikzfigM{./qcancilla-axioms/3CNOT-right}\end{align}
    \end{subfigure}

    \hspace{-1.5em}\begin{subfigure}{0.28\textwidth}
      \begin{align}\label{CNOTPCNOT-ancilla}\tag{C}\tikzfigM{./qc-axioms/CNOTPphiCNOT}=\tikzfigM{./qc-axioms/PphiId}\end{align}
    \end{subfigure}\hspace{1.2em}
    \begin{subfigure}{0.26\textwidth}
      \begin{align}\label{bigebre-ancilla}\tag{B}\tikzfigM{./qc-axioms/CNOTNOTC}=\tikzfigM{./qc-axioms/SWAPCNOT}\end{align}
    \end{subfigure}\hspace{1.2em}
    \begin{subfigure}{0.39\textwidth}
      \begin{align}\label{CZ-ancilla}\tag{CZ}\tikzfigM{./qc-axioms/H2CNOTH2}=\tikzfigM{./qc-axioms/CZ}\end{align}
    \end{subfigure}
    \vspace{.5em}
  \end{center}\end{minipage}}
  \vspace{.5em}

  \ovalbox{\begin{minipage}{.985\textwidth}\begin{center}
    \vspace{-.5em}
    \hspace{-1.5em}\begin{subfigure}{0.21\textwidth}
      \begin{align}\label{P0-ancilla}\tag{P$_0$}\tikzfigM{./qc-axioms/P0}=\tikzfigM{./qc-axioms/Id}\end{align}
    \end{subfigure}\hspace{5em}
    \begin{subfigure}{0.38\textwidth}\begin{align}\label{eulerH-ancilla}\tag{E$_{\textup{H}}$}\tikzfigM{./qc-axioms/H}=\tikzfigM{./qc-axioms/eulerH}\end{align}\end{subfigure}
    \hspace{-1.5em}\begin{subfigure}{0.60\textwidth}
      \begin{align}\label{euler-ancilla}\tag{E}\tikzfigM{./qc-axioms/euler-left}=\tikzfigM{./qc-axioms/euler-right}\end{align}
    \end{subfigure}
    \vspace{.5em}
  \end{center}\end{minipage}}
  \vspace{.5em}

  \ovalbox{\begin{minipage}{.985\textwidth}\begin{center}
    \vspace{-.5em}
    \hspace{-1.5em}\begin{subfigure}{0.25\textwidth}
        \begin{align}\label{gphaseaddition-ancilla}\tag{S$_+$}\tikzfigM{./qc-axioms/sphi1}\tikzfigM{./qc-axioms/sphi2}=\tikzfigM{./qc-axioms/sphi1plusphi2}\end{align}
      \end{subfigure}\hspace{5em}
    \begin{subfigure}{0.37\textwidth}\begin{align*}\label{Paddition-ancilla}\tag{P$_+$}\tikzfigM{./qcoriginal-axioms/Pphi1Pphi2}=\tikzfigM{./qcoriginal-axioms/Pphi1phi2}\end{align*}\end{subfigure}
    \hspace{-1.5em}\begin{subfigure}{0.58\textwidth}
      \begin{align}\tag{E'}\tikzfigM{./qcprime-axioms/eulerprimebis-left}=\tikzfigM{./qcprime-axioms/eulerprimebis-right}\label{euler-prime-ancilla}\end{align}
    \end{subfigure}
    \vspace{.5em}
  \end{center}\end{minipage}}
  \caption{\normalfont The equational theory $\QCancilla$ (resp. $\QCancillaprime$) is depicted by the top box together with the first (resp. second) rounded box.\label{fig:qcancilla}}
\end{figure*}

In this section we consider quantum circuits implementing isometries\footnote{An isometry is a linear map $V$ s.t. $V^\dagger \circ V$ is the identity.} using ancillary qubits, a.k.a. ancillae, as additional work space. This extension allows one to initialise fresh qubits in the $\ket{0}$-state and to release qubits when they become useless. Note that to guarantee that the overall evolution is an isometry, one can only release a qubit in the $\ket{0}$-state.

Following \cite{CDPV}, we define quantum circuits with ancillae as the circuits generated by $\gH$, $\gP$, $\gCNOT$ and $\gs$ together with two new generators $\ginit$ and $\gdest$, respectively denoting qubit initialisation and qubit termination. Because of the constraint that removed qubits must be in the $\ket{0}$-state, we define the language of quantum circuits with ancillae in two steps.

\begin{definition}
  Let $\propQCpreancilla$ be the prop generated by $\gs:0\to 0$, $\gH:1\to 1$, $\gP:1\to 1$, $\gCNOT:2\to 2$, $\ginit:0\to 1$ and $\gdest:1\to 0$ for any $\varphi\in\R$.
\end{definition}

\begin{definition}[Semantics]
The semantics $\interp \cdot $ of vanilla quantum circuits (see \Cref{def:QCsem}) is extended with $\interp \ginit = \ket 0$ and $\interp \gdest = \bra 0$.
\end{definition}

Notice that the semantics of a $\propQCpreancilla$-circuit is not necessarily an isometry as $\interp \gdest$ is not isometric. As a consequence, $\propQCancilla$ is defined as the subclass of $\propQCpreancilla$-circuits with an isometric semantics.

\begin{definition}[Quantum circuits with ancillae]
  Let $\propQCancilla$ be the sub-prop of $\propQCpreancilla$-circuits $C$ such that $\interp C$ is an isometry. 
\end{definition}

The three following sound equations reveal the behaviour of the new generators.
\begin{align}
  \tag{A}\tikzfigS{./qcancilla-axioms/initdest}&=\tikzfigS{./qcancilla-axioms/empty}\\[0cm]
  \tag{AP}\tikzfigS{./qcancilla-axioms/initPphi}&=\tikzfigS{./qcancilla-axioms/initId}\\[0cm]
  \tag{ACX}\tikzfigS{./qcancilla-axioms/initCNOT}&=\tikzfigS{./qcancilla-axioms/initIdId}
\end{align}

It as been shown in \cite{CDPV} that the addition of those three equations is sufficient to get the completeness for quantum circuits with ancillae.
\begin{theorem}[Completeness \cite{CDPV}]
  Adding Equations \eqref{initdest-ancilla}, \eqref{initP-ancilla} and \eqref{initCNOT-ancilla} to a complete equational theory for vanilla quantum circuits leads to a complete equational theory for quantum circuits with ancillae.
\end{theorem}

It turns out that all instances of \Cref{ctrl2pi} for $n>3$ are derivable from the other equations within this new model. This result has been proved in \cite{CDPV}\footnote{More precisely, the inductive proof of \cref{euler3d} from its $3$-qubit case, given in Appendix E.3 of \cite{CDPV}, can be trivially adapted into a proof of \cref{ctrl2pi}.} and is a consequence of the fact that, using ancillae, multi-controlled operations can be provably\footnote{The proof of such equations is done using $\QCold$. However, the proof also holds using $\QC_3$ as \Cref{euler3d} is not used and all the other equations of $\QCold$ are provable in $\QC_3$.} expressed as follows.
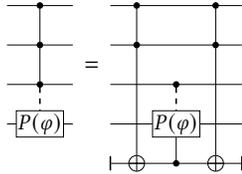
\begin{equation*}\label{Paltdef}
  \tikzfigS{./qcancilla-completeness/initmctrlP}=\tikzfigS{./qcancilla-completeness/initmctrlPdef}
\end{equation*}

Moreover the remaining instance of \Cref{ctrl2pi} acting on $n=3$ qubits can be replaced by the following (more elementary) equation (cf. \Cref{appendix:proof5CXancilla}).
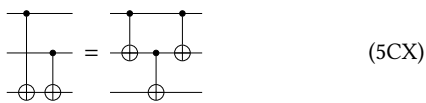
\begin{equation}\tag{5CX}
  \tikzfigS{./qcancilla-axioms/3CNOT-left}=\tikzfigS{./qcancilla-axioms/3CNOT-right}
\end{equation}

This leads to the complete equational theory $\QCancilla$ (resp. $\QCancillaprime$) for quantum circuits with ancillae depicted in \cref{fig:qcancilla}. It is remarkable that, in the presence of ancillae, all equations act on a bounded number of qubits while in the case of vanilla quantum circuits, an equation acting on $n$ qubits is necessary for any $n\in\N$ (see \Cref{thm:unboundedness}).

It turns out that, in these more expressive settings, most of the necessity arguments given for $\QC$ (cf. \Cref{thm:minimalityQC}) do not hold anymore, especially because of the fact that one can create new wires using \Cref{initdest-ancilla} and then apply equations acting on more qubits. In fact, some rules of the equational theories $\QCancilla$ and $\QCancillaprime$ can even be derived from the others. We leave open the question of whether $\QCancilla$ and $\QCancillaprime$ are minimal or not after these additional simplifications.

\begin{proposition}\label{prop:proofP0QCancillaprime}
  \Cref{P0-ancilla} can be derived from the other equations of $\QCancillaprime$.
\end{proposition}
\begin{proof}
  \begin{gather*}
    \tikzfigS{./qc-axioms/Id}
    \eqeqref{HH-ancilla}\tikzfigS{./qc-axioms/HH}
    \eqeqref{initdest-ancilla}\tikzfigS{./identitiesancilla/initdestsurHH}
    \eqeqref{initCNOT-ancilla}\tikzfigS{./identitiesancilla/initH2CNOTH2dest}\\[0.15cm]
    \eqeqref{CZ-ancilla}\tikzfigS{./identitiesancilla/initCZdest}
    \eqdeuxeqref{initP-ancilla}{initCNOT-ancilla}\tikzfigS{./identitiesancilla/initdestsurPpi2Pmoinspi2}\\[0.15cm]
    \eqeqref{initdest-ancilla}\tikzfigS{./identitiesancilla/Ppi2Pmoinspi2}
    \eqeqref{Paddition-ancilla}\tikzfigS{./qc-axioms/P0}
  \end{gather*}
\end{proof}

\begin{proposition}\label{prop:proofgphaseadditionQCancilla}
  \Cref{gphaseaddition-ancilla} can be derived from the other equations of $\QCancilla$.
\end{proposition}
\begin{proof}
  \begin{gather*}
    \tikzfigS{./identitiesancilla/sphi1sphi2}
    \eqeqref{initdest-ancilla}\tikzfigS{./identitiesancilla/sphi1sphi2surinitdest}
    \eqeqref{initP-ancilla}\tikzfigS{./identitiesancilla/sphi1sphi2sum2}\\[0.15cm]
    \eqdeuxeqref{HH-ancilla}{P0-ancilla}\tikzfigS{./identitiesancilla/sphi1sphi2sum3}\\[0.4cm]
    \eqeqref{RXdef}\tikzfigS{./identitiesancilla/sphi1sphi2sum4}\\[0.15cm]
    \eqeqref{euler-ancilla}\tikzfigS{./identitiesancilla/sphi1sphi2sum5}\\[0.15cm]
    \eqeqref{euler-ancilla}\tikzfigS{./identitiesancilla/sphi1sphi2sum6}\\[0.15cm]
    \eqeqref{RXdef}\tikzfigS{./identitiesancilla/sphi1sphi2sum7}\\[0.15cm]
    \eqdeuxeqref{HH-ancilla}{P0-ancilla}\tikzfigS{./identitiesancilla/sphi1sphi2sum8}
    \eqdeuxeqref{initP-ancilla}{initdest-ancilla}\tikzfigS{./identitiesancilla/sphi1sphi2sum9}
    \eqeqref{gphaseempty-ancilla}\tikzfigS{./qc-axioms/sphi1plusphi2}
  \end{gather*}
\end{proof}

\subsection{Quantum circuits with discard}

\begin{figure*}
  \fbox{\begin{minipage}{.985\textwidth}\begin{center}
    \vspace{-.5em}
    \hspace{-1.5em}\begin{subfigure}{0.22\textwidth}
      \begin{align}\label{HH-discard}\tag{H$^2$}\tikzfigM{./qc-axioms/HH}=\tikzfigM{./qc-axioms/Id}\end{align}
    \end{subfigure}\hspace{3em}
    \begin{subfigure}{0.22\textwidth}
      \begin{align}\label{initP-discard}\tag{AP}\tikzfigM{./qcancilla-axioms/initPphi}=\tikzfigM{./qcancilla-axioms/initId}\end{align}
    \end{subfigure}\hspace{3em}
    \begin{subfigure}{0.18\textwidth}
      \begin{align}\label{initdiscard-discard}\tag{D}\tikzfigM{./qcground-axioms/initdiscard}=\tikzfigM{./gates/empty}\end{align}
    \end{subfigure}
    
    \hspace{-1.5em}\begin{subfigure}{0.22\textwidth}
      \begin{align}\label{Hdiscard-discard}\tag{HD}\tikzfigM{./qcground-axioms/Hdiscard}=\tikzfigM{./qcground-axioms/Iddiscard}\end{align}
    \end{subfigure}\hspace{5em}
    \begin{subfigure}{0.22\textwidth}
      \begin{align}\label{Pdiscard-discard}\tag{PD}\tikzfigM{./qcground-axioms/Pphidiscard}=\tikzfigM{./qcground-axioms/Iddiscard}\end{align}
    \end{subfigure}\

    \hspace{-1.5em}\begin{subfigure}{0.23\textwidth}
      \begin{align}\label{initCNOT-discard}\tag{ACX}\tikzfigM{./qcancilla-axioms/initCNOT}=\tikzfigM{./qcancilla-axioms/initIdId}\end{align}
    \end{subfigure}\hspace{1.2em}
    \begin{subfigure}{0.23\textwidth}
      \begin{align}\label{CNOTdiscard-discard}\tag{CXD}\tikzfigM{./qcground-axioms/CNOTdiscard}=\tikzfigM{./qcground-axioms/IdIddiscarddiscard}\end{align}
    \end{subfigure}\hspace{1.2em}
    \begin{subfigure}{0.28\textwidth}
      \begin{align}\label{5CX-discard}\tag{5CX}\tikzfigM{./qcancilla-axioms/3CNOT-left}=\tikzfigM{./qcancilla-axioms/3CNOT-right}\end{align}
    \end{subfigure}

    \hspace{-1.5em}\begin{subfigure}{0.28\textwidth}
      \begin{align}\label{CNOTPCNOT-discard}\tag{C}\tikzfigM{./qc-axioms/CNOTPphiCNOT}=\tikzfigM{./qc-axioms/PphiId}\end{align}
    \end{subfigure}\hspace{1.2em}
    \begin{subfigure}{0.26\textwidth}
      \begin{align}\label{bigebre-discard}\tag{B}\tikzfigM{./qc-axioms/CNOTNOTC}=\tikzfigM{./qc-axioms/SWAPCNOT}\end{align}
    \end{subfigure}\hspace{1.2em}
    \begin{subfigure}{0.39\textwidth}
      \begin{align}\label{CZ-discard}\tag{CZ}\tikzfigM{./qc-axioms/H2CNOTH2}=\tikzfigM{./qc-axioms/CZ}\end{align}
    \end{subfigure}
    \vspace{.5em}
  \end{center}\end{minipage}}
  \vspace{.5em}

  \ovalbox{\begin{minipage}{.985\textwidth}\begin{center}
    \vspace{-.5em}
    \hspace{-1.5em}\begin{subfigure}{0.21\textwidth}
      \begin{align}\label{P0-discard}\tag{P$_0$}\tikzfigM{./qc-axioms/P0}=\tikzfigM{./qc-axioms/Id}\end{align}
    \end{subfigure}\hspace{5em}
    \begin{subfigure}{0.37\textwidth}\begin{align}\label{eulerH-discard}\tag{E$_{\textup{H}}$}\tikzfigM{./qc-axioms/H}=\tikzfigM{./qc-axioms/eulerH}\end{align}\end{subfigure}
    \hspace{-1.5em}\begin{subfigure}{0.58\textwidth}
      \begin{align}\label{euler-discard}\tag{Euler}\tikzfigM{./qcground-axioms/euler-left}=\tikzfigM{./qcground-axioms/euler-right}\end{align}
    \end{subfigure}
    \vspace{.5em}
  \end{center}\end{minipage}}
  \vspace{.5em}

  \ovalbox{\begin{minipage}{.985\textwidth}\begin{center}
    \vspace{-.5em}
    \hspace{-1.5em}\begin{subfigure}{0.38\textwidth}\begin{align}\label{Paddition-discard}\tag{P$_+$}\tikzfigM{./qcoriginal-axioms/Pphi1Pphi2}=\tikzfigM{./qcoriginal-axioms/Pphi1phi2}\end{align}\end{subfigure}\hspace{1.2em}
    \begin{subfigure}{0.54\textwidth}
      \begin{align}\label{euler-prime-discard}\tag{E'}\tikzfigM{./qcground-axioms/eulerprimebis-left}=\tikzfigM{./qcground-axioms/eulerprimebis-right}\end{align}
    \end{subfigure}
    \vspace{.5em}
  \end{center}\end{minipage}}
  \caption{\normalfont The equational theory $\QCground$ (resp. $\QCgroundprime$) is depicted by the top box together with the first (resp. second) rounded box.\label{fig:qcdiscard}}
\end{figure*}

In this section, we consider quantum circuits with discard, which are universal for CPTP maps. This extension enables one to trace out qubits. Contrary to quantum circuits with ancillae, any qubit can be discarded whatever its state is.

Following \cite{CDPV}, we define quantum circuits with discard as the circuits generated by $\gH$, $\gP$ and $\gCNOT$ together with two additional generators $\ginit$ and $\gdiscard$, respectively denoting qubit initialisation and qubit discarding.
\begin{definition}[Quantum circuits with discard]
  Let $\propQCground$ be the prop of quantum circuits with discard generated by ${\gH:1\to 1}$, $\gP:1\to 1$, $\gCNOT:2\to 2$, $\ginit:0\to 1$ and $\gdiscard:1\to 0$ for any $\varphi\in\R$.
\end{definition}

The ability to discard qubits implies that the evolution represented by such a circuit is not pure anymore. As a consequence, the semantics is a completely positive trace-preserving (CPTP) map acting on density matrices (trace-1 positive semi-definite Hermitian matrices). Formally, the new semantics is defined as follows.
\begin{definition}[Semantics]
  For any quantum $\propQCground$-circuit $C:n\to m$, let $\CPTP{C}: \mathcal M_{2^n,2^n}(\C) \to  \mathcal M_{2^m,2^m}(\C) $ be the \emph{semantics} of $C$ inductively defined as the linear map $\CPTP{C_2\circ C_1} = \CPTP{C_2}\circ\CPTP{C_1}$; $\CPTP{C_1\otimes C_2} = \CPTP{C_1}\otimes\CPTP{C_2}$; $\CPTP{\gdiscard} = \rho \mapsto tr(\rho)$ and for any other generator $g$, $\CPTP g = \rho \mapsto \interp g \rho \interp g^\dagger$, where $tr(M)$ is the trace of the matrix $M$ and $M^\dagger$ its adjoint. 
\end{definition}

Notice that the global phase generator $\gs$ is not part of the prop anymore. If it were, its interpretation would be $\CPTP{\gs} = \rho \mapsto \interp{\gs} \rho \interp{\gs}^\dagger = e^{i\varphi}\rho e^{-i\varphi}=\rho$, which is the same as that of the empty circuit. Thus, similarly to quantum circuits up to global phases, the X-rotation is defined as $\tikzfigS{./gates/RXtheta}\defeq\tikzfigS{./shortcut/HPthetaH-nogphase}$ (the same definition as \cref{fig:shortcutcircuits} but without the global phase).

There exist constructions \cite{Huot2019universal,Carette2021completeness} to turn any complete equational theory for quantum circuits into a complete equational theory for circuits with discard (see \cite{CDPV} for details).

Concretely, we consider the six following sound equations that  capture the behaviour
of the new generators:
\begin{align}
  \tag{D}\tikzfigS{./qcground-axioms/initdiscard}&=\tikzfigS{./gates/empty}\\[0cm]
  \tag{AP}\tikzfigS{./qcground-axioms/initPphi}&=\tikzfigS{./qcground-axioms/initId}\\[0cm]
  \tag{ACX}\tikzfigS{./qciso-axioms/initCNOT}&=\tikzfigS{./qciso-axioms/initIdId}\\[0cm]
  \tag{HD}\tikzfigS{./qcground-axioms/Hdiscard}&=\tikzfigS{./qcground-axioms/Iddiscard}\\[0cm]
  \tag{PD}\tikzfigS{./qcground-axioms/Pphidiscard}&=\tikzfigS{./qcground-axioms/Iddiscard}\\[0cm]
  \tag{CXD}\tikzfigS{./qcground-axioms/CNOTdiscard}&=\tikzfigS{./qcground-axioms/IdIddiscarddiscard}
\end{align}

The addition of those six equations is sufficient to get the completeness for quantum circuits with discard:
\begin{theorem}[Completeness \cite{CDPV}]
  Adding Equations \eqref{initdiscard-discard}, \eqref{initP-discard}, \eqref{initCNOT-discard}, \eqref{Hdiscard-discard}, \eqref{Pdiscard-discard} and \eqref{CNOTdiscard-discard} to a complete equational theory for vanilla quantum circuits up to global phases leads to a complete equational theory for quantum circuits with discard. 
\end{theorem}

The derivations done with $\propQCancilla$-circuits are also true for $\propQCground$-circuits (just replace the instances of \Cref{initdest-ancilla} by \Cref{initdiscard-discard}). Thus, \Cref{ctrl2pi} can be completely removed if we add \Cref{5CX-discard}. This leads to the complete equational theory $\QCground$ (resp. $\QCgroundprime$) for quantum circuits with discard depicted in \cref{fig:qcdiscard}. Again, in these more expressive settings, most of the necessity arguments given for $\QC$ (cf. \Cref{thm:minimalityQC}) do not hold anymore and some additional rules can be derived from the others (cf. \Cref{prop:proofgphaseadditionQCancilla} and \Cref{prop:proofP0QCancillaprime}). The question of minimality is similarly still open.

\section*{Acknowledgments}
The authors want to thank Nathan Claudet and Renaud Vilmart for fruitful discussions. 
This work is supported by the the \emph{Plan France 2030} through the PEPR integrated project EPiQ ANR-22-PETQ-0007 and the HQI platform ANR-22-PNCQ-0002;  and by the European projects NEASQC and HPCQS.

\bibliography{ref}

\onecolumn
\appendix

\crefalias{section}{appendix}
\crefalias{subsection}{appendix}
\crefalias{subsubsection}{appendix}
\crefalias{paragraph}{appendix}
\crefalias{subparagraph}{appendix}

\section*{Appendix content}

\startcontents[sections]
\printcontents[sections]{l}{1}{\setcounter{tocdepth}{2}}

\section{Useful intermediate circuit equations}\label{appendix:identities}

\subsection{Proof of \cref{prop:simplederivableequations}}\label{appendix:proofsimplederivableequations}
\simplederivableequations*
\begin{proof} 
  To prove \Cref{P2pi}, we first introduce a $R_X(0)$ gate on both sides of the $P(2\pi)$ gate, we then use \eqref{euler} leading to a circuit  where all parameters are $0$,  equivalent to the identity: \begin{eqnarray*}
    \tikzfigS{./qc-completeness/P2pi_00}&
    \eqquatreeqref{HH}{P0}{gphaseempty}{gphaseaddition}&\tikzfigS{./qc-completeness/P2pi_01}\\
    &\eqeqref{euler}&\tikzfigS{./qc-completeness/P2pi_02}\\
    &\eqquatreeqref{HH}{P0}{gphaseempty}{gphaseaddition}&\tikzfigS{./qc-completeness/P2pi_03}
  \end{eqnarray*}
  To prove \Cref{Paddition},  we first introduce 3 pairs of Hadamard gates, a $P(0)$ gate and the appropriate global phases, to get a circuit $C_1$ such  that  \eqref{euler} can be applied with $\alpha_1=\varphi_1$, $\alpha_2=0$ and $\alpha_3= \varphi_2$. We obtain a circuit $C_2$ with some angles $\beta_i$ and notice that since \eqref{euler} has been applied with $\alpha_2=0$ the angles $\beta_i$ only depend on $\varphi_1+\varphi_2$. Thus  \eqref{euler}  also leads to $C_2$ when applied on a variant of $C_1$ where $\varphi_1$ is replaced by $\varphi_1+\varphi_2$, and $\varphi_2$ is replaced by $0$, which essentially terminates the proof: 
  \begin{eqnarray*}
    \!\tikzfigS{./qc-completeness/Paddition_00}\!\!\!\!\!\!\!\!\!\!\!\!
    &\eqeqref{HH}&\!\!\!\tikzfigS{./qc-completeness/Paddition_01}\\
    &\eqtroiseqref{RXdef}{gphaseempty}{gphaseaddition}&\!\!\!\tikzfigS{./qc-completeness/Paddition_02}\\
    &\eqeqref{P0}&\!\!\!\tikzfigS{./qc-completeness/Paddition_03}\\
    &\eqeqref{euler}&\!\!\!\tikzfigS{./qc-completeness/Paddition_04}\\
    &\eqeqref{euler}&\!\!\!\tikzfigS{./qc-completeness/Paddition_05}\\
    &\eqtroiseqref{HH}{P0}{gphaseaddition}&\!\!\!\tikzfigS{./qc-completeness/Paddition_06}\\
    &\eqtroiseqref{RXdef}{gphaseempty}{gphaseaddition}&\!\!\!\tikzfigS{./qc-completeness/Paddition_07}
  \end{eqnarray*}
  Similarly, we can prove \Cref{XPX} as follows
  \begin{eqnarray*}
    \!\tikzfigS{./qc-completeness/XPX_00}\!\!\!\!\!\!\!\!
    &\eqdeuxeqref{Xdef}{Zdef}&\!\!\!\!\!\tikzfigS{./qc-completeness/XPX_01}\\
    &\eqdeuxeqref{Paddition}{P2pi}&\!\!\!\!\!\tikzfigS{./qc-completeness/XPX_015}\\
    &\eqtroiseqref{RXdef}{gphaseempty}{gphaseaddition}&\!\!\!\!\!\tikzfigS{./qc-completeness/XPX_02}\\[.5em]
    &\eqeqref{euler}&\!\!\!\!\!\tikzfigS{./qc-completeness/XPX_03}\\[.5em]
    &\eqquatreeqref{P2pi}{Paddition}{gphaseempty}{gphaseaddition}&\!\!\!\!\!\tikzfigS{./qc-completeness/XPX_04}\\
    &\eqquatreeqref{HH}{P0}{gphaseaddition}{gphaseempty}&\!\!\!\!\!\tikzfigS{./qc-completeness/XPX_05}
  \end{eqnarray*}
  where the application of \Cref{euler} leads to $\beta_0=\varphi-\pi\bmod{2\pi}$ and $\beta_1=-\varphi\bmod{2\pi}$ because of the restriction $\beta_0,\beta_1\in[0,2\pi)$. Fortunately, this is not an issue as we can then use  \Cref{gphaseempty} and \Cref{gphaseaddition} for $\beta_0$, and \Cref{P2pi} and \Cref{Paddition} for $\beta_1$.
\end{proof}

\subsection{Proof of \cref{lem:allAxiomExeptEuler3D}}\label{appendix:proofallAxiomExeptEuler3D}
\allAxiomExeptEulerthreeD*
\begin{proof}
The proof that \Cref{eulerconditioned} can be replaced by \Cref{euler} is given in \Cref{sec:discussionEuler}. The proof of \Cref{S0,bigebrebis} is given below.

\begin{proof}[Proof of \Cref{S0}]
\begin{gather*}
  \tikzfigS{./qcoriginal-axioms/s0}
  \eqeqref{gphaseempty}\tikzfigS{./identities/s0s2pi}
  \eqeqref{gphaseaddition}\tikzfigS{./qc-axioms/s2pi}
  \eqeqref{gphaseempty}\tikzfigS{./qc-axioms/empty}
\end{gather*}
\end{proof}

\begin{proof}[Proof of \Cref{bigebrebis}]
\begin{gather*}
  \tikzfigS{./qcoriginal-axioms/SWAP_00}
  \eqeqref{bigebre}\tikzfigS{./qcoriginal-axioms/SWAP_01}
  \eqeqref{P0}\tikzfigS{./qcoriginal-axioms/SWAP_02}
  \eqeqref{CNOTPCNOT}\tikzfigS{./qcoriginal-axioms/SWAP_03}
  \eqeqref{P0}\tikzfigS{./qcoriginal-axioms/SWAP_04}
\end{gather*}
\end{proof}

The proof of \Cref{3CNOTstarget} is given below together with all needed intermediate equations.

\begin{multicols}{3}
  \noindent\begin{align}\label{CNOTCNOT}\tikzfigS{./qcoriginal-axioms/CNOTCNOT_00}=\tikzfigS{./qcoriginal-axioms/CNOTCNOT_03}\end{align}
  \begin{align}\label{PcommutCNOT}\tikzfigS{./qcoriginal-axioms/PcommutCNOT_00}=\tikzfigS{./qcoriginal-axioms/PcommutCNOT_02}\end{align}
  \begin{align}\label{Pphasegadget}\tikzfigS{./identities/Pphasegadget-step-0}=\tikzfigS{./identities/Pphasegadget-step-3}\end{align}
  \begin{align}\label{HHCNOTHH}\tikzfigS{./identities/HHCNOTHH_00}=\tikzfigS{./identities/HHCNOTHH_06}\end{align}
  \begin{align}\label{ctrlPminuspi}\tikzfigS{./identities/ctrlPminuspi_00}=\tikzfigS{./identities/ctrlPminuspi_04}\end{align}
\end{multicols}

\begin{proof}[Proof of \cref{CNOTCNOT}]
  \begin{gather*}
    \tikzfigS{./qc-completeness/CNOTCNOT}
    \eqeqref{P0}\tikzfigS{./qc-completeness/CNOTP0CNOT}
    \eqeqref{CNOTPCNOT}\tikzfigS{./qc-completeness/P0Id}
    \eqeqref{P0}\tikzfigS{./qc-completeness/IdId}
  \end{gather*}
\end{proof}

\begin{proof}[Proof of \cref{PcommutCNOT}]
  \begin{gather*}
    \tikzfigS{./qc-completeness/PphiCNOT}
    \eqeqref{CNOTCNOT}\tikzfigS{./qc-completeness/CNOTCNOTPphiCNOT}
    \eqeqref{CNOTPCNOT}\tikzfigS{./qc-completeness/CNOTPphi}
  \end{gather*}
\end{proof}

\begin{proof}[Proof of \cref{Pphasegadget}]
  \begin{gather*}
    \tikzfigS{./identities/Pphasegadget-step-0}
    =\tikzfigS{./identities/Pphasegadget-step-1}
    \eqeqref{bigebre}\tikzfigS{./identities/Pphasegadget-step-2}
    \eqeqref{CNOTPCNOT}\tikzfigS{./identities/Pphasegadget-step-3}
  \end{gather*}
\end{proof}

\begin{proof}[Proof of \cref{HHCNOTHH}]
  \begin{gather*}
    \tikzfigS{./identities/HHCNOTHH_00}
    \eqeqref{CZ}\tikzfigS{./identities/HHCNOTHH_01}
    \eqeqref{Pphasegadget}\tikzfigS{./identities/HHCNOTHH_04}
    \eqeqref{CZ}\tikzfigS{./identities/HHCNOTHH_05}
    \eqeqref{HH}\tikzfigS{./identities/HHCNOTHH_06}
  \end{gather*}
\end{proof}

\begin{proof}[Proof of \cref{ctrlPminuspi}]
  \begin{gather*}
    \tikzfigS{./identities/ctrlPminuspi_00}
    \eqtroiseqref{P0}{Paddition}{CNOTCNOT}\tikzfigS{./identities/ctrlPminuspi_01}
    \eqeqref{CZ}\tikzfigS{./identities/ctrlPminuspi_02}\\[0.4cm]
    \eqdeuxeqref{HH}{CNOTCNOT}\tikzfigS{./identities/ctrlPminuspi_03}
    \eqdeuxeqref{PcommutCNOT}{Pphasegadget}\tikzfigS{./identities/ctrlPminuspi_04}
  \end{gather*}
\end{proof}

\begin{proof}[Proof of \Cref{3CNOTstarget}]
  \begin{gather*}
    \tikzfigS{./qcoriginal-axioms/3CNOTstarget_00}
    \eqeqref{CNOTCNOT}\tikzfigS{./qcoriginal-axioms/3CNOTstarget_01}
    =\tikzfigS{./qcoriginal-axioms/3CNOTstarget_02}\\[0.4cm]
    \eqeqref{HHCNOTHH}\tikzfigS{./qcoriginal-axioms/3CNOTstarget_03}
    \eqeqref{HH}\tikzfigS{./qcoriginal-axioms/3CNOTstarget_04}\\[0.4cm]
    \eqdeuxeqref{CZ}{ctrlPminuspi}\tikzfigS{./qcoriginal-axioms/3CNOTstarget_05}\\[0.4cm]
    \eqdeuxeqref{Pphasegadget}{mctrlPdef}\tikzfigS{./qcoriginal-axioms/3CNOTstarget_06-swaps}
    \eqdeuxeqref{ctrl2pi}{HH}\tikzfigS{./qcoriginal-axioms/3CNOTstarget_07}
    =\tikzfigS{./qcoriginal-axioms/3CNOTstarget_08}
  \end{gather*}
\end{proof}

\end{proof}

\subsection{Intermediate provable equations}\label{appendix:intermediateresults}

In order to prove \Cref{euler3d} we require a huge amount of intermediate equations, most of which have already been proved in \cite{CHMPV} and \cite{CDPV} with the old versions of the equational theory. To use these equations as a step in a derivation in $\QC$, one would need to ensure that there is no circular reasoning, i.e.~proving $P$ using $Q$ and then $Q$ using $P$. As the equational theories of \cite{CHMPV} and \cite{CDPV} are not the same as $\QC$, this is not direct and one have to be unambiguous about the dependencies between all the equations.

\Cref{lem:allAxiomExeptEuler3D} ensures that all the results that have previously been proved in \cite{CDPV} and \cite{CHMPV} without using \Cref{euler3d} still hold in $\QC_3$ (and thus in $\QC$). In this section we sum up all such equations that we need and justify that there is no dependency issue. We also prove some new useful equations.

As the following equations have been proved in \cite{CDPV} and \cite{CHMPV} without using \Cref{euler3d}, by \Cref{lem:allAxiomExeptEuler3D}, they are also provable in $\QC_3$.

\begin{multicols}{3}
  \noindent\begin{align}\tag{R$_{\textup{X}2\pi}$}\tikzfigS{./qc-completeness/RX2pi_00}=\tikzfigS{./qc-completeness/RX2pi_02}\end{align}
  \begin{align}\label{XX}\tag{X$^2$}\tikzfigS{./identities/XX-step-0}=\tikzfigS{./identities/XX-step-3}\end{align}
  \begin{align}\label{ZZ}\tag{Z$^2$}\tikzfigS{./identities/ZZ-step-0}=\tikzfigS{./identities/ZZ-step-4}\end{align}
  \begin{align}\label{Zminuspi}\tag{P$_{-\pi}$}\tikzfigS{./identities/Zminuspi-step-0}=\tikzfigS{./identities/Zminuspi-step-4}\end{align}
  \begin{align}\tag{R$_{\textup{X}+}$}\tikzfigS{./identities/RXaddition-step-0}=\tikzfigS{./identities/RXaddition-step-4}\end{align}
  \begin{align}\label{ZRZ}\tag{R$_{\textup{X}-}$}\tikzfigS{./qcprime-completeness/ZRZ_00}=\tikzfigS{./qcprime-completeness/ZRZ_04}\end{align}
  \begin{align}\label{Heulerminus}\tikzfigS{./gates/H}=\tikzfigS{./identities/Heulerminus}\end{align}
  \begin{align}\label{HeulerRPR}\tikzfigS{./gates/H}=\tikzfigS{./identities/HeulerRPR}\end{align}
  \begin{align}\label{RXcommutCNOT}\tikzfigS{./identities/RXcommutCNOT-step-0}=\tikzfigS{./identities/RXcommutCNOT-step-5}\end{align}
  \begin{align}\label{XcommutCNOT}\tikzfigS{./identities/XcommutCNOT-step-0}=\tikzfigS{./identities/XcommutCNOT-step-5}\end{align}
  \begin{align}\label{RXphasegadget}\tikzfigS{./identities/RXphasegadget-step-0}=\tikzfigS{./identities/RXphasegadget-step-5}\end{align}
  \begin{align}\label{XCNOTXX}\tikzfigS{./qcoriginal-axioms/XCNOTXX_00}=\tikzfigS{./qcoriginal-axioms/XCNOTXX_15}\end{align}
  \begin{align}\label{ZCNOTZZ}\tikzfigS{./identities/CNOTZZ-step-0}=\tikzfigS{./identities/CNOTZZ-step-5}\end{align}
  \begin{align}\label{CNOTstargetcommut}\tikzfigS{./identities/CNOTstargetcommut-step-0}=\tikzfigS{./identities/CNOTstargetcommut-step-3}\end{align}
  \begin{align}\label{CNOTscontrolcommut}\tikzfigS{./qcoriginal-axioms/CNOTscontrolcommut_00}=\tikzfigS{./qcoriginal-axioms/CNOTscontrolcommut_08}\end{align}
  \begin{align}\label{3CNOTscontrol}\tikzfigS{./identities/3CNOTscontrol-step-0}=\tikzfigS{./identities/3CNOTscontrol-step-3}\end{align}
\end{multicols}

In \cite{CHMPV,CDPV}, the multi-controlled phase gate was defined as
\begin{equation}\label{mctrlPfromRX}
  \tikzfigS{./shortcut/mctrlPphi}\defeq\tikzfigS{./shortcut/mctrlPphidef}\left(=\tikzfigS{./shortcut/mctrlPphidefwithgphase}\right)
\end{equation}
and not with \Cref{mctrlPdef}. Thus, all the results that we state below are provable with the definition given by \Cref{mctrlPfromRX}. However, this is not an issue as we prove at the the end of this section that both definitions are in fact equivalent (we can derive one into the other in $\QC_3$).

Following \cite{CHMPV}, we use the standard bullet-based graphical notation for multi-controlled gates where the \emph{anti-control} \tikzfigS{./bulletnotation/w} is a shortcut notation for \tikzfigS{./bulletnotation/XbX}. Moreover, we define the more sophisticated notation $\Lambda^x_yG$ where $x\in\{0,1\}^k$ and $y\in\{0,1\}^l$ to denote the gate $G$ with $k$ controls above (where the $i$-th control is black if $x_i=1$ and white if $x_i=0$), and $l$ controls below (where the $j$-th control is black if $y_j=1$ and white if $y_j=0$). For instance, 
\begin{equation*}
  \Lambda^{10}_{1}P(\varphi)\defeq\tikzfigS{./bulletnotation/bwPb}
\end{equation*}

Additionally, we define $\bar\Lambda^x_yG\defeq\prod_{x'\in\{0,1\}^k,\; y'\in\{0,1\}^l,\; xy\ne x'y'}\Lambda^{x'}_{y'}G$ where $x\in\{0,1\}^k$ and $y\in\{0,1\}^l$. For instance,
\begin{equation*}
  \bar\Lambda^{111}_{\epsilon}P(\varphi)\defeq\tikzfigS{./bulletnotation/barlambda111P}
\end{equation*}
Note that the product can be permuted in any way by using \cref{blackwhitecontrolcommut} below (one can check that the proof of \cref{blackwhitecontrolcommut} does not rely on the notation $\bar\Lambda^x_yG$). For this reason, by abuse of notation, we will use the notation $\bar\Lambda^x_yG$ to denote the product taken in any order, and implicitly use \cref{blackwhitecontrolcommut} to reorder it when needed.

Many equations explaining the intuive behaviour of this bullet-based notation have been proved in \cite{CHMPV} without using \cref{euler3d}. Then, by using \Cref{lem:allAxiomExeptEuler3D} all those equations are provable in $\QC_3$. In the following we state all the required results as separated lemmata.

\begin{lemma}\label{swapcontrl}
  It follows from the axioms of $\QC_3$ that all control qubits play a similar role, i.e.~control can be provably swapped. This is illustrated by the following example.
  \begin{equation*}
    \tikzfigS{./bulletnotation/swapbbwPswap}=\tikzfigS{./bulletnotation/bwbP}
  \end{equation*}
\end{lemma}
\begin{proof}
  This is a direct consequence of the results of \cite{CHMPV}, namely Proposition~11 in \cite{CHMPV}.
\end{proof}

\begin{lemma}\label{mctrlPlift}
  It follows from the axioms of $\QC_3$ that the target qubit of a multi-controlled $P$ gate is actually equivalent to the control qubits. This is illustrated by the following example.
  \begin{equation*}
    \tikzfigS{./bulletnotation/bbbP}=\tikzfigS{./bulletnotation/bPbb}
  \end{equation*}
\end{lemma}
\begin{proof}
  This is a direct consequence of the results of \cite{CHMPV}, namely Proposition~12 in \cite{CHMPV}.
\end{proof}

\begin{lemma}\label{blackwhitecontrolsindependent}
  It follows from the axioms of $\QC_3$ that combining a control and anti-control on the same qubit makes the evolution independent of this qubit. This is illustrated by the following example.
  \begin{equation*}
    \tikzfigS{./bulletnotation/bwPwbbPw}=\tikzfigS{./bulletnotation/bvPw}
  \end{equation*}
\end{lemma}
\begin{proof}
  This is a direct consequence of the results of \cite{CHMPV}, namely Proposition~15 in \cite{CHMPV} (together with \cref{swapcontrl}).
\end{proof}

\begin{lemma}\label{blackwhitecontrolcommut}
  It follows from the axioms of $\QC_3$ that controlled and anti-controlled gates commute (even if the target qubits are not the same in both gates). This is illustrated by the following examples.
  \begin{equation*}
    \tikzfigS{./bulletnotation/bbPwbbXb}=\tikzfigS{./bulletnotation/bbXbbbPw} \qquad\qquad
    \tikzfigS{./bulletnotation/bXbwwbPw}=\tikzfigS{./bulletnotation/wbPwbXbw}
  \end{equation*}
\end{lemma}
\begin{proof}
  This is a direct consequence of the results of \cite{CHMPV}, namely Proposition~16 and Proposition~17 in \cite{CHMPV}.
\end{proof}

\begin{lemma}\label{Palwayscommute}
  It follows from the axioms of $\QC_3$ that two multi-controlled $P$ gates always commute, regardless of the colours and positions of their controls, and of the positions of their targets. This is illustrated by the following example.
  \begin{equation*}
    \tikzfigS{./bulletnotation/wbPhi1bPhi2w}=\tikzfigS{./bulletnotation/bPhi2wwbPhi1}
  \end{equation*}
\end{lemma}
\begin{proof}
  This is a direct consequence of the results of \cite{CHMPV}, namely of Lemma 54 in \cite{CHMPV} together with Lemmata \ref{swapcontrl}, \ref{mctrlPlift} and \ref{blackwhitecontrolsindependent}.
\end{proof}

\begin{lemma}\label{controlscommut}
  It follows from the axioms of $\QC_3$ that multi-controlled gates commute whenever their target qubits are not control qubits of the other. This is illustrated by the following example.
  \begin{equation*}
    \tikzfigS{./bulletnotation/bbRphi1bbRphi2}=\tikzfigS{./bulletnotation/bbRphi2bbRphi1}
  \end{equation*}
\end{lemma}
\begin{proof}
  This is a direct consequence of \cref{mctrlPfromRX} together with \cref{HH}, and \cref{Palwayscommute}.
\end{proof}

\begin{lemma}\label{PcommuteswithRXcontrols}
  It follows from the axioms of $\QC_3$ that a multi-controlled $P$ gate commutes with a multi-controlled $R_X$ gate whenever the target qubit of the $R_X$ gate is neither the target qubit nor a control qubit of the $P$ gate. This is illustrated by the following example.
  \begin{equation*}
    \tikzfigS{./bulletnotation/bPhi1bwPhi2b}=\tikzfigS{./bulletnotation/wPhi2bbPhi1b}
  \end{equation*}
\end{lemma}
\begin{proof}
  This is a direct consequence of \cref{mctrlPfromRX} together with \cref{HH}, and \cref{Palwayscommute}.
\end{proof}

Additionally, the following equations have been shown in \cite{CHMPV} without using \Cref{euler3d}. Then, by \Cref{lem:allAxiomExeptEuler3D}, they are also provable in $\QC_3$.

\begin{multicols}{2}
  \noindent\begin{equation}\label{mctrls0}\tikzfigS{identities/mctrls0}=\tikzfigS{identities/Idn}\end{equation}
  \begin{equation}\label{mctrlP0}\tikzfigS{identities/mctrlP0}=\tikzfigS{identities/Idn}\end{equation}
  \begin{equation}\label{mctrlRX0}\tikzfigS{identities/mctrlRX0}=\tikzfigS{identities/Idn}\end{equation}
  \begin{equation}\label{mctrlsaddition}\tikzfigS{identities/mctrlsphi1mctrlsphi2}=\tikzfigS{identities/mctrlsphi12}\end{equation}
  \begin{equation}\label{mctrlPaddition}\tikzfigS{identities/mctrlPphi1Pphi2}=\tikzfigS{identities/mctrlPphi1phi2}\end{equation}
  \begin{equation}\label{mctrlRXaddition}\tikzfigS{identities/mctrlRXtheta1RXtheta2}=\tikzfigS{identities/mctrlRXtheta1theta2}\end{equation}
  \begin{equation}\label{mctrlRXcommutX}\tikzfigS{identities/XmctrlRX}=\tikzfigS{identities/mctrlRXX}\end{equation}
  \begin{equation}\label{mctrlPop}\tikzfigS{identities/XmctrlPphiX}=\tikzfigS{identities/mctrlPphimctrlPminusphi}\end{equation}
  \begin{equation}\label{mctrlPopfullctrlRX}\tikzfigS{identities/mctrlRXpimctrlPphimctrlRXmoinspi}=\tikzfigS{identities/mctrlPphimctrlPminusphi}\end{equation}
  \begin{align}\label{mctrlRXop}\tikzfigS{./identities/ZmctrlRXthetaZ}=\tikzfigS{./identities/mctrlRXminustheta}\end{align}
  \begin{equation}\label{mctrlZRX}\tikzfigS{identities/mctrlPpiRX}=\tikzfigS{identities/mctrlRXmoinsPpi}\end{equation}
  \begin{equation}\label{mctrlPcommutCNOT}\tikzfigS{identities/mctrlPcommutCNOTL}=\tikzfigS{identities/mctrlPcommutCNOTR}\end{equation}
  \begin{equation}\label{mctrlRXcommutCNOT}\tikzfigS{identities/mctrlRXCNOT}=\tikzfigS{identities/CNOTmctrlRX}\end{equation}
  \begin{equation}\label{mctrlRXphasegadget}\tikzfigS{identities/mctrlRXphasegadgetL}=\tikzfigS{identities/mctrlRXphasegadgetR}\end{equation}
  \begin{equation}\label{mctrlPcommutgadget}\tikzfigS{identities/mctrlPcommutgadgetL}=\tikzfigS{identities/mctrlPcommutgadgetR}\end{equation}
  \begin{equation}\label{mctrlRXcommutgadget}\tikzfigS{identities/mctrlRXcommutgadgetL}=\tikzfigS{identities/mctrlRXcommutgadgetR}\end{equation}
\end{multicols}

\begin{proof}
  Equations \eqref{mctrls0}, \eqref{mctrlP0}, \eqref{mctrlRX0}, \eqref{mctrlsaddition}, \eqref{mctrlPaddition} and \eqref{mctrlRXaddition} are proved in Proposition~13 of \cite{CHMPV}. \Cref{mctrlRXcommutX} is a particular case of Lemma~52 of \cite{CHMPV}. \Cref{mctrlPop} follows directly from Lemma~53 of \cite{CHMPV} (together with \cref{Palwayscommute}). The proof of \Cref{mctrlRXop} is given in the proof of Lemma~47 of \cite{CHMPV}. Equations \eqref{mctrlPcommutCNOT} and \eqref{mctrlRXcommutCNOT} can be straightforwardly proved using \cref{mctrlPfromRX} (of the present paper) and Lemma~50 of \cite{CHMPV}. \Cref{mctrlRXphasegadget} is the particular case of Lemma~46 of \cite{CHMPV} with only black controls. \Cref{mctrlRXcommutgadget} is a particular case of Equation~(24) of \cite{CHMPV} and \Cref{mctrlPcommutgadget} can be straightforwardly proved using \cref{mctrlPfromRX} (of the present paper) and \Cref{mctrlRXcommutgadget}. \Cref{mctrlZRX} corresponds to Lemma~58 of \cite{CHMPV} and \Cref{mctrlPopfullctrlRX} can be proved as follows.
  \begin{gather*}
    \tikzfigS{./identities/mctrlRXpimctrlPphimctrlRXmoinspi}
    \eqtroiseqref{XX}{mctrlRX0}{mctrlRXaddition}\tikzfigS{./identities/mctrlPopfullctrlRX_01}
    \overset{\text{\cref{blackwhitecontrolcommut}}}{=}\tikzfigS{./identities/mctrlPopfullctrlRX_02}\\[0.4cm]
    \overset{\text{\cref{blackwhitecontrolsindependent}}}{=}\tikzfigS{./identities/mctrlPopfullctrlRX_03}
    \eqtroiseqref{RXdef}{Xdef}{Zminuspi}\tikzfigS{./identities/mctrlPopfullctrlRX_04}
    \eqdeuxeqref{gphaseaddition}{S0}\tikzfigS{./identities/mctrlPopfullctrlRX_05}
    \eqeqref{mctrlPop}\tikzfigS{./identities/mctrlPphimctrlPminusphi}
  \end{gather*}
\end{proof}

All the previous results have already been proved in \cite{CHMPV} and \cite{CDPV}. This gives a bunch of intermediate equations that we can use. We also prove some new equations in $\QC$.

\begin{multicols}{4}
  \noindent\begin{equation}\label{ctrl2pi2qubits}\tikzfigS{identities/ctrlP2pi}=\tikzfigS{identities/Id2}\end{equation}
  \begin{equation}\label{mctrls2pi}\tikzfigS{identities/mctrls2pi}=\tikzfigS{identities/Idn}\end{equation}
  \begin{equation}\label{mctrlRX2pi}\tikzfigS{identities/mctrlRX2pi}=\tikzfigS{identities/mctrlspisurfil}\end{equation}
  \begin{equation}\label{mctrlRX4pi}\tikzfigS{identities/mctrlRX4pi}=\tikzfigS{identities/Idn}\end{equation}
\end{multicols}
\begin{equation}\label{mctrlPthroughRXpi}\tikzfigS{identities/mctrlRXpimctrlPw}=\tikzfigS{identities/mctrlPmctrlRXpi}\end{equation}

\begin{proof}[Proof of \Cref{ctrl2pi2qubits}]
  \begin{gather*}
    \tikzfigS{./identities/ctrlP2pi}
    \eqeqref{mctrlPdef}\tikzfigS{./identities/ctrlP2pi_04}
    \eqdeuxeqref{Zdef}{Zminuspi}\tikzfigS{./identities/ctrlP2pi_05}
    \eqeqref{ZCNOTZZ}\tikzfigS{./identities/ctrlP2pi_06}
    \eqdeuxeqref{CNOTCNOT}{ZZ}\tikzfigS{./identities/Id2}
  \end{gather*} 
\end{proof}

\begin{proof}[Proof of \Cref{mctrls2pi}]
  This is a direct consequence of Equations \eqref{gphaseempty}, \eqref{P2pi}, \eqref{ctrl2pi2qubits} and \eqref{ctrl2pi}.
\end{proof}

\begin{proof}[Proof of \Cref{mctrlRX2pi}]
  \begin{gather*}
    \tikzfigS{./identities/mctrlRX2pi}
    \eqtroiseqref{HH}{mctrls2pi}{mctrlsaddition}\tikzfigS{./identities/mctrlRX2pi_01}
    \eqeqref{mctrlPfromRX}\tikzfigS{./identities/mctrlRX2pi_02}
    \eqquatreeqref{ctrl2pi}{ctrl2pi2qubits}{P2pi}{HH}\tikzfigS{./identities/mctrlspisurfil}
  \end{gather*} 
\end{proof}
\begin{proof}[Proof of \Cref{mctrlRX4pi}]
  This is a direct consequence of Equations \eqref{mctrlRXaddition}, \eqref{mctrlRX2pi}, \eqref{mctrlsaddition} and \eqref{mctrls2pi}.
\end{proof}

\begin{proof}[Proof of \Cref{mctrlPthroughRXpi}]
  \begin{gather*}
    \tikzfigS{identities/mctrlPmctrlRXpi}
    \eqtroiseqref{XX}{mctrlRX0}{mctrlRXaddition}\tikzfigS{identities/mctrlPthroughRXpi1}
    \overset{\text{\cref{blackwhitecontrolsindependent}}}{=}\tikzfigS{identities/mctrlPthroughRXpi2}\\[0.4cm]
    \overset{\text{\cref{blackwhitecontrolcommut}}}{=}\tikzfigS{identities/mctrlPthroughRXpi3}
    \eqtroiseqref{RXdef}{Zdef}{Xdef}\tikzfigS{identities/mctrlPthroughRXpi4}
    \eqeqref{XX}\tikzfigS{identities/mctrlPthroughRXpi5}\\[0.4cm]
    \eqtroiseqref{RXdef}{Zdef}{Xdef}\tikzfigS{identities/mctrlPthroughRXpi6}
    \overset{\text{\cref{blackwhitecontrolsindependent}}}{=}\tikzfigS{identities/mctrlPthroughRXpi7}
    \eqtroiseqref{mctrlRXaddition}{mctrlRX0}{XX}\tikzfigS{identities/mctrlRXpimctrlPw}
  \end{gather*}
\end{proof}

Finally, all the results given above (except \cref{ctrl2pi2qubits}) are proved using the definition of the multi-controlled phase gate given by \Cref{mctrlPfromRX} instead of \Cref{mctrlPdef}. It was already proved in \cite{CDPV} (without using \cref{euler3d}) that \cref{mctrlPdef} can be derived from \cref{mctrlPfromRX}. To prove that the definitions are actually equivalent, it remains to prove that conversely, \cref{mctrlPfromRX} can be derived from \cref{mctrlPdef}. To do so, we cannot rely on \cref{mctrlPfromRX} itself, and therefore we cannot directly use the results above unless they do not involve multi-controlled phase gates. To circumvent this issue, we proceed by induction on the number of qubits. \Cref{mctrlPfromRX} on $1$ qubit follows directly from \crefnosortnocompress{RXdef,HH,S0,gphaseaddition}. To do the induction step, note that for any $n\geq 1$, to prove the results above on up to $n$ qubits, we only use \cref{mctrlPfromRX} on at most $n$ qubits. Hence, we can use the results above on up to $n$ qubits to prove \cref{mctrlPfromRX} on $n+1$ qubits. Note additionally that the property of \cref{swapcontrl} for multi-controlled $R_X$ gates does not involve multi-controlled $P$ gates, therefore it remains valid for any number of qubits. These two remarks being done, the induction step is as follows:
\begin{gather*}
  \tikzfigS{./identities/mctrlPinducdef_00}
  \overset{\text{\cref{swapcontrl} for $R_X$}}{\eqeqref{mctrlRXdef}}\tikzfigS{./identities/mctrlPinducdef_01}
  \eqeqref{HH}\tikzfigS{./identities/mctrlPinducdef_02}\\[0.4cm]
  \eqeqref{HHCNOTHH}\tikzfigS{./identities/mctrlPinducdef_03}
  \eqdeuxeqref{mctrls0}{mctrlsaddition}\tikzfigS{./identities/mctrlPinducdef_04}\\[0.4cm]
  \overset{\text{\cref{PcommuteswithRXcontrols}}}{=}\tikzfigS{./identities/mctrlPinducdef_05}
  \eqeqref{mctrlPfromRX}\tikzfigS{./identities/mctrlPinducdef_06}
\end{gather*}

\subsection{Proof of Equation \eqref{ctrl2pi} on 3 qubits in $\QCancilla$ and $\QCancillaprime$}\label{appendix:proof5CXancilla}

\begin{gather*}
  \tikzfigS{./qcancilla-completeness/ccP2pi_00-swaps}
  \eqdeuxeqref{mctrlPdef}{Pphasegadget}\tikzfigS{./qcancilla-completeness/ccP2pi_01}
  \eqdeuxeqref{CZ}{ctrlPminuspi}\tikzfigS{./qcancilla-completeness/ccP2pi_02}\\[0.4cm]
  \eqeqref{HH}\tikzfigS{./qcancilla-completeness/ccP2pi_03}
  \eqeqref{HHCNOTHH}\tikzfigS{./qcancilla-completeness/ccP2pi_04}
  =\tikzfigS{./qcancilla-completeness/ccP2pi_05}\\[0.4cm]
  \eqeqref{5CX-ancilla}\tikzfigS{./qcancilla-completeness/ccP2pi_06}
  \eqeqref{CNOTCNOT}\tikzfigS{./qcancilla-completeness/ccP2pi_07}
  \eqeqref{HH}\tikzfigS{./qcancilla-completeness/ccP2pi_08}
\end{gather*}

\section{Completeness of $\QC$}

\subsection{Proof of Equation \eqref{Euler2dmulticontrolled}}\label{preuveEuler2dmulticontrolled}

\newcommand{\eqrefJ}[1]{\textup{(\hyperref[Euler2dmulticontrolled]{\ref*{euler}$^*_{#1}$})}}
\newcommand{\crefJ}[1]{Equation~\textup{(\hyperref[Euler2dmulticontrolled]{\ref*{euler}$^*_{#1}$})}}
\newcommand{\eqrefJinv}[1]{\textup{(\hyperref[Euler2dmulticontrolled]{\ref*{euler}$^{*\textup{inv}}_{#1}$})}}
\newcommand{\crefJinv}[1]{Equation~\textup{(\hyperref[Euler2dmulticontrolled]{\ref*{euler}$^{*\textup{inv}}_{#1}$})}}

For this proof, it is useful to remark that the following equation is sound:
\begin{equation}\label{inverseEuler2dmulticontrolled}\tag{\ref*{euler}$^{*\textup{inv}}_n$}\tikzfigS{./Preuve-Euler2d/inverseEuler2dleft-multicontrolled-compact}=\tikzfigS{./Preuve-Euler2d/inverseEuler2dright-multicontrolled-compact}\end{equation}
where $\beta'_0=\beta_0+\frac{\alpha_1-\alpha_2+\alpha_3+\beta_1-\beta_2+\beta_3}2$, with $\alpha_i$ and $\beta_i$ being parameters of any instance of \Cref{eulerstar}.

The soundness of \cref{inverseEuler2dmulticontrolled} follows from that of \cref{Euler2dmulticontrolled} and of the axioms of $\QC$, thanks to the following derivation:
\begin{eqnarray*}
  \tikzfigS{./Preuve-Euler2d/inverseEuler2dleft-multicontrolled-compact}
  &\eqquatreeqref{mctrlPfromRX}{HH}{mctrls0}{mctrlsaddition}&\tikzfigS{./Preuve-Euler2d/inverseEuler2dmctrl0}\\[0.4cm]
  &\overset{\text{Lemmas \ref{Palwayscommute}}}{\overset{\text{and \ref{PcommuteswithRXcontrols}}}{\eqeqref{mctrlsaddition}}}&\tikzfigS{./Preuve-Euler2d/inverseEuler2dmctrl1}\\[0.4cm]
  &\eqdeuxeqref{Euler2dmulticontrolled}{mctrlsaddition}&\tikzfigS{./Preuve-Euler2d/inverseEuler2dmctrl2}\\[0.4cm]
  &\eqquatreeqref{mctrlPfromRX}{HH}{mctrls0}{mctrlsaddition}&\tikzfigS{./Preuve-Euler2d/inverseEuler2dmctrl3}\\[0.4cm]
  &\overset{\text{Lemmas \ref{Palwayscommute}}}{\overset{\text{and \ref{PcommuteswithRXcontrols}}}{\eqeqref{mctrlsaddition}}}&\tikzfigS{./Preuve-Euler2d/inverseEuler2dright-multicontrolled-compact}
\end{eqnarray*}

Note that once one knows that \cref{inverseEuler2dmulticontrolled} is sound, it actually corresponds to \cref{Euler2dmulticontrolled} applied from right to left together with \cref{mctrls0,mctrlsaddition}. Conversely, \cref{Euler2dmulticontrolled} corresponds to \cref{inverseEuler2dmulticontrolled} applied from right to left together with \cref{mctrls0,mctrlsaddition}. Hence, \cref{Euler2dmulticontrolled} and \cref{inverseEuler2dmulticontrolled} are equivalent in the presence of \cref{mctrls0,mctrlsaddition} (and therefore in the presence of the axioms of $\QC$).\bigskip

To prove that \cref{Euler2dmulticontrolled} follows from the axioms of $\QC$, we proceed by induction on $n$. \crefJ{1} corresponds to \cref{eulerstar}. It remains to derive \crefJ{n} using \crefJ{n-1} and the axioms of $\QC$, for any $n\geq2$. In the rest of the proof, note that whenever we add or remove $\,\minitikzfig[0.6]{shortcut/0+1P2pi}\,$ in a circuit, for simplicity we will justify that by \cref{ctrl2pi} even though for small $n$ that may actually correspond to \crefornosortnocompress{P2pi,ctrl2pi2qubits}.\medskip

One has, on the one hand,
\begin{eqnarray*}
  &&\tikzfigS{./Preuve-Euler2d/Euler2dleft-multicontrolled}\\[0.4cm]
  &\overset{\text{\cref{swapcontrl}}}{\eqdeuxeqref{mctrlRXdef}{mctrlRXcommutgadget}}&\tikzfigS{./Preuve-Euler2d/Euler2dinduction1}\\[0.4cm]
  &\eqeqref{HH}&\tikzfigS{./Preuve-Euler2d/Euler2dinduction2}\\[0.4cm]
  &
  \eqtroiseqref{CZ}{ctrlPminuspi}{mctrlPdef}&\tikzfigS{./Preuve-Euler2d/Euler2dinductioncorr3}\\[0.4cm]
  &\overset{\text{\cref{mctrlPlift}}}{=}&\tikzfigS{./Preuve-Euler2d/Euler2dinductioncorr4}\\[0.4cm]
  &\overset{\text{\Cref{Palwayscommute}}}{\eqdeuxeqref{mctrlPaddition}{mctrlP0}}&\tikzfigS{./Preuve-Euler2d/Euler2dinductioncorr5}\\[0.4cm]
  &\overset{\text{Lemmas \ref{blackwhitecontrolsindependent}}}{\overset{\text{and \ref{blackwhitecontrolcommut}}}{\eqdeuxeqref{mctrlPaddition}{mctrlP0}}}&\tikzfigS{./Preuve-Euler2d/Euler2dinductioncorr6}
\\[0.4cm]
  &\eqeqref{mctrlPdef}&\tikzfigS{./Preuve-Euler2d/Euler2dinductioncorr7}\\[0.4cm]
  &
  \overset{\text{\cref{controlscommut}}}{\eqdeuxeqref{mctrlPcommutCNOT}{mctrlPaddition}}&\tikzfigS{./Preuve-Euler2d/Euler2dinductioncorr8}\\[0.4cm]
  &\eqeqref{mctrlPcommutgadget}&\tikzfigS{./Preuve-Euler2d/Euler2dinductioncorr9}\\[0.4cm]
  &\eqdeuxeqref{mctrlRXcommutCNOT}{CNOTCNOT}&\tikzfigS{./Preuve-Euler2d/Euler2dinductioncorr10}\\[0.4cm]
  &
\overset{\eqrefJinv{n-1}}{=}&\tikzfigS{./Preuve-Euler2d/Euler2dinductioncorr11}\\[0.4cm]
  &\overset{\text{Lemmas \ref{Palwayscommute}}}{\overset{\text{and \ref{PcommuteswithRXcontrols}}}{=}}&\tikzfigS{./Preuve-Euler2d/Euler2dinductioncorr11bis}\\[0.4cm]
  &\eqeqref{mctrlRXcommutCNOT}&\tikzfigS{./Preuve-Euler2d/Euler2dinductioncorr12}\\[0.4cm]
  &\overset{\eqrefJ{n-1}}{\overset{\text{Lemmas \ref{Palwayscommute}}}{\overset{\text{and \ref{PcommuteswithRXcontrols}}}{\overset{\text{\eqref{mctrlsaddition}}}{=}}}}&\tikzfigS{./Preuve-Euler2d/Euler2dinductioncorr13}\\[0.4cm]
  &\eqdeuxeqref{mctrlPcommutgadget}{mctrlPaddition}&\tikzfigS{./Preuve-Euler2d/Euler2dinductioncorr14}\\[0.4cm]
  &\eqdeuxeqref{mctrlRXcommutCNOT}{CNOTCNOT}&\tikzfigS{./Preuve-Euler2d/Euler2dinductioncorr15}\\[0.4cm]
  &
\overset{\eqrefJinv{n-1}}{\overset{\text{Lemmas \ref{Palwayscommute}}}{\overset{\text{and \ref{PcommuteswithRXcontrols}}}{\overset{\text{\eqref{mctrlsaddition}}}{=}}}}&\tikzfigS{./Preuve-Euler2d/Euler2dinductioncorr16}\\[0.4cm]
  &\eqdeuxeqref{mctrlRXcommutCNOT}{mctrlRXaddition}&\tikzfigS{./Preuve-Euler2d/Euler2dinductioncorr17}.
\end{eqnarray*}

On the other hand,
\begin{eqnarray*}
  &&\tikzfigS{./Preuve-Euler2d/Euler2dright-multicontrolled}\\[0.4cm]
  &\overset{\eqref{mctrlPfromRX}\eqref{mctrlPaddition}}{\overset{\text{Lemmas \ref{Palwayscommute}}}{\overset{\text{and \ref{PcommuteswithRXcontrols}}}{=}}}&\tikzfigS{./Preuve-Euler2d/Euler2dright1}\\[0.4cm]
  &\eqdeuxeqref{eulerH}{Heulerminus}&\tikzfigS{./Preuve-Euler2d/Euler2dright2}\\[0.4cm]
  &\overset{\text{Lemmas \ref{blackwhitecontrolsindependent}}}{\overset{\text{and \ref{blackwhitecontrolcommut}}}{\eqdeuxeqref{mctrlPaddition}{mctrlP0}}}&\tikzfigS{./Preuve-Euler2d/Euler2dright3}\\[0.4cm]
  &\overset{\text{similar}}{\overset{\text{derivation}}{\overset{\text{to above}}{=}}}&\tikzfigS{./Preuve-Euler2d/Euler2dright4}\\[0.4cm]
  &
\overset{\eqrefJinv{n-1}}{\overset{\text{Lemmas \ref{Palwayscommute}}}{\overset{\text{and \ref{PcommuteswithRXcontrols}}}{\overset{\text{\eqref{mctrlsaddition}\eqref{mctrlRXaddition}}}{=}}}}&\tikzfigS{./Preuve-Euler2d/Euler2dright5}.
\end{eqnarray*}

Let $A\coloneqq\interp{\tikzfigS{./Preuve-Euler2d/RXPRXL-L}}$, $B\coloneqq\interp{\tikzfigS{./Preuve-Euler2d/RXPRXL-R}}$, $C\coloneqq\interp{\tikzfigS{./Preuve-Euler2d/RXPRXR-L}}$ and $D\coloneqq\interp{\tikzfigS{./Preuve-Euler2d/RXPRXR-R}}$. Let also\footnote{Note that the definition of $M_0$ implies that it is actually equal to the identity. However, this does not play any role in the proof.} $M_0\coloneqq(\bra{1...10}\otimes I_2)\interp{\tikzfigS{./Preuve-Euler2d/Euler2dleft-multicontrolled}}(\ket{1...10}\otimes I_2)$ and\linebreak $M_1\coloneqq(\bra{1...11}\otimes I_2)\interp{\tikzfigS{./Preuve-Euler2d/Euler2dleft-multicontrolled}}(\ket{1...11}\otimes I_2)$\medskip. With the abuse of notation of denoting $P(\varphi)\coloneqq\interp{\tikzfigS{./gates/Pphi}}=\begin{pmatrix}1&0\\0&e^{i\varphi}\end{pmatrix}$ and $X\coloneqq\interp{\tikzfigS{./gates/X}}=\begin{pmatrix}0&1\\1&0\end{pmatrix}$, it follows from the first derivation that
\[M_0=BP(\psi_8)A\]
and
\[M_1\eqspace e^{i\frac{\alpha_2}2}BXP(\psi_8)XA\eqspace e^{i(\frac{\alpha_2}2+\psi_8)}BP(-\psi_8)A.\]
On the other hand, it follows from the second derivation together with the soundness of \Cref{Euler2dmulticontrolled} that
\[M_0=DP(\theta_4)C\]
and
\[M_1\eqspace e^{i\theta_0}DXP(\theta_4)XC\eqspace e^{i(\theta_0+\theta_4)}DP(-\theta_4)C.\]
Therefore, one has
\[M_0\eqspace BP(\psi_8)A\eqspace DP(\theta_4)C\]
so that
\begin{eqnabc}\label{eqM0}B=DP(\theta_4)CA^\dag P(-\psi_8).\end{eqnabc}
One also has
\[M_0^\dag M_1\begin{array}[t]{@{\ }c@{\ }c@{\ }c@{\ }c}=& \bigl(A^\dag P(-\psi_8)B^\dag\bigr) \bigl(e^{i(\frac{\alpha_2}2+\psi_8)}BP(-\psi_8)A\bigr)&=& e^{i(\frac{\alpha_2}2+\psi_8)}A^\dag P(-2\psi_8)A\\[1em] =& \bigl(C^\dag P(-\theta_4)D^\dag\bigr) \bigl(e^{i(\theta_0+\theta_4)}DP(-\theta_4)C\bigr)&=& e^{i(\theta_0+\theta_4)}C^\dag P(-2\theta_4)C\end{array}\]
so that
\[e^{i(\frac{\alpha_2}2+\psi_8)}A^\dag P(-2\psi_8)A\eqspace e^{i(\theta_0+\theta_4)}C^\dag P(-2\theta_4)C\]
from which it follows that
\begin{eqnabc}\label{eqM0M1}e^{i(\frac{\alpha_2}2+\psi_8)} P(-2\psi_8)AC^\dag\eqspace e^{i(\theta_0+\theta_4)}AC^\dag P(-2\theta_4).\end{eqnabc}
Let us denote the entries of $AC^\dag$ by $AC^\dag\eqqcolon\begin{pmatrix}u_{00}&u_{01}\\u_{10}&u_{11}\end{pmatrix}$. \Cref{eqM0M1} becomes
\begin{eqnabc}\label{eqM0M1U}\begin{pmatrix}e^{i(\frac{\alpha_2}2+\psi_8)}u_{00}&e^{i(\frac{\alpha_2}2+\psi_8)}u_{01}\\e^{i(\frac{\alpha_2}2-\psi_8)}u_{10}&e^{i(\frac{\alpha_2}2-\psi_8)}u_{11}\end{pmatrix}=\begin{pmatrix}e^{i(\theta_0+\theta_4)}u_{00}&e^{i(\theta_0-\theta_4)}u_{01}\\e^{i(\theta_0+\theta_4)}u_{10}&e^{i(\theta_0-\theta_4)}u_{11}\end{pmatrix}.\end{eqnabc}
Since $AC^\dag$ is unitary, there are three cases:
\begin{itemize}
\item $u_{00},u_{01},u_{10},u_{11}\neq0$
\item $u_{01}=u_{10}=0$ and $u_{00},u_{11}\neq0$
\item $u_{00}=u_{11}=0$ and $u_{01},u_{10}\neq0$.
\end{itemize}
\bigskip
\begin{itemize}
\item If $u_{00},u_{01},u_{10},u_{11}\neq0$, then \Cref{eqM0M1U} gives $\frac{\alpha_2}2+\psi_8\equiv\theta_0+\theta_4\equiv\frac{\alpha_2}2-\psi_8\pmod {2\pi}$, so that $\psi_8\equiv 0\pmod\pi$; and $\theta_0+\theta_4\equiv\frac{\alpha_2}2+\psi_8\equiv\theta_0-\theta_4\pmod {2\pi}$, so that $\theta_4\equiv0\pmod\pi$. Thus there are four cases:
\begin{itemize}
\item $\psi_8\equiv\theta_4\equiv0\pmod{2\pi}$ and $\frac{\alpha_2}2\equiv\theta_0\pmod{2\pi}$
\item $\psi_8\equiv\theta_4\equiv\pi\pmod{2\pi}$ and $\frac{\alpha_2}2\equiv\theta_0\pmod{2\pi}$
\item $\psi_8\equiv\pi\pmod{2\pi}$, $\theta_4\equiv0\pmod{2\pi}$ and $\frac{\alpha_2}2\equiv\theta_0+\pi\pmod{2\pi}$
\item $\psi_8\equiv0\pmod{2\pi}$, $\theta_4\equiv\pi\pmod{2\pi}$ and $\frac{\alpha_2}2\equiv\theta_0+\pi\pmod{2\pi}$
\end{itemize}
In the first and fourth cases, the circuit obtained from the left-hand side of \Cref{Euler2dmulticontrolled} can be transformed as follows:
\begin{eqnarray*}
  &&\tikzfigS{./Preuve-Euler2d/Euler2dinductioncorr17}\\[0.4cm]
  &\eqtroiseqref{ctrl2pi}{mctrlPaddition}{CNOTCNOT}&\tikzfigS{./Preuve-Euler2d/Euler2dinductioncorr17-2}\\[0.4cm]
  &\overset{\eqrefJ{n-1}}{\overset{\text{Lemmas \ref{Palwayscommute}}}{\overset{\text{and \ref{PcommuteswithRXcontrols}}}{\overset{\text{\eqref{mctrlsaddition}\eqref{mctrlPaddition}\eqref{mctrlRXaddition}}}{=}}}}&\tikzfigS{./Preuve-Euler2d/Euler2dinductioncorr17-3}
\end{eqnarray*}
while in the second and third cases, it can be transformed as follows:
\begin{eqnarray*}
  &&\tikzfigS{./Preuve-Euler2d/Euler2dinductioncorr17}\\[0.4cm]
  &\eqdeuxeqref{ctrl2pi}{mctrlPaddition}&\tikzfigS{./Preuve-Euler2d/Euler2dinductioncorr17-4}\\[0.4cm]
  &\overset{\eqref{mctrlP0}\eqref{mctrlPaddition}}{\overset{\text{Lemmas \ref{Palwayscommute}}}{\overset{\text{and \ref{controlscommut}}}{=}}}&\tikzfigS{./Preuve-Euler2d/Euler2dinductioncorr17-5}\\[0.4cm]
  &\eqeqref{mctrlPdef}&\tikzfigS{./Preuve-Euler2d/Euler2dinductioncorr17-6}\\[0.4cm]
  &\eqeqref{ctrl2pi}&\tikzfigS{./Preuve-Euler2d/Euler2dinductioncorr17-7}\\[0.4cm]
  &\overset{\eqrefJ{n-1}}{\overset{\text{Lemmas \ref{Palwayscommute}}}{\overset{\text{and \ref{PcommuteswithRXcontrols}}}{\overset{\text{\eqref{mctrlsaddition}\eqref{mctrlPaddition}\eqref{mctrlRXaddition}}}{=}}}}&\tikzfigS{./Preuve-Euler2d/Euler2dinductioncorr17-8}.
\end{eqnarray*}
Concerning the circuit obtained from the right-hand side of \Cref{Euler2dmulticontrolled}, in the first and third cases, it can be transformed as follows:
\begin{eqnarray*}
  &&\tikzfigS{./Preuve-Euler2d/Euler2dright5}\\[0.4cm]
  &\eqtroiseqref{ctrl2pi}{mctrlPaddition}{CNOTCNOT}&\tikzfigS{./Preuve-Euler2d/Euler2dright5-2}\\[0.4cm]
  &\overset{\eqrefJ{n-1}}{\overset{\text{Lemmas \ref{Palwayscommute}}}{\overset{\text{and \ref{PcommuteswithRXcontrols}}}{\overset{\text{\eqref{mctrlsaddition}\eqref{mctrlPaddition}\eqref{mctrlRXaddition}}}{=}}}}&\tikzfigS{./Preuve-Euler2d/Euler2dright5-3}
\end{eqnarray*}
while in the second and fourth cases, it can undergo a similar transformation as the left-hand side did in the second and third cases.

In all four cases, the parameters of the two final circuits obtained from the two sides of \cref{Euler2dmulticontrolled} can be made equal as follows. First, in each case, the relation between $\frac{\alpha_2}2$ and $\theta_0$ implies that the parameter of the leftmost multi-controlled $P$ gate differs by a multiple of $2\pi$ between the two circuits. Hence it can be made the same by means of \cref{ctrl2pi,mctrlPaddition}. Then, in both circuits, one can make the parameters of the four other multi-controlled gates to satisfy the conditions of a normal form (see \cref{def:normalform}), by applying first \crefJinv{n-1}, then \crefJ{n-1} with an appropriate RHS (namely, corresponding to the multi-controlled version of a normal form -- such a RHS exists thanks to \cref{lem:unicityNormalFormEuler}), and finally \cref{mctrls2pi,mctrlsaddition}. Since the two circuit have same semantics, the uniqueness in \cref{lem:unicityNormalFormEuler} implies that one gets two identical circuits, which completes the derivation of \crefJ{n}.
\bigskip

\item If $u_{01}=u_{10}=0$ and $u_{00},u_{11}\neq0$, then \Cref{eqM0M1U} gives $\frac{\alpha_2}2+\psi_8\equiv\theta_0+\theta_4\pmod {2\pi}$ and $\frac{\alpha_2}2-\psi_8\equiv\theta_0-\theta_4\pmod {2\pi}$, so that $\psi_8\equiv\theta_4\pmod\pi$. Thus there are two cases:
\begin{itemize}
\item $\psi_8\equiv\theta_4\pmod{2\pi}$ and $\frac{\alpha_2}2\equiv\theta_0\pmod{2\pi}$
\item $\psi_8\equiv\theta_4+\pi\pmod{2\pi}$ and $\frac{\alpha_2}2\equiv\theta_0+\pi\pmod{2\pi}$.
\end{itemize}
Let $\upsilon_0\coloneqq\arg(u_{00})$ and $\upsilon_1\coloneqq\arg(\frac{u_{11}}{u_{00}})$, so that $A=e^{i\upsilon_0}P(\upsilon_1)C$. In the first case, \Cref{eqM0} gives $B=e^{-i\upsilon_0}DP(-\upsilon_1)$, and one has
\begin{eqnarray*}
  &&\tikzfigS{./Preuve-Euler2d/Euler2dright5}\\[0.4cm]
  &\eqdeuxeqref{ctrl2pi}{mctrlPaddition}&\tikzfigS{./Preuve-Euler2d/Euler2dright5-4}\\[0.4cm]
  &\eqtroiseqref{mctrlP0}{mctrlPaddition}{mctrlPcommutgadget}&\tikzfigS{./Preuve-Euler2d/Euler2dright5-5}\\[0.4cm]
  &\overset{\eqrefJ{n-1}}{\overset{\text{Lemmas \ref{Palwayscommute}}}{\overset{\text{and \ref{PcommuteswithRXcontrols}}}{\overset{\text{\eqref{mctrlsaddition}\eqref{mctrlPaddition}}}{=}}}}&\tikzfigS{./Preuve-Euler2d/Euler2dright5-5-1}\\[0.4cm]
  &\overset{\eqrefJinv{n-1}}{\overset{\text{Lemmas \ref{Palwayscommute}}}{\overset{\text{and \ref{PcommuteswithRXcontrols}}}{\overset{\text{\eqref{mctrlsaddition}}}{=}}}}&\tikzfigS{./Preuve-Euler2d/Euler2dinductioncorr17}.
\end{eqnarray*}
The fact that we can get the desired parameters at the last step is justified by the fact that up to a global phase, $\interp{\tikzfigS{./Preuve-Euler2d/diag-case1-L}}=\interp{\tikzfigS{./Preuve-Euler2d/RXPRXL-L-sansphase}}\mathrel{(=}A)$ and $\interp{\tikzfigS{./Preuve-Euler2d/diag-case1-R}}=\interp{\tikzfigS{./Preuve-Euler2d/RXPRXL-R}}\mathrel{(=}B)$, so that there exists an appropriate instance of \crefJinv{n-1} (indeed, all sound instances of \crefJinv{n-1} can be obtained from \crefJ{n-1}).\medskip

In the second case, \Cref{eqM0} gives $B=e^{-i\upsilon_0}DP(-\upsilon_1-\pi)$, and one has
\begin{eqnarray*}
  &&\tikzfigS{./Preuve-Euler2d/Euler2dright5}\\[0.4cm]
  &\eqdeuxeqref{ctrl2pi}{mctrlPaddition}&\tikzfigS{./Preuve-Euler2d/Euler2dright5-6}\\[0.4cm]
  &\eqtroiseqref{mctrlP0}{mctrlPaddition}{mctrlPcommutgadget}&\tikzfigS{./Preuve-Euler2d/Euler2dright5-7}\\[0.4cm]
  &\overset{\eqref{mctrlPaddition}}{\overset{\text{Lemmas \ref{Palwayscommute}}}{\overset{\text{and \ref{PcommuteswithRXcontrols}}}{\eqeqref{CNOTCNOT}}}}&\tikzfigS{./Preuve-Euler2d/Euler2dright5-8}\\[0.4cm]
  &\eqeqref{mctrlPdef}&\tikzfigS{./Preuve-Euler2d/Euler2dright5-9}\\[0.4cm]
  &\eqeqref{ctrl2pi}&\tikzfigS{./Preuve-Euler2d/Euler2dright5-10}\\[0.4cm]
  &\overset{\eqrefJ{n-1}}{\overset{\text{Lemmas \ref{Palwayscommute}}}{\overset{\text{and \ref{PcommuteswithRXcontrols}}}{\overset{\text{\eqref{mctrlsaddition}\eqref{mctrlPaddition}}}{=}}}}&\tikzfigS{./Preuve-Euler2d/Euler2dright5-5-2}\\[0.4cm]
  &\overset{\eqrefJinv{n-1}}{\overset{\text{Lemmas \ref{Palwayscommute}}}{\overset{\text{and \ref{PcommuteswithRXcontrols}}}{\overset{\text{\eqref{mctrlsaddition}}}{=}}}}&\tikzfigS{./Preuve-Euler2d/Euler2dinductioncorr17}.
\end{eqnarray*}
The fact that we can get the desired parameters at the last step is justified by a similar argument to the first case.\bigskip

\item If $u_{00}=u_{11}=0$ and $u_{01},u_{10}\neq0$, then 
\Cref{eqM0M1U} gives $\frac{\alpha_2}2+\psi_8\equiv\theta_0-\theta_4\pmod {2\pi}$ and $\frac{\alpha_2}2-\psi_8\equiv\theta_0+\theta_4\pmod {2\pi}$, so that $\psi_8\equiv-\theta_4\pmod\pi$. Thus there are two cases:
\begin{itemize}
\item $\psi_8\equiv-\theta_4\pmod{2\pi}$ and $\frac{\alpha_2}2\equiv\theta_0\pmod{2\pi}$
\item $\psi_8\equiv-\theta_4+\pi\pmod{2\pi}$ and $\frac{\alpha_2}2\equiv\theta_0+\pi\pmod{2\pi}$
\end{itemize}
Let $\upsilon'_0\coloneqq\arg(u_{10})$ and $\upsilon'_1\coloneqq\arg(\frac{u_{01}}{u_{10}})$, so that $A=e^{i\upsilon'_0}XP(\upsilon'_1)C$. In the first case, \Cref{eqM0} gives $B=e^{-i(\upsilon'_0+\psi_8)}DP(-\upsilon'_1)X$, and one has
\begin{eqnarray*}
  &&\tikzfigS{./Preuve-Euler2d/Euler2dright5}\\[0.4cm]
  &\eqdeuxeqref{ctrl2pi}{mctrlPaddition}&\tikzfigS{./Preuve-Euler2d/Euler2dright5-11}\\[0.4cm]
  &\eqtroiseqref{mctrlP0}{mctrlPaddition}{mctrlPcommutgadget}&\tikzfigS{./Preuve-Euler2d/Euler2dright5-12}\\[0.4cm]
  &\overset{\eqref{mctrls0}\eqref{mctrlsaddition}\eqref{mctrlPopfullctrlRX}}{=}&\tikzfigS{./Preuve-Euler2d/Euler2dright5-13-RX}\\[0.4cm]
  &\overset{
\eqref{mctrlRXcommutCNOT}}{\overset{\text{Lemmas \ref{Palwayscommute}}}{\overset{\text{and \ref{PcommuteswithRXcontrols}}}{\overset{\text{\eqref{mctrlsaddition}}}{=}}}}
&\tikzfigS{./Preuve-Euler2d/Euler2dright5-14-RX}\\[0.4cm]
  &\overset{\eqrefJ{n-1},}{\overset{\text{Lemmas \ref{Palwayscommute}}}{\overset{\text{and \ref{PcommuteswithRXcontrols},}}{\overset{\text{\eqref{mctrlsaddition}\eqref{mctrlPaddition}}}{=}}}}&\tikzfigS{./Preuve-Euler2d/Euler2dright5-5-3}\\[0.4cm]
  &\overset{\eqrefJinv{n-1},}{\overset{\text{Lemmas \ref{Palwayscommute}}}{\overset{\text{and \ref{PcommuteswithRXcontrols},}}{\overset{\text{\eqref{mctrlsaddition}}}{=}}}}&\tikzfigS{./Preuve-Euler2d/Euler2dinductioncorr17}.
\end{eqnarray*}
The fact that we can get the desired parameters at the last step is justified by a similar argument to the two preceding cases (note in particular that $\interp{\minitikzfig{./Preuve-Euler2d/RXpi}}=e^{-i\frac\pi2}X$ and $\interp{\minitikzfig{./Preuve-Euler2d/RXminuspishort}}=e^{i\frac\pi2}X$).\medskip

In the second case, \Cref{eqM0} gives $B=e^{-i(\upsilon'_0+\psi_8)}DP(-\upsilon'_1+\pi)X$, and one has
\begin{eqnarray*}
  &&\tikzfigS{./Preuve-Euler2d/Euler2dright5}\\[0.4cm]
  &\eqdeuxeqref{ctrl2pi}{mctrlPaddition}&\tikzfigS{./Preuve-Euler2d/Euler2dright5-15}\\[0.4cm]
  &\eqtroiseqref{mctrlP0}{mctrlPaddition}{mctrlPcommutgadget}&\tikzfigS{./Preuve-Euler2d/Euler2dright5-16}\\[0.4cm]
  &\overset{\eqref{mctrlPaddition}\eqref{CNOTCNOT}}{\overset{\text{Lemmas \ref{Palwayscommute}}}{\overset{\text{and \ref{PcommuteswithRXcontrols}}}{\eqeqref{mctrlPcommutCNOT}}}}&\tikzfigS{./Preuve-Euler2d/Euler2dright5-17}\\[0.4cm]
  &\eqdeuxeqref{mctrlPcommutgadget}{mctrlPdef}&\tikzfigS{./Preuve-Euler2d/Euler2dright5-18}\\[0.4cm]
  &\eqeqref{ctrl2pi}&\tikzfigS{./Preuve-Euler2d/Euler2dright5-19}\\[0.4cm]
  &\overset{\eqref{mctrls0}\eqref{mctrlsaddition}\eqref{mctrlPopfullctrlRX}}{=}&\tikzfigS{./Preuve-Euler2d/Euler2dright5-20-RX}\\[0.4cm]
  &\overset{
\eqref{mctrlRXcommutCNOT}}{\overset{\text{Lemmas \ref{Palwayscommute}}}{\overset{\text{and \ref{PcommuteswithRXcontrols}}}{\overset{\text{\eqref{mctrlsaddition}}}{=}}}}&\tikzfigS{./Preuve-Euler2d/Euler2dright5-21-RX}\\[0.4cm]
  &\overset{\eqrefJ{n-1},}{\overset{\text{Lemmas \ref{Palwayscommute}}}{\overset{\text{and \ref{PcommuteswithRXcontrols},}}{\overset{\text{\eqref{mctrlsaddition}\eqref{mctrlPaddition}}}{=}}}}&\tikzfigS{./Preuve-Euler2d/Euler2dright5-5-4}\\[0.4cm]
  &\overset{\eqrefJinv{n-1}}{\overset{\text{Lemmas \ref{Palwayscommute}}}{\overset{\text{and \ref{PcommuteswithRXcontrols}}}{\overset{\text{\eqref{mctrlsaddition}}}{=}}}}&\tikzfigS{./Preuve-Euler2d/Euler2dinductioncorr17}
\end{eqnarray*}
The fact that we can get the desired parameters at the last step is justified by a similar argument to the preceding cases.
\end{itemize}

\subsection{Proof of Equation \eqref{Euler3dsansphases-multicontrolled}}\label{preuveEuler3dsansphases-multicontrolled}

\begin{eqnarray*}
  &&\tikzfigS{./Preuve-Euler-sans-phases/Euler3dsansphasesleft-multicontrolled}\\[0.4cm]
  &\overset{\text{\cref{swapcontrl}}}{\eqeqref{mctrlRXdef}}&\tikzfigS{./Preuve-Euler-sans-phases/Eulersansphases1}\\[0.4cm]
  &\eqeqref{mctrlRXphasegadget}&\tikzfigS{./Preuve-Euler-sans-phases/Eulersansphases2}\\[0.4cm]
  &\eqeqref{mctrlRXcommutCNOT}&\tikzfigS{./Preuve-Euler-sans-phases/Eulersansphases3}\\[0.4cm]
  &\eqdeuxeqref{HHCNOTHH}{HH}&\tikzfigS{./Preuve-Euler-sans-phases/Eulersansphases35}\\[0.4cm]
  &\overset{\text{\Cref{controlscommut}}}{\overset{\eqref{CNOTCNOT}}{=}}&\tikzfigS{./Preuve-Euler-sans-phases/Eulersansphases4}\\[0.4cm]
  &\overset{\text{\cref{PcommuteswithRXcontrols}}}{\overset{\eqref{mctrlsaddition}\eqref{mctrls0}}{=}}&\tikzfigS{./Preuve-Euler-sans-phases/Eulersansphases5}\\[0.4cm]
  &\eqeqref{mctrlPfromRX}&\tikzfigS{./Preuve-Euler-sans-phases/Eulersansphases7}\\[0.4cm]
  &\eqeqref{Euler2dmulticontrolled}&\tikzfigS{./Preuve-Euler-sans-phases/Eulersansphases8}\\[0.4cm]
  &\overset{\text{\Cref{Palwayscommute} and \ref{PcommuteswithRXcontrols}}}{\overset{\eqref{mctrlsaddition}\eqref{mctrls0}}{=}}&\tikzfigS{./Preuve-Euler-sans-phases/Eulersansphases10}\\[0.4cm]
  &\overset{\text{\Cref{Palwayscommute} and \ref{PcommuteswithRXcontrols}}}{=}&\tikzfigS{./Preuve-Euler-sans-phases/Eulersansphases11}\\[0.4cm]
  &\eqeqref{mctrlPfromRX}&\tikzfigS{./Preuve-Euler-sans-phases/Eulersansphases13}\\[0.4cm]
  &\overset{\text{\Cref{Palwayscommute} and \ref{PcommuteswithRXcontrols}}}{\overset{\eqref{mctrlsaddition}\eqref{mctrls0}}{=}}&\tikzfigS{./Preuve-Euler-sans-phases/Eulersansphases14}\\[0.4cm]
  &\eqdeuxeqref{HH}{HHCNOTHH}&\tikzfigS{./Preuve-Euler-sans-phases/Eulersansphases15}\\[0.4cm]
  &\eqeqref{CNOTCNOT}&\tikzfigS{./Preuve-Euler-sans-phases/Eulersansphases16}\\[0.4cm]
  &\eqdeuxeqref{HH}{HHCNOTHH}&\tikzfigS{./Preuve-Euler-sans-phases/Eulersansphases17}\\[0.4cm]
  &\eqeqref{mctrlRXcommutCNOT}&\tikzfigS{./Preuve-Euler-sans-phases/Eulersansphases18}\\[0.4cm]
  &\eqeqref{mctrlRXphasegadget}&\tikzfigS{./Preuve-Euler-sans-phases/Eulersansphases19}\\[0.4cm]
  &\overset{\text{\cref{swapcontrl}}}{\eqeqref{mctrlRXdef}}&\tikzfigS{./Preuve-Euler-sans-phases/Eulersansphases20}
\end{eqnarray*}

\subsection{Proof of Equation \eqref{euler3d}}\label{preuveeuler3d}

We actually derive the following equation, which is the original 
version of \Cref{euler3d} given in \cite{CHMPV}, with one more parameter (and similar conditions on the angles):
\begin{equation}\label{euler3doriginal}\tag{E$_{\textup{3D}}^{\text{original}}$}
  \tikzfigS{./qcoriginal-axioms/euler3d-left}\;=\;\tikzfigS{./qcoriginal-axioms/euler3doriginal-right}
\end{equation}

It has been proved in \cite{CDPV} (Appendix C.4) that the respective right-hand sides of \Cref{euler3d} and \cref{euler3doriginal} can be transformed into each other by using the axioms of $\QCold$ without \cref{euler3d} together with \Cref{ctrl2pi}. Hence, by \cref{lem:allAxiomExeptEuler3D}, \cref{euler3d} and \cref{euler3doriginal} are equivalent in the presence of the axioms of $\QC$. In other words, deriving \cref{euler3doriginal} will give us \cref{euler3d}.

Note additionally that in \cite{CHMPV,CDPV}, \cref{euler3doriginal,euler3d} are written with the muti-controlled $P$ gates defined with \cref{mctrlPfromRX} instead of \cref{mctrlPdef}. However, since we have proved in \cref{appendix:intermediateresults} that the two definitions are equivalent, it suffices to derive \cref{euler3doriginal} written with the current definition given by \cref{mctrlPdef}.\bigskip

If $\gamma_4\equiv \pi\pmod{4\pi}$, then one has
\begin{eqnarray*}
  \tikzfigS{./qcoriginal-axioms/euler3d-left}&\eqdeuxeqref{mctrlRX4pi}{mctrlRXaddition}&\tikzfigS{./proof-Mstar/Euleravecphasescaspi1}\\[0.4cm]
  &\overset{\eqref{mctrlP0}\eqref{mctrlPaddition}\eqref{XX}}{\overset{\text{\cref{blackwhitecontrolsindependent}}}{=}}&\tikzfigS{./proof-Mstar/Euleravecphasescaspi2}\\[0.4cm]
  &\overset{\text{Lemmas \ref{PcommuteswithRXcontrols}}}{\overset{\text{and \ref{blackwhitecontrolcommut}}}{=}}&\tikzfigS{./proof-Mstar/Euleravecphasescaspi2etdemi}\\[0.4cm]
  &\eqtroiseqref{mctrlRXcommutX}{mctrlPthroughRXpi}{XX}&\tikzfigS{./proof-Mstar/Euleravecphasescaspi3}\\[0.4cm]
  &\eqeqref{Euler3dsansphases-multicontrolled}&\tikzfigS{./proof-Mstar/Euleravecphasescaspi4}\\[0.4cm]
  &\eqeqref{mctrlP0}&\tikzfigS{./proof-Mstar/Euleravecphasescaspi5}
\end{eqnarray*}

If $\gamma_4\equiv 3\pi\pmod{4\pi}$, then one has
\begin{eqnarray*}
  \tikzfigS{./qcoriginal-axioms/euler3d-left}&\eqdeuxeqref{mctrlRX4pi}{mctrlRXaddition}&\tikzfigS{./proof-Mstar/Euleravecphasescasmoinspi1}\\[0.4cm]
  &\eqeqref{mctrlRXaddition}&\tikzfigS{./proof-Mstar/Euleravecphasescasmoinspi2}\\[0.4cm]
  &\eqeqref{mctrlRX2pi}&\tikzfigS{./proof-Mstar/Euleravecphasescasmoinspi3}\\[0.4cm]
  &\overset{\eqref{mctrlP0}\eqref{mctrlPaddition}\eqref{XX}}{\overset{\text{\cref{blackwhitecontrolsindependent}}}{=}}&\tikzfigS{./proof-Mstar/Euleravecphasescasmoinspi4}\\[0.4cm]
  &\overset{\text{Lemmas \ref{PcommuteswithRXcontrols}}}{\overset{\text{and \ref{blackwhitecontrolcommut}}}{=}}&\tikzfigS{./proof-Mstar/Euleravecphasescasmoinspi5}\\[0.4cm]
  &\eqtroiseqref{mctrlRXcommutX}{mctrlPthroughRXpi}{XX}&\tikzfigS{./proof-Mstar/Euleravecphasescasmoinspi6}\\[0.4cm]
  &\eqeqref{Euler3dsansphases-multicontrolled}&\tikzfigS{./proof-Mstar/Euleravecphasescasmoinspi7}\\[0.4cm]
  &\eqeqref{mctrlP0}&\tikzfigS{./proof-Mstar/Euleravecphasescasmoinspi8}
\end{eqnarray*}

If $\gamma_4\not\equiv\pi\pmod{2\pi}$, then with $\theta=2\arctan\!\left(\cos(\frac{\gamma_3}2)\tan(\frac{\gamma_4}2)\right)$,\smallskip\newline $\beta_2=2\arg\!\left(\cos(\frac{\gamma_4}2)\sqrt{1+\cos^2(\frac{\gamma_3}2)\tan^2(\frac{\gamma_4}2)}+i\sin(\frac{\gamma_3}2)\sin(\frac{\gamma_4}2)\right)$ and\smallskip\newline $\beta_3=2\arg\!\left(\sin(\frac{\gamma_3}2)-i\dfrac{\cos(\frac{\gamma_3}2)}{\cos(\frac{\gamma_4}2)}\right)$, one has

\begin{eqnarray*}
  \tikzfigS{./qcoriginal-axioms/euler3d-left}&\eqdeuxeqref{mctrlRX0}{mctrlRXaddition}&\tikzfigS{./proof-Mstar/Euleravecphases1}\\[0.4cm]
  &\overset{\text{\cref{mctrlPlift}}}{=}&\tikzfigS{./proof-Mstar/Euleravecphases2}\\[0.4cm]
  &\eqeqref{Euler3dsansphases-multicontrolled}&\tikzfigS{./proof-Mstar/Euleravecphases3}\\[0.4cm]
  &\eqeqref{Euler2dmulticontrolled}&\tikzfigS{./proof-Mstar/Euleravecphases4}\\[0.4cm]
  &\eqeqref{mctrlPthroughRXpi}&\tikzfigS{./proof-Mstar/Euleravecphases5}\\[0.4cm]
  &\overset{\text{\cref{blackwhitecontrolcommut}}}{=}&\tikzfigS{./proof-Mstar/Euleravecphases6}\\[0.4cm]
  &\eqeqref{Euler3dsansphases-multicontrolled}&\tikzfigS{./proof-Mstar/Euleravecphases7}\\[0.4cm]
  &=&\tikzfigS{./proof-Mstar/Euleravecphases8}\\[0.4cm]
  &\overset{\eqref{mctrlRXcommutX}}{\overset{\text{\cref{mctrlPlift}}}{=}}&\tikzfigS{./proof-Mstar/Euleravecphases8etdemi}\\[0.4cm]
  &\overset{\eqref{Euler2dmulticontrolled}}{\overset{\text{Lemmas \ref{Palwayscommute}}}{\overset{\text{and \ref{PcommuteswithRXcontrols}}}{=}}}&\tikzfigS{./proof-Mstar/Euleravecphases9}\\[0.4cm]
  &\overset{\eqref{XX}\eqref{mctrlPop}}{\overset{\text{Lemmas \ref{Palwayscommute}}}{\overset{\text{and \ref{PcommuteswithRXcontrols}}}{\overset{\text{\eqref{mctrlPaddition}}}{=}}}}&\tikzfigS{./proof-Mstar/Euleravecphases10}\\[0.4cm]
  &\eqdeuxeqref{mctrlRXcommutX}{XX}&\tikzfigS{./proof-Mstar/Euleravecphases11}\\[0.4cm]
  &\overset{\eqref{mctrlP0}}{\overset{\text{Lemmas \ref{mctrlPlift}}}{\overset{\text{and \ref{Palwayscommute}}}{=}}}&\tikzfigS{./proof-Mstar/Euleravecphases12}
\end{eqnarray*}

Finally, in all cases, 
it has been proved in \cite{CHMPV} (end of Appendix D-F) that the angles can be made to satisfy the conditions of \Cref{euler3doriginal} by using the axioms of the equational theory given in that paper, without using \Cref{euler3doriginal} except in the form of \Cref{ctrl2pi} (and the corollaries of \Cref{ctrl2pi}) and of \Cref{Euler2dmulticontrolled}. The other axioms of the equational theory of \cite{CHMPV} were proved in \cite{CDPV} to follow from the axioms of $\QCold$, and the derivations do not use \Cref{euler3d}. As we have already proved in \cref{lem:allAxiomExeptEuler3D} that the axioms of $\QCold$, except \Cref{euler3d}, follow from those of $\QC$, this concludes the proof.

\section{Minimality of $\QC$}\label{appendix:minimalityQC}

\subsection{Proof of Lemma \ref{lem:EulercardinalityofR}}\label{appendix:proofEulercardinalityofR}
\EulercardinalityofR*
\newcommand{\gPphione}{\scalebox{.69}{$\begin{tikzpicture}
	\begin{pgfonlayer}{nodelayer}
		\node [style=none] (0) at (-1.25, 0) {};
		\node [style=none] (1) at (1.5, 0) {};
		\node [style=gate] (2) at (0.125, 0) {$P(\varphi_1)$};
	\end{pgfonlayer}
	\begin{pgfonlayer}{edgelayer}
		\draw (0.center) to (1.center);
	\end{pgfonlayer}
\end{tikzpicture}
$}}
\newcommand{\gPphitwo}{\scalebox{.69}{$\begin{tikzpicture}
	\begin{pgfonlayer}{nodelayer}
		\node [style=none] (0) at (-1.25, 0) {};
		\node [style=none] (1) at (1.5, 0) {};
		\node [style=gate] (2) at (0.125, 0) {$P(\varphi_2)$};
	\end{pgfonlayer}
	\begin{pgfonlayer}{edgelayer}
		\draw (0.center) to (1.center);
	\end{pgfonlayer}
\end{tikzpicture}
$}}
\newcommand{\gPphithree}{\scalebox{.69}{$\begin{tikzpicture}
	\begin{pgfonlayer}{nodelayer}
		\node [style=none] (0) at (-1.25, 0) {};
		\node [style=none] (1) at (1.5, 0) {};
		\node [style=gate] (2) at (0.125, 0) {$P(\varphi_3)$};
	\end{pgfonlayer}
	\begin{pgfonlayer}{edgelayer}
		\draw (0.center) to (1.center);
	\end{pgfonlayer}
\end{tikzpicture}
$}}
\newcommand{\gPphiprimeone}{\scalebox{.69}{$\begin{tikzpicture}
	\begin{pgfonlayer}{nodelayer}
		\node [style=none] (0) at (-1.25, 0) {};
		\node [style=none] (1) at (1.5, 0) {};
		\node [style=gate] (2) at (0.125, 0) {$P(\varphi'_1)$};
	\end{pgfonlayer}
	\begin{pgfonlayer}{edgelayer}
		\draw (0.center) to (1.center);
	\end{pgfonlayer}
\end{tikzpicture}
$}}
\newcommand{\gPphiprimetwo}{\scalebox{.69}{$\begin{tikzpicture}
	\begin{pgfonlayer}{nodelayer}
		\node [style=none] (0) at (-1.25, 0) {};
		\node [style=none] (1) at (1.5, 0) {};
		\node [style=gate] (2) at (0.125, 0) {$P(\varphi'_2)$};
	\end{pgfonlayer}
	\begin{pgfonlayer}{edgelayer}
		\draw (0.center) to (1.center);
	\end{pgfonlayer}
\end{tikzpicture}
$}}
\newcommand{\gPpi}{\scalebox{.69}{$\begin{tikzpicture}
	\begin{pgfonlayer}{nodelayer}
		\node [style=none] (4) at (-1.25, 0) {};
		\node [style=none] (5) at (1.25, 0) {};
		\node [style=gate] (12) at (0, 0) {$P({\pi})$};
	\end{pgfonlayer}
	\begin{pgfonlayer}{edgelayer}
		\draw (4.center) to (5.center);
	\end{pgfonlayer}
\end{tikzpicture}
$}}
\newcommand{\simbot}{\overset\bot\sim}
\begin{proof}

Given two $P$ gates $\gPphione$ and $\gPphitwo$ in a $1$-qubit circuit, we write\footnote{Note that we are making an abuse of notation since there might be for instance several gates with parameter $\varphi_1$ in the circuit. This ambiguity can be eliminated by denoting the gates with abstract names such as $g_1,g_2$, however we prefer to keep the abuse of notation to improve the readability of the proof.} $\gPphione\sim\gPphitwo$, if the circuit is of the form $\tikzfigS{proof-indep-Euler/Pphi1andPphi2in1qubitcircuit}$ or $\tikzfigS{proof-indep-Euler/Pphi2andPphi1in1qubitcircuit}$, where $C_1,C_2$ and $C_3$ are arbitrary $1$-qubit circuits, and $\interp{C_2}$ is diagonal. We write $\gPphione\simbot\gPphitwo$, if the circuit is of the same form but with $\interp{C_2}$ anti-diagonal.

One can easily prove that the relation $\sim$ is transitive, and that given any three gates $\gPphione$, $\gPphitwo$ and $\gPphithree$ in a one-qubit circuit, if $\gPphione\simbot\gPphitwo$ and $\gPphitwo\simbot\gPphithree$ then $\gPphione\sim\gPphithree$. Moreover, we make $\sim$ into an equivalence relation by taking $\gP\sim\gP$ for any gate $\gP$ in a $1$-qubit circuit.

Given a $1$-qubit circuit $C$, a \emph{sign assignation} for $C$ is a function $S$ from the set of the $P$ gates of $C$ to $\{-1,1\}$, such that for any two gates $\gPphione$ and $\gPphitwo$ in $C$, if $\gPphione\sim\gPphitwo$ then $S(\gPphione)=S(\gPphitwo)$ and if $\gPphione\simbot\gPphitwo$ then $S(\gPphione)=-S(\gPphitwo)$.

Any $1$-qubit circuit admits at least one sign assignation. Indeed, one can partition the set of its $P$ gates into equivalence classes of $\sim$. Additionally, if two gates are such that $\gPphione\simbot\gPphitwo$, then for any elements $\gPphiprimeone$ and $\gPphiprimetwo$ of their respective equivalence classes, one also has $\gPphiprimeone\simbot\gPphiprimetwo$. Note also that since having $\gPphione\simbot\gPphitwo$ and $\gPphitwo\simbot\gPphithree$ implies $\gPphione\sim\gPphithree$, for each equivalence class there is at most one corresponding equivalence class whose elements are $\simbot$-related with its own elements. Thus, one can define a sign assignation by choosing a sign for each equivalence class, in such a way that for every pair of $\simbot$-related equivalence classes, the two classes are given opposite signs.\bigskip

Given a $1$-qubit circuit $C$ and a sign assignation $S$ for $C$, we define\footnote{The sum is over the $P$ gates of $C$, therefore (despite the abuse of notation) two gates with same parameter contribute separately to the sum.}
\[\interpEq{C}{E}^S\coloneqq\left(\sum_{\gP\text{ gate of $C$}}S(\gP)\varphi\right)\bmod\frac\pi2.\]

Note that since \Cref{euler} acts on 1 qubit, it suffices to prove that it has an instance which is not a consequence of the equations of $A$ acting on at most 1 qubit, namely \crefnosortnocompress{gphaseempty,gphaseaddition,HH,P0,eulerH,Paddition,XPX}, and the chosen instances of \Cref{eulerstar}.

Let
\[\Phi_A\coloneqq\left\{s'_1\beta_1+s'_2\beta_2+s'_3\beta_3-s_1\alpha_1-s_2\alpha_2-s_3\alpha_3\middle|\substack{\displaystyle(\tikzfigS{qc-axioms/euler-left}=\tikzfigS{qc-axioms/euler-right})\in A,\medskip\\\displaystyle s_1,s_2,s_3,s'_1,s'_2,s'_3\in\{-1,1\}}\right\},\]
and let $\langle \Phi_A\rangle$ be the additive subgroup of $\R$ generated by $\Phi_A\cup\{\frac\pi2\}$. Note that the cardinality of $\langle \Phi_A\rangle$ is $\max(\aleph_0,\mathrm{card}(\Phi_A))$, and that since the set of instances of \Cref{eulerstar} belonging to $A$ has cardinality strictly less than $2^{\aleph_0}$, the cardinality of $\Phi_A$, and therefore that of $\langle \Phi_A\rangle$, is strictly less than $2^{\aleph_0}$.\bigskip

First, given any two circuits $C_1$ and $C_2$ that can be transformed into each other using \crefnosortnocompress{gphaseempty,gphaseaddition,HH,P0,eulerH,Paddition,XPX}, and any sign assignation $S_1$ for $C_1$, there exists a sign assignation $S_2$ for $C_2$ such that $\interpEq{C_1}{E}^{S_1}=\interpEq{C_2}{E}^{S_2}$. Indeed, to prove this, it suffices to consider the case where one transforms $C_1$ into $C_2$ by a single step of rewriting using one of these equations. If the equation used is Equation~\eqref{gphaseempty}, \eqref{gphaseaddition} or \eqref{HH}, then it suffices to take $S_2$ as assigning the same signs as $S_1$. If the equation used is \Cref{P0} from left to right, then it suffices to restrict $S_1$ to the remaining gates. If the equation used is \Cref{P0} from left to right, then it suffices to extend $S_1$ by assigning any allowed sign to the new $0$-angled $P$ gate.  If the equation used is \Cref{eulerH} from left to right, then any extension of $S_1$ which assigns allowed signs to the new $\frac\pi2$-angled $P$ gates works, since the sum is taken modulo $\frac\pi2$ in the definition of $\interpEq{\cdot}{E}^{S}$. For the same reason, if the equation used is \Cref{eulerH} from right to left, the restriction of $S_1$ to the remaining gates works. If the equation used is \Cref{Paddition} from left to right, then in $C_1$, the two concerned gates are $\sim$-related, therefore they have the same sign according to $S_1$; then taking $S_2$ as assigning this sign to the new gate (and being identical to $S_1$ elsewhere) works. If the equation used is \Cref{Paddition} from right to left, then taking $S_2$ as assigning the sign of the replaced gate of $C_1$ to both new gates in $C_2$ (and being identical to $S_1$ elsewhere) works. If the equation used is \Cref{XPX}, from left to right, then because of the semantics of the $X$ gates, the new gate $\gPminus$ in $C_2$ is $\simbot$-related to any gate to which the gate $\gP$ was $\sim$-related in $C_1$, and vice-versa; therefore, taking $S_2$ as assigning to $\gPminus$ the opposite sign to the one assigned by $S_1$ to $\gP$, and being a restriction of $S_1$ elsewhere, works (note that whichever sign they were assigned, the $\gPpi$ gates contained in the $X$ gates --~see \Cref{Xdef,Zdef}~-- did not contribute to $\interpEq{C_1}{E}^{S_1}$ as the sum is taken modulo $\frac\pi2$ in the definition). If the equation applied is \Cref{XPX} from right to left, then similarly, taking $S_2$ such that $\gP$ in $C_2$ has an opposite sign to $\gPminus$ in $C_1$ works (with any arbitrary allowed sign assigned to the new $\gPpi$ gates contained in the $X$ gates, and $S_2$ being identical to $S_1$ elsewhere).

Second, given two circuits $C_1$ and $C_2$ such that $C_2$ is obtained from $C_1$ by a single step of rewriting using an instance of \Cref{eulerstar} belonging to $A$ (either from left to right or from right to left), and given a sign assignation $S_1$ for $C_1$, by defining a sign assignation $S_2$ for $C_2$ as assigning arbitrary allowed signs to the modified $P$ gates in $C_2$, and being identical to $S_1$ elsewhere, one gets $\interpEq{C_2}{E}^{S_2}-\interpEq{C_1}{E}^{S_1}+k\frac\pi2\in\Phi_A$ for some $k\in\Z$. It follows that for any two circuits $C_1$ and $C_2$ that can be transformed into each other using the equations of $A$, and any sign assignation $S_1$ for $C_1$, there exists a sign assignation $S_2$ for $C_2$ such that $\interpEq{C_2}{E}^{S_2}-\interpEq{C_1}{E}^{S_1}\in\langle \Phi_A\rangle$.\bigskip

One can show (see below) that in \cref{euler}, for any $\alpha_1,\alpha_2,\alpha_3\not\equiv0\pmod\pi$, one has\footnote{Where $\cot$ denotes the cotangent function.}
\begin{align*}
&\beta_1=-\arctan\left(\cos(\alpha_1)\cot(\alpha_2)+\frac{\sin(\alpha_1)\cot(\alpha_3)}{\sin(\alpha_2)}\right)+\frac\pi2+k_1\pi\\[0.4cm]
&\beta_2=\arccos\bigl(\cos(\alpha_1)\cos(\alpha_3)-\sin(\alpha_1)\cos(\alpha_2)\sin(\alpha_3)\bigr)\\[0.4cm]
&\beta_3=-\arctan\left(\cos(\alpha_3)\cot(\alpha_2)+\frac{\sin(\alpha_3)\cot(\alpha_1)}{\sin(\alpha_2)}\right)+\frac\pi2+k_3\pi
\end{align*}
with $k_1,k_3\in\Z$. Let $b_1,b_2,b_3\colon\mathcal (\R^3\setminus(\pi\Z{\times}\R{\times}\R\;\cup\;\R{\times}\pi\Z{\times}\R\;\cup\;\R{\times}\R{\times}\pi\Z))\to\R$ be defined as
\begin{align*}
&b_1(\alpha_1,\alpha_2,\alpha_3)=-\arctan\left(\cos(\alpha_1)\cot(\alpha_2)+\frac{\sin(\alpha_1)\cot(\alpha_3)}{\sin(\alpha_2)}\right)\\[0.4cm]
&b_2(\alpha_1,\alpha_2,\alpha_3)=\arccos\bigl(\cos(\alpha_1)\cos(\alpha_3)-\sin(\alpha_1)\cos(\alpha_2)\sin(\alpha_3)\bigr)\\[0.4cm]
&b_3(\alpha_1,\alpha_2,\alpha_3)=-\arctan\left(\cos(\alpha_3)\cot(\alpha_2)+\frac{\sin(\alpha_3)\cot(\alpha_1)}{\sin(\alpha_2)}\right).
\end{align*}
The functions $b_1$, $b_2$ and $b_3$ are of class $C^\infty$, and for instance, one has, for any $\alpha_1,\alpha_2,\alpha_3$:
\begin{align*}
&\frac{\partial b_1}{\partial\alpha_2}(\alpha_1,\alpha_2,\alpha_3)=\frac{\cos(\alpha_1)+\sin(\alpha_1)\cos(\alpha_2)\cot(\alpha_3)}{\sin(\alpha_2)^2+\left(\cos(\alpha_1)\cos(\alpha_2)+\sin(\alpha_1)\cot(\alpha_3)\right)^2}\\[0.4cm]
&\frac{\partial b_2}{\partial\alpha_2}(\alpha_1,\alpha_2,\alpha_3)=-\frac{\sin(\alpha_1)\sin(\alpha_2)\sin(\alpha_3)}{\sqrt{1-\left(\cos(\alpha_1)\cos(\alpha_3)-\sin(\alpha_1)\cos(\alpha_2)\sin(\alpha_3)\right)^2}}\\[0.4cm]
&\frac{\partial b_3}{\partial\alpha_2}(\alpha_1,\alpha_2,\alpha_3)=\frac{\cos(\alpha_3)+\sin(\alpha_3)\cos(\alpha_2)\cot(\alpha_1)}{\sin(\alpha_2)^2+\left(\cos(\alpha_3)\cos(\alpha_2)+\sin(\alpha_3)\cot(\alpha_1)\right)^2}.
\end{align*}
By taking for instance $\alpha_1=\alpha_2=\alpha_3=\frac\pi4$, one has $\dfrac{\partial b_1}{\partial\alpha_2}(\frac\pi4,\frac\pi4,\frac\pi4)=\dfrac{\partial b_3}{\partial\alpha_2}(\frac\pi4,\frac\pi4,\frac\pi4)=\dfrac{2+6\sqrt2}{17}\approx0.6168$ and $\dfrac{\partial b_2}{\partial\alpha_2}(\frac\pi4,\frac\pi4,\frac\pi4)=-\dfrac1{\sqrt{5+2\sqrt2}}\approx-0.3574$.

Given $s_1,s_2,s_3,s'_1,s'_2,s'_3\in\{-1,1\}$, let $F_{s_i,s'_i}\colon\mathcal (\R^3\setminus(\pi\Z{\times}\R{\times}\R\;\cup\;\R{\times}\pi\Z{\times}\R\;\cup\;\R{\times}\R{\times}\pi\Z))\to\R$ be defined as $F_{s_i,s'_i}(\alpha_1,\alpha_2,\alpha_3)=s'_1b_1(\alpha_1,\alpha_2,\alpha_3)+s'_2b_2(\alpha_1,\alpha_2,\alpha_3)+s'_3b_3(\alpha_1,\alpha_2,\alpha_3)-s_1\alpha_1-s_2\alpha_2-s_3\alpha_3$. One can check that for any choice of the $s_i$ and $s'_i$, $\dfrac{\partial F_{s_i,s'_i}}{\partial\alpha_2}(\frac\pi4,\frac\pi4,\frac\pi4)\neq0$\smallskip\smallskipbeforeline, so that all the $F_{s_i,s'_i}$ are injective on a small sub-domain of the form $\{\frac\pi4\}\times[\frac\pi4-\varepsilon,\frac\pi4+\varepsilon]\times\{\frac\pi4\}$ with $\varepsilon>0$. This sub-domain has cardinality $2^{\aleph_0}$. Additionally, for each choice of the $s_i$ and $s'_i$, since $F_{s_i,s'_i}$ is injective on this sub-domain, the set $F_{s_i,s'_i}^{-1}(\langle \Phi_A\rangle)\cap\left(\{\frac\pi4\}\times[\frac\pi4-\varepsilon,\frac\pi4+\varepsilon]\times\{\frac\pi4\}\right)$ is at most of the cardinality of $\langle \Phi_A\rangle$, which is strictly less than $2^{\aleph_0}$. This implies that the union $\displaystyle\bigcup_{s_i,s'_i\in\{-1,1\}}F_{s_i,s'_i}^{-1}(\langle \Phi_A\rangle)\cap\left(\{\tfrac\pi4\}\times[\tfrac\pi4-\varepsilon,\tfrac\pi4+\varepsilon]\times\{\tfrac\pi4\}\right)$ also has cardinality strictly less than $2^{\aleph_0}$. Hence, their exists $\alpha_2\in[\frac\pi4-\varepsilon,\frac\pi4+\varepsilon]$ such that for any choice of the $s_i$ and $s'_i$, $F_{s_i,s'_i}(\frac\pi4,\alpha_2,\frac\pi4)\notin\langle \Phi_A\rangle$. Then 
if we consider the corresponding instance of \Cref{euler}, $\tikzfigS{proof-indep-Euler/euler-left-pi4alpha2pi4}=\tikzfigS{qc-axioms/euler-right}$\smallskip, there does not exist sign assignations $S_1$ and $S_2$ such that $\interpEq{\tikzfigS{qc-axioms/euler-right}}{E}^{S_2}-\interpEq{\tikzfigS{proof-indep-Euler/euler-left-pi4alpha2pi4}}{E}^{S_1}\in\langle \Phi_A\rangle$, as this would give us a choice of the $s_i$ and $s'_i$ such that $F_{s_i,s'_i}(\frac\pi4,\alpha_2,\frac\pi4)\in\langle \Phi_A\rangle$. This proves that this instance of \Cref{euler} is not a consequence of the axioms of $A$.\bigskip\bigskip

It remains to prove the formulas given above for $\beta_1$, $\beta_2$ and $\beta_3$ in \cref{euler}, using the expressions given in \cref{sec:discussionEuler}. First, note that if $\alpha_1,\alpha_2,\alpha_3\not\equiv0\pmod\pi$ then $z,z'\neq0$. For $i\in\{1,2,3\}$, we denote $c_i\coloneqq\cos(\frac{\alpha_i}2)$ and $s_i\coloneqq\sin(\frac{\alpha_i}2)$.

One has $\beta_1=\arg(zz')$, with
\begin{align*}
\Re(zz')&=\cos^2\left(\frac{\alpha_2}2\right)\cos\left(\frac{\alpha_1+\alpha_3}2\right)\sin\left(\frac{\alpha_1+\alpha_3}2\right)+\sin^2\left(\frac{\alpha_2}2\right)\cos\left(\frac{\alpha_1-\alpha_3}2\right)\sin\left(\frac{\alpha_1-\alpha_3}2\right)\\[0.4cm]
&=\frac12\bigl(c_2^2\sin(\alpha_1+\alpha_3)+s_2^2\sin(\alpha_1-\alpha_3)\bigr)\\[0.4cm]
&=\frac12\bigl(c_2^2(\sin(\alpha_1)\cos(\alpha_3)+\cos(\alpha_1)\sin(\alpha_3))+s_2^2(\sin(\alpha_1)\cos(\alpha_3)-\cos(\alpha_1)\sin(\alpha_3))\bigr)\\[0.4cm]
&=\frac12\bigl(\sin(\alpha_1)\cos(\alpha_3)(c_2^2+s_2^2)+\cos(\alpha_1)\sin(\alpha_3)(c_2^2-s_2^2)\bigr)\\[0.4cm]
&=\frac12\bigl(\sin(\alpha_1)\cos(\alpha_3)+\cos(\alpha_1)\cos(\alpha_2)\sin(\alpha_3)\bigr)
\end{align*}
and
\begin{align*}
\Im(zz')&=\cos\left(\frac{\alpha_2}2\right)\sin\left(\frac{\alpha_2}2\right)\left(\cos\left(\frac{\alpha_1-\alpha_3}2\right)\sin\left(\frac{\alpha_1+\alpha_3}2\right)-\cos\left(\frac{\alpha_1+\alpha_3}2\right)\sin\left(\frac{\alpha_1-\alpha_3}2\right)\right)\\[0.4cm]
&=c_2s_2\bigl((c_1c_3+s_1s_3)(s_1c_3+c_1s_3)-(c_1c_3-s_1s_3)(s_1c_3-c_1s_3)\bigr)\\[0.4cm]
&=2c_2s_2\bigl(c_1^2c_3s_3+s_1^2c_3s_3\bigr)\\[0.4cm]
&=2c_2s_2c_3s_3\\[0.4cm]
&=\frac12\sin(\alpha_2)\sin(\alpha_3).
\end{align*}
If $\alpha_1,\alpha_2,\alpha_3\not\equiv0\pmod\pi$, then this implies that $\Im(zz')\neq0$, so that $\beta_1\equiv\frac\pi2-\arctan\left(\frac{\Re(zz')}{\Im(zz')}\right)\pmod\pi$, which gives the desired formula.

To prove the analogous formula for $\beta_3$, it suffices to remark that by exchanging $\alpha_1$ and $\alpha_3$ in the expressions given in \cref{sec:discussionEuler}, one computes $\beta_3$ instead of $\beta_1$ (and vice-versa).

One has
\begin{align*}
&\beta_2=2\arg\left(i+\left\lvert\frac{z}{z'}\right\rvert\right)=\arg\left(\left(i+\left\lvert\frac{z}{z'}\right\rvert\right)^2\right)=\arg\left(\left\lvert\frac{z}{z'}\right\rvert^2\!-1+2i\left\lvert\frac{z}{z'}\right\rvert\right).
\end{align*}
Since the imaginary part of $\left\lvert\frac{z}{z'}\right\rvert^2-1+2i\left\lvert\frac{z}{z'}\right\rvert$ is positive, its argument is equal to
\begin{align*}
\arccos\left(\frac{\left\lvert\frac{z}{z'}\right\rvert^2-1}{\left\lvert\left\lvert\frac{z}{z'}\right\rvert^2-1+2i\left\lvert\frac{z}{z'}\right\rvert\right\rvert}\right)&=\arccos\left(\frac{\left\lvert\frac{z}{z'}\right\rvert^2-1}{\sqrt{\left\lvert\frac{z}{z'}\right\rvert^4+1-2\left\lvert\frac{z}{z'}\right\rvert^2+4\left\lvert\frac{z}{z'}\right\rvert^2}}\right)\\[0.4cm]
&=\arccos\left(\frac{\left\lvert\frac{z}{z'}\right\rvert^2-1}{\left\lvert\frac{z}{z'}\right\rvert^2+1}\right)\\[0.4cm]
&=\arccos\left(\frac{\left\lvert z\right\rvert^2-\left\lvert z'\right\rvert^2}{\left\lvert z\right\rvert^2+\left\lvert z'\right\rvert^2}\right).
\end{align*}
One has $\left\lvert z\right\rvert^2=\cos^2\left(\frac{\alpha_2}{2}\right)\cos^2\left(\frac{\alpha_1+\alpha_3}{2}\right)+\sin^2\left(\frac{\alpha_2}{2}\right)\cos^2\left(\frac{\alpha_1-\alpha_3}{2}\right)$ and $\binoppenalty10000\relpenalty10000\left\lvert z'\right\rvert^2=\cos^2\left(\frac{\alpha_2}{2}\right)\sin^2\left(\frac{\alpha_1+\alpha_3}{2}\right)+\sin^2\left(\frac{\alpha_2}{2}\right)\sin^2\left(\frac{\alpha_1-\alpha_3}{2}\right)$, so that \smallskipbeforeline$\left\lvert z\right\rvert^2+\left\lvert z'\right\rvert^2=1$ and
\begin{align*}
\left\lvert z\right\rvert^2-\left\lvert z'\right\rvert^2&=c_2^2\left(\cos^2\left(\frac{\alpha_1+\alpha_3}{2}\right)-\sin^2\left(\frac{\alpha_1+\alpha_3}{2}\right)\right)+s_2^2\left(\cos^2\left(\frac{\alpha_1-\alpha_3}{2}\right)-\sin^2\left(\frac{\alpha_1-\alpha_3}{2}\right)\right)\\[0.4cm]
&=c_2^2\cos(\alpha_1+\alpha_3)+s_2^2\cos(\alpha_1-\alpha_3)\\[0.4cm]
&=c_2^2(\cos(\alpha_1)\cos(\alpha_3)-\sin(\alpha_1)\sin(\alpha_3))+s_2^2(\cos(\alpha_1)\cos(\alpha_3)+\sin(\alpha_1)\sin(\alpha_3))\\[0.4cm]
&=\cos(\alpha_1)\cos(\alpha_3)(c_2^2+s_2^2)-\sin(\alpha_1)\sin(\alpha_3)(c_2^2-s_2^2)\\[0.4cm]
&=\cos(\alpha_1)\cos(\alpha_3)-\sin(\alpha_1)\cos(\alpha_2)\sin(\alpha_3)
\end{align*}
which gives the desired expression.
\end{proof}

\subsection{Necessity of the remaining axioms}\label{appendix:necessityproofs}
\begin{proposition}
  \Cref{gphaseempty} cannot be removed from $\QC$ without loosing completeness.
\end{proposition}
\begin{proof}
  \Cref{gphaseempty} is the only axiom that is not sound according to the alternative interpretation $\interpEq{\cdot}{S_{2\pi}}$. Intuitively, \Cref{gphaseempty} is the only rule acting on $0$ qubits that transforms a circuit containing some global phase gates into one that contains none.
  \begin{gather*}
    \interpEq{C_2\circ C_1}{S_{2\pi}} = max(\interpEq{C_2}{S_{2\pi}},\interpEq{C_1}{S_{2\pi}})\\
    \interpEq{C_1\otimes C_2}{S_{2\pi}} = max(\interpEq{C_1}{S_{2\pi}},\interpEq{C_2}{S_{2\pi}})\\
    \interpEq{\gs}{S_{2\pi}}=1 \qquad\textup{and}\qquad \interpEq{g}{S_{2\pi}}=0 \textup{ for any other generator }g
  \end{gather*}
\end{proof}

\begin{proposition}
  \Cref{gphaseaddition} cannot be removed from $\QC$ without loosing completeness.
\end{proposition}
\begin{proof}
  For any $\psi\in\R\setminus\{2\pi\}$, the instances of \Cref{gphaseaddition} where exactly one of the two sides contains a global phase gate with parameter $\psi$ are the only ones that are not sound according to the alternative interpretation $\interpEq{\cdot}{S_+}^\psi$. Note that such instances exist for any $\psi$, for instance $\tikzfigS{./examples/gphasepsipluspisur4}$. Thus at least one such instance for every possible $\psi$ needs to be taken as an axiom. Intuitively, the instances of \Cref{gphaseaddition} with a parameter $\psi$ on exactly one side are the only rules acting on $0$ qubits that transform a circuit containing some global phase gates with parameter $\psi$ into a circuit that contains none.
  \begin{gather*}
    \interpEq{C_2\circ C_1}{S_+}^\psi = max(\interpEq{C_2}{S_+}^\psi,\interpEq{C_1}{S_+}^\psi)\\
    \interpEq{C_1\otimes C_2}{S_+}^\psi = max(\interpEq{C_1}{S_+}^\psi,\interpEq{C_2}{S_+}^\psi)\\
    \interpEq{\gs}{S_+}^\psi=\begin{cases}1&\text{if $\varphi=\psi$}\\0&\text{if $\varphi\neq\psi$}\end{cases} \qquad\textup{and}\qquad \interpEq{g}{S_+}^\psi=0 \textup{ for any other generator }g
  \end{gather*}
\end{proof}

\begin{proposition}
  \Cref{HH} cannot be removed from $\QC$ without loosing completeness.
\end{proposition}
\begin{proof}
  \Cref{HH} is the only axiom that is not sound according to the alternative interpretation $\interpEq{\cdot}{H^2}$. Intuitively, \Cref{HH} is the only rule acting on at most $1$ qubit that transforms a circuit containing some Hadamard gates into one that contains none.
  \begin{gather*}
    \interpEq{C_2\circ C_1}{H^2} = max(\interpEq{C_2}{H^2},\interpEq{C_1}{HH})\\
    \interpEq{C_1\otimes C_2}{H^2} = max(\interpEq{C_1}{H^2},\interpEq{C_2}{HH})\\
    \interpEq{\gH}{H^2}=1 \qquad\textup{and}\qquad \interpEq{g}{H^2}=0 \textup{ for any other generator }g
  \end{gather*}
\end{proof}

\begin{proposition}
  \Cref{P0} cannot be removed from $\QC$ without loosing completeness.
\end{proposition}
\begin{proof}
  \Cref{P0} is the only axiom that is not sound according to the alternative interpretation $\interpEq{\cdot}{P_0}$. Intuitively, \Cref{P0} is the only rule acting on at most $1$ qubit that does not preserve the parity of number of $1$-qubit generators.
  \begin{gather*}
    \interpEq{C_2\circ C_1}{P_0} = \interpEq{C_1}{P_0}+\interpEq{C_2}{P_0}\bmod{2}\\
    \interpEq{C_1\otimes C_2}{P_0} = \interpEq{C_1}{P_0}+\interpEq{C_2}{P_0}\bmod{2}\\
    \interpEq{\gH}{P_0}=\interpEq{\gP}{P_0} = 1 \qquad\textup{and}\qquad \interpEq{g}{P_0}=0 \textup{ for any other generator }g
  \end{gather*}
\end{proof}

\begin{proposition}
  \Cref{CNOTPCNOT} cannot be removed from $\QC$ without loosing completeness.
\end{proposition}
\begin{proof}
  \Cref{CNOTPCNOT} is the only axiom that is not sound according to the alternative interpretation $\interpEq{\cdot}{CX^2}$. Intuitively, \Cref{CNOTPCNOT} is the only rule acting on at most $2$ qubits that transforms a circuit containing some CNot gates into one that contains none.
  \begin{gather*}
    \interpEq{C_2\circ C_1}{CX^2} = max(\interpEq{C_2}{CX^2},\interpEq{C_1}{CX^2})\\
    \interpEq{C_1\otimes C_2}{CX^2} = max(\interpEq{C_1}{CX^2},\interpEq{C_2}{CX^2})\\
    \interpEq{\gCNOT}{CX^2}=1 \qquad\textup{and}\qquad \interpEq{g}{CX^2}=0 \textup{ for any other generator }g
  \end{gather*}
\end{proof}

\begin{proposition}
  \Cref{bigebre} cannot be removed from $\QC$ without loosing completeness.
\end{proposition}
\begin{proof}
  \Cref{bigebre} is the only axiom that is not sound according to the following alternative interpretation $\interpEq{\cdot}{B}$. Intuitively, \Cref{bigebre} is the only rule that does not preserve the parity of number of swap gates. Note that, as quantum circuits have as much output qubits as input qubits, all the deformation rules of the prop formalism preserve the parity of swap gates.
  \begin{gather*}
    \interpEq{C_2\circ C_1}{B} = \interpEq{C_2}{B}\circ\interpEq{C_1}{B}\\
    \interpEq{C_1\otimes C_2}{B} = \interpEq{C_1}{B}\otimes\interpEq{C_2}{B}\\
    \interpEq{\gempty}{B}=\interpEq{\gs}{B}=\interp{\gempty}, \qquad \interpEq{\gI}{B}=\interpEq{\gH}{B}=\interpEq{\gP}{B}=\interp{\gI}, \qquad \interpEq{\gCNOT}{B}=\interp{\gII}, \qquad \interpEq{\gSWAP}{B}=\interp{\gSWAP}
  \end{gather*}
\end{proof}

\begin{proposition}
  \Cref{CZ} cannot be removed from $\QC$ without loosing completeness.
\end{proposition}
\begin{proof}
  The following alternative interpretation $\interpEq{\cdot}{CZ}$ excludes \Cref{CZ}. Intuitively, \Cref{CZ} is the only rule that does not preserve the parity of number of cnot gates plus the parity of swap gates.
  \begin{gather*}
    \interpEq{C_2\circ C_1}{CZ} = \interpEq{C_2}{CZ}\circ\interpEq{C_1}{CZ}\\
    \interpEq{C_1\otimes C_2}{CZ} = \interpEq{C_1}{CZ}\otimes\interpEq{C_2}{CZ}\\
    \interpEq{\gempty}{CZ}=\interpEq{\gs}{CZ}=\interp{\gempty}, \qquad \interpEq{\gI}{CZ}=\interpEq{\gH}{CZ}=\interpEq{\gP}{CZ}=\interp{\gI}, \qquad \interpEq{\gCNOT}{CZ}=\interp{\gCNOT}, \qquad \interpEq{\gSWAP}{CZ}=\interp{\gSWAP}
  \end{gather*}
\end{proof}

\begin{proposition}
  \Cref{eulerH} cannot be removed from $\QC$ without loosing completeness.
\end{proposition}
\begin{proof}
  \Cref{eulerH} is the only axiom that is not sound according to the alternative interpretation $\interpEq{\cdot}{E_H}$. Intuitively, \Cref{eulerH} is the only rule that does not preserve the parity of number of Hadamard gates.
  \begin{gather*}
    \interpEq{C_2\circ C_1}{E_H} = \interpEq{C_2}{E_H}+\interpEq{C_1}{E_H}\bmod 2\\
    \interpEq{C_1\otimes C_2}{E_H} = \interpEq{C_1}{E_H}+\interpEq{C_2}{E_H}\bmod 2\\
    \interpEq{\gH}{E_H}=1 \qquad\textup{and}\qquad \interpEq{g}{E_H}=0 \textup{ for any other generator }g
  \end{gather*}
\end{proof}

\section{Alternative equational theory $\QCprime$}
\subsection{Proofs of Equations \eqref{eulerH} and \eqref{euler}}\label{appendix:completenessQCprime}

\begin{proof}[Proof of \Cref{eulerH} in $\QCprime$]
  \begin{gather*}
    \tikzfigS{./qcprime-completeness/Heuler-step-0}
    \eqquatreeqref{HH}{P0}{gphaseempty}{gphaseaddition}\tikzfigS{./qcprime-completeness/Heuler-step-1}
    \eqeqref{eulerprime}\tikzfigS{./qcprime-completeness/Heuler-step-2}
    \eqdeuxeqref{gphaseempty}{gphaseaddition}\tikzfigS{./qcprime-completeness/Heuler-step-3}
  \end{gather*}
\end{proof}

\begin{proof}[Proof of \Cref{P2pi} in $\QCprime$]
  \begin{eqnarray*}
    \tikzfigS{./qcprime-completeness/P2pi_00}
    &\eqquatreeqref{HH}{P0}{gphaseempty}{gphaseaddition}&\tikzfigS{./qcprime-completeness/P2pi_01}\\
    &\eqtroiseqref{RXdef}{gphaseempty}{gphaseaddition}&\tikzfigS{./qcprime-completeness/P2pi_02}\\
    &\eqeqref{eulerprime}&\tikzfigS{./qcprime-completeness/P2pi_03}\\
    &\eqquatreeqref{eulerH}{HH}{gphaseaddition}{gphaseempty}&\tikzfigS{./qcprime-completeness/P2pi_04}
  \end{eqnarray*}
\end{proof}

\begin{proof}[Proof of \Cref{XPX} in $\QCprime$]
  \begin{eqnarray*}
    \tikzfigS{./qcprime-completeness/XPX_00}
    &\eqdeuxeqref{Xdef}{Zdef}&\tikzfigS{./qcprime-completeness/XPX_01}\\
    &\eqeqref{HH}&\tikzfigS{./qcprime-completeness/XPX_02}\\
    &\eqtroiseqref{RXdef}{gphaseempty}{gphaseaddition}&\tikzfigS{./qcprime-completeness/XPX_03}\\
    &\eqeqref{eulerprime}&\tikzfigS{./qcprime-completeness/XPX_045}\\
    &\eqquatreeqref{gphaseaddition}{gphaseempty}{Paddition}{P2pi}&\tikzfigS{./qcprime-completeness/XPX_04}\\
    &\eqtroiseqref{gphaseaddition}{Paddition-prime}{P2pi}&\tikzfigS{./qcprime-completeness/XPX_05}\\
    &\eqdeuxeqref{eulerH}{HH}&\tikzfigS{./qcprime-completeness/XPX_06}
  \end{eqnarray*}
\end{proof}

\begin{proof}[Proof of \Cref{ZRZ} in $\QCprime$]
  \begin{eqnarray*}
    \tikzfigS{./qcprime-completeness/ZRZ_00}
    \eqeqref{HH}\tikzfigS{./qcprime-completeness/ZRZ_01}
    \eqeqref{RXdef}\tikzfigS{./qcprime-completeness/ZRZ_02}
    \eqdeuxeqref{XPX}{gphaseaddition}\tikzfigS{./qcprime-completeness/ZRZ_03}
    \eqeqref{RXdef}\tikzfigS{./qcprime-completeness/ZRZ_04}
  \end{eqnarray*}
\end{proof}

\begin{proof}[Proof of \Cref{euler} in $\QCprime$]
  \begin{eqnarray*}
    \tikzfigS{./qcprime-completeness/euler_00}
    &\eqeqref{Paddition}&\tikzfigS{./qcprime-completeness/euler_01}\\
    &\eqeqref{HH}&\tikzfigS{./qcprime-completeness/euler_02}\\
    &\eqtroiseqref{RXdef}{gphaseempty}{gphaseaddition}&\tikzfigS{./qcprime-completeness/euler_03}\\
    &\eqdeuxeqref{eulerprime}{gphaseaddition}&\tikzfigS{./qcprime-completeness/euler_04}\\
    &\eqeqref{Paddition}&\tikzfigS{./qcprime-completeness/euler_05}
  \end{eqnarray*}

  Considering $x$ as a variable, we choose $x$ such that $\gamma_3(x)+\delta_1(x)=0\pmod{\pi}$. Notice that the existence of such an $x$ is proved in \cite{vilmart2018nearoptimal}.

  If $\gamma_3+\delta_1=0\pmod{2\pi}$, then we can assume w.l.o.g. that $\gamma_3+\delta_1=0$ (otherwise we can apply \Cref{P2pi,Paddition}).
  \begin{eqnarray*}
    \tikzfigS{./qcprime-completeness/euler_00}
    &=&\tikzfigS{./qcprime-completeness/euler_case0_00}\\
    &\eqeqref{P0}&\tikzfigS{./qcprime-completeness/euler_case0_01}\\
    &\eqquatreeqref{gphaseaddition}{Paddition}{RXdef}{HH}&\tikzfigS{./qcprime-completeness/euler_case0_02}
  \end{eqnarray*}

  If $\gamma_3+\delta_1=\pi\pmod{2\pi}$, then we can assume w.l.o.g. that $\gamma_3+\delta_1=\pi$ (otherwise we can apply \Cref{P2pi,Paddition}).
  \begin{eqnarray*}
    \tikzfigS{./qcprime-completeness/euler_00}
    &=&\tikzfigS{./qcprime-completeness/euler_casepi_00}\\
    &\eqtroiseqref{ZRZ}{P2pi}{Paddition}&\tikzfigS{./qcprime-completeness/euler_casepi_01}\\
    &\eqquatreeqref{gphaseaddition}{Paddition}{RXdef}{HH}&\tikzfigS{./qcprime-completeness/euler_casepi_02}
  \end{eqnarray*}
  In both cases we have a circuit of the form $\tikzfigS{qc-axioms/euler-right}$. We conclude the proof 
  in the same way as for \Cref{lem:procedureNormalFormEuler}, except that instead of applying \cref{euler} with $R_X(\pi)$ and $R_X(-\pi)$ gates, we apply \cref{XPX} after remarking that according to \cref{Zdef,RXdef,Xdef}, an $R_X(\pi)$ gate is nothing but an $X$ gate up to a global phase.
\end{proof}

\subsection{Proof of Proposition \ref{prop:eulerprimeclifford}}\label{appendix:minimalityQCprime}
\eulerprimeclifford*
\begin{proof}
  \fbox{$\implies$} Supposing $\alpha_i'=0\pmod{\frac{\pi}{2}}$ for $i\in\{1,3\}$, we have to show that $\beta_i'=0\pmod{\frac{\pi}{2}}$ for $i\in\{1,2,3\}$. As the functions used to compute the angles $\beta_i'$ are $4\pi$-periodic, there is only a finite number on instances to check. We can do so by direct computation. Note that in the three cases (corresponding to the three cases to compute the angles $\beta_i'$) we can also reduce the number of instances to check as follows.

  \noindent\textbf{Case $\bf{z'=0}$.}
  \begin{gather*}
    z'=0\implies\begin{cases}
      \sin(\frac{\alpha_1'+\alpha_3'}{2})=0\\
      \cos(\frac{\alpha_1'-\alpha_3'}{2})=0
    \end{cases}\implies\begin{cases}
      \frac{\alpha_1'+\alpha_3'}{2}=0\pmod{\pi}\\
      \frac{\alpha_1'-\alpha_3'}{2}=\frac{\pi}{2}\pmod{\pi}
    \end{cases}\implies\begin{cases}
      \alpha_1'=\frac{\pi}{2}\pmod{\pi}\\
      \alpha_3'=-\frac{\pi}{2}\pmod{\pi}
    \end{cases}
  \end{gather*}

  \noindent\textbf{Case $\bf{z=0}$.}
  \begin{gather*}
    z=0\implies\begin{cases}
      \cos(\frac{\alpha_1'+\alpha_3'}{2})=0\\
      \sin(\frac{\alpha_1'-\alpha_3'}{2})=0
    \end{cases}\implies\begin{cases}
      \frac{\alpha_1'+\alpha_3'}{2}=\frac{\pi}{2}\pmod{\pi}\\
      \frac{\alpha_1'-\alpha_3'}{2}=0\pmod{\pi}
    \end{cases}\implies\begin{cases}
      \alpha_1'=\frac{\pi}{2}\pmod{\pi}\\
      \alpha_3'=\frac{\pi}{2}\pmod{\pi}
    \end{cases}
  \end{gather*}

  \noindent\textbf{Case $\bf{z\ne0}$ and $\bf{z'\ne0}$.} In this case, we can even easily prove the statement. First note that $\nicefrac{\alpha_1'+\alpha_3'}{2}=0\pmod{\nicefrac{\pi}{4}}$. We distinguish the case $\nicefrac{\alpha_1'+\alpha_3'}{2}=0\pmod{\nicefrac{\pi}{2}}$ and the case $\nicefrac{\alpha_1'+\alpha_3'}{2}=\nicefrac{\pi}{4}\pmod{\nicefrac{\pi}{2}}$.

  \textit{Subcase $\nicefrac{\alpha_1'+\alpha_3'}{2}=0\pmod{\nicefrac{\pi}{2}}$.} As $\alpha_i'=0\pmod{\nicefrac{\pi}{2}}$, it is also the case that $\nicefrac{\alpha_1'-\alpha_3'}{2}=0\pmod{\nicefrac{\pi}{2}}$. Then, one can check that $\arg(z)=0\pmod{\nicefrac{\pi}{2}}$, $\arg(z')=0\pmod{\nicefrac{\pi}{2}}$ and $|z|=|z'|$. Hence we get $\beta_1'=\arg(z){+}\arg(z')=0\pmod{\nicefrac{\pi}{2}}$, $\beta_2'=2\arg\left(i{+}\left\lvert\frac{z}{z'}\right\rvert\right)=\nicefrac{\pi}{2}$ and $\beta_3'=\arg(z){-}\arg(z')=0\pmod{\nicefrac{\pi}{2}}$.

  \textit{Subcase $\nicefrac{\alpha_1'+\alpha_3'}{2}=\nicefrac{\pi}{4}\pmod{\nicefrac{\pi}{2}}$.} As $\alpha_i'=0\pmod{\nicefrac{\pi}{2}}$, it is also the case that $\nicefrac{\alpha_1'-\alpha_3'}{2}=\nicefrac{\pi}{4}\pmod{\nicefrac{\pi}{2}}$. Then, one can check that $\arg(z)=\nicefrac{\pi}{4}\pmod{\nicefrac{\pi}{2}}$, $\arg(z')=\nicefrac{\pi}{4}\pmod{\nicefrac{\pi}{2}}$ and $|z|=|z'|$. Hence we get $\beta_1'=\arg(z){+}\arg(z')=0\pmod{\nicefrac{\pi}{2}}$, $\beta_2'=2\arg\left(i{+}\left\lvert\frac{z}{z'}\right\rvert\right)=\nicefrac{\pi}{2}$ and $\beta_3'=\arg(z){-}\arg(z')=0\pmod{\nicefrac{\pi}{2}}$.

  \vspace{1em}
  \noindent\fbox{$\impliedby$} Supposing $\beta_i'=0\pmod{\frac{\pi}{2}}$ for $i\in\{1,2,3\}$, we have to show that $\alpha_i'=0\pmod{\frac{\pi}{2}}$ for $i\in\{1,3\}$. We distinguish three cases, corresponding to the three cases to compute the angles $\beta_i'$. In the following we use the notation $H\defeq\interp{\gH}$, $P(\varphi)\defeq\interp{\gP}$ and $R_X(\theta)\defeq\interp{\gRX}$.

  \noindent\textbf{Case $\bf{z'=0}$.} In this case $\beta_2'=\beta_3'=0$, then by soundness of \Cref{eulerprime} we have
  \begin{gather*}
    R_X(\alpha_3')HR_X(\alpha_1')=e^{i\beta_0'}P(0)R_X(0)P(\beta_1')
    \overset{\eqref{RXdef}}{\iff} e^{-i(\nicefrac{\alpha_1'}{2}+\nicefrac{\alpha_3'}{2})}HP(\alpha_3')HP(\alpha_1')H=e^{i\beta_0'}P(\beta_1')\\
    \iff P(\alpha_3')HP(\alpha_1')=e^{i(\beta_0'+\nicefrac{\alpha_1'}{2}+\nicefrac{\alpha_3'}{2})}HP(\beta_1')H\\
    \iff \frac{1}{\sqrt{2}}\begin{pmatrix}
      1 & e^{i\alpha_1'}\\
      e^{i\alpha_3'} & -e^{i(\alpha_1'+\alpha_3')}
    \end{pmatrix} = e^{i(\beta_0'+\nicefrac{\alpha_1'}{2}+\nicefrac{\alpha_3'}{2}+\nicefrac{\beta_1'}{2})}\begin{pmatrix}
      \cos(\nicefrac{\beta_1'}{2}) & -i\sin(\nicefrac{\beta_1'}{2})\\
      -i\sin(\nicefrac{\beta_1'}{2}) & \cos(\nicefrac{\beta_1'}{2})
    \end{pmatrix}\\
    \implies\begin{cases}
      -e^{i(\alpha_1'+\alpha_3')}=1\\
      e^{i\alpha_1'}=e^{i\alpha_3'}
    \end{cases}
    \implies\begin{cases}
      \alpha_1'+\alpha_3'+\pi=0\pmod{2\pi}\\
      \alpha_1'=\alpha_3'\pmod{2\pi}
    \end{cases}
    \implies\begin{cases}
      \alpha_1'=\frac{\pi}{2}\pmod{\pi}\\
      \alpha_3'=\frac{\pi}{2}\pmod{\pi}
    \end{cases}\implies\begin{cases}
      \alpha_1'=0\pmod{\nicefrac{\pi}{2}}\\
      \alpha_3'=0\pmod{\nicefrac{\pi}{2}}
    \end{cases}
  \end{gather*}

  \noindent\textbf{Case $\bf{z=0}$.} In this case $\beta_2'=\pi$ and $\beta_3'=0$, then by soundness of \Cref{eulerprime} we have 
  \begin{gather*}
    R_X(\alpha_3')HR_X(\alpha_1')=e^{i\beta_0'}P(0)R_X(\pi)P(\beta_1')
    \iff R_X(\alpha_3'-\pi)HR_X(\alpha_1')=e^{i\beta_0'}P(\beta_1')\\
    \overset{\eqref{RXdef}}{\iff} e^{-i(\nicefrac{\alpha_1'}{2}+\nicefrac{\alpha_3'}{2}-\nicefrac{\pi}{2})}HP(\alpha_3'-\pi)HP(\alpha_1')H=e^{i\beta_0'}P(\beta_1')
    \iff P(\alpha_3'-\pi)HP(\alpha_1')=e^{i(\beta_0'+\nicefrac{\alpha_1'}{2}+\nicefrac{\alpha_3'}{2}-\nicefrac{\pi}{2})}HP(\beta_1')H\\
    \iff \frac{1}{\sqrt{2}}\begin{pmatrix}
      1 & e^{i\alpha_1'}\\
      e^{i(\alpha_3'-\pi)} & -e^{i(\alpha_1'+\alpha_3'-\pi)}
    \end{pmatrix} = e^{i(\beta_0'+\nicefrac{\alpha_1'}{2}+\nicefrac{\alpha_3'}{2}-\nicefrac{\pi}{2}+\nicefrac{\beta_1'}{2})}\begin{pmatrix}
      \cos(\nicefrac{\beta_1'}{2}) & -i\sin(\nicefrac{\beta_1'}{2})\\
      -i\sin(\nicefrac{\beta_1'}{2}) & \cos(\nicefrac{\beta_1'}{2})
    \end{pmatrix}\\
    \implies\begin{cases}
      -e^{i(\alpha_1'+\alpha_3'-\pi)}=1\\
      e^{i\alpha_1'}=e^{i(\alpha_3'-\pi)}
    \end{cases}
    \implies\begin{cases}
      \alpha_1'+\alpha_3'=0\pmod{2\pi}\\
      \alpha_1'=\alpha_3'-\pi\pmod{2\pi}
    \end{cases}
    \implies\begin{cases}
      \alpha_1'=\frac{\pi}{2}\pmod{\pi}\\
      \alpha_3'=-\frac{\pi}{2}\pmod{\pi}
    \end{cases}\implies\begin{cases}
      \alpha_1'=0\pmod{\nicefrac{\pi}{2}}\\
      \alpha_3'=0\pmod{\nicefrac{\pi}{2}}
    \end{cases}
  \end{gather*}

  \noindent\textbf{Case $\bf{z\ne0}$ and $\bf{z'\ne0}$.} In this case $\beta_2'=0\pmod{\nicefrac{\pi}{2}}$ implies
  \begin{equation*}
    \frac{\beta_2'}{2}=\arg\left(i+\frac{|z|}{|z'|}\right)=\arctan\left(\frac{\sqrt{\cos^2(\nicefrac{\alpha_1'+\alpha_3'}{2})+\sin^2(\nicefrac{\alpha_1'-\alpha_3'}{2})}}{\sqrt{\sin^2(\nicefrac{\alpha_1'+\alpha_3'}{2})+\cos^2(\nicefrac{\alpha_1'-\alpha_3'}{2})}}\right)=0\pmod{\nicefrac{\pi}{4}}
  \end{equation*}

  \textit{Subcase $\nicefrac{\beta_2'}{2}=0\pmod{\nicefrac{\pi}{2}}$.} In this case, using $\arctan(x)=0\pmod{\nicefrac{\pi}{2}}\implies x=0$ and $\nicefrac{\beta_2'}{2}\in{]0,\nicefrac{\pi}{2}[}$, we have $\cos^2(\nicefrac{\alpha_1'+\alpha_3'}{2})+\sin^2(\nicefrac{\alpha_1'-\alpha_3'}{2})=0$ which implies
  \begin{gather*}
    \begin{cases}
      \cos^2(\frac{\alpha_1'+\alpha_3'}{2})=0\\
      \sin^2(\frac{\alpha_1'-\alpha_3'}{2})=0
    \end{cases}\implies\begin{cases}
      \frac{\alpha_1'+\alpha_3'}{2}=\frac{\pi}{2}\pmod{\pi}\\
      \frac{\alpha_1'-\alpha_3'}{2}=0\pmod{\pi}
    \end{cases}\implies\begin{cases}
      \alpha_1'+\alpha_3'=\pi\pmod{2\pi}\\
      \alpha_1'-\alpha_3'=0\pmod{2\pi}
    \end{cases}\\
    \implies\begin{cases}
      \alpha_1'=\frac{\pi}{2}\pmod{\pi}\\
      \alpha_3'=\frac{\pi}{2}\pmod{\pi}
    \end{cases}\implies\begin{cases}
      \alpha_1'=0\pmod{\nicefrac{\pi}{2}}\\
      \alpha_3'=0\pmod{\nicefrac{\pi}{2}}
    \end{cases}
  \end{gather*}

  \textit{Subcase $\nicefrac{\beta_2'}{2}=\nicefrac{\pi}{4}\pmod{\nicefrac{\pi}{2}}$.} In this case, using $\arctan(x)=\nicefrac{\pi}{4}\pmod{\nicefrac{\pi}{2}}\implies x=\pm 1$ and $\nicefrac{\beta_2'}{2}\in{]0,\nicefrac{\pi}{2}[}$, we have
  \begin{gather*}
    \sqrt{\cos^2(\nicefrac{\alpha_1'+\alpha_3'}{2})+\sin^2(\nicefrac{\alpha_1'-\alpha_3'}{2})}=\sqrt{\sin^2(\nicefrac{\alpha_1'+\alpha_3'}{2})+\cos^2(\nicefrac{\alpha_1'-\alpha_3'}{2})}\\
    \implies\cos^2(\nicefrac{\alpha_1'+\alpha_3'}{2})+\sin^2(\nicefrac{\alpha_1'-\alpha_3'}{2})=\sin^2(\nicefrac{\alpha_1'+\alpha_3'}{2})+\cos^2(\nicefrac{\alpha_1'-\alpha_3'}{2})\\
    \implies 1-\sin^2(\nicefrac{\alpha_1'+\alpha_3'}{2})+\sin^2(\nicefrac{\alpha_1'-\alpha_3'}{2})=\sin^2(\nicefrac{\alpha_1'+\alpha_3'}{2})+1-\sin^2(\nicefrac{\alpha_1'-\alpha_3'}{2})\\
    \implies \sin^2(\nicefrac{\alpha_1'-\alpha_3'}{2})=\sin^2(\nicefrac{\alpha_1'+\alpha_3'}{2})\implies \sin(\nicefrac{\alpha_1'-\alpha_3'}{2})=\pm\sin(\nicefrac{\alpha_1'+\alpha_3'}{2})\\
    \implies\frac{\alpha_1'-\alpha_3'}{2}=\pm\frac{\alpha_1'+\alpha_3'}{2}\pmod{\pi}\implies\alpha_1'-\alpha_3'=\pm\alpha_1'\pm\alpha_3'\pmod{2\pi}
  \end{gather*}
  \noindent If $\alpha_1'{-}\alpha_3'=\alpha_1'{+}\alpha_3'\pmod{2\pi}$ then $\alpha_3'=0\pmod{\pi}$ and by computing the value of $\beta_3'$ (as the functions are $4\pi$-periodic in $\alpha_3'$ we just need to check the computation for the four instances $\alpha_3'\in\{0,\pi,2\pi,3\pi\}$) we obtain $\beta_3'=\pm\alpha_1'\pmod{\nicefrac{\pi}{2}}$ which together with $\beta_3'=0\pmod{\nicefrac{\pi}{2}}$ implies that $\alpha_1'=0\pmod{\nicefrac{\pi}{2}}$.

  \noindent If $\alpha_1'{-}\alpha_3'={-}\alpha_1'{-}\alpha_3'\pmod{2\pi}$ then $\alpha_1'=0\pmod{\pi}$ and by computing the value of $\beta_1'$ (as the functions are $4\pi$-periodic in $\alpha_1'$ we just need to check the computation for the four instances $\alpha_1'\in\{0,\pi,2\pi,3\pi\}$) we obtain $\beta_1'=\pm\alpha_3'\pmod{\nicefrac{\pi}{2}}$ which together with $\beta_1'=0\pmod{\nicefrac{\pi}{2}}$ implies that $\alpha_3'=0\pmod{\nicefrac{\pi}{2}}$.
\end{proof}

\end{document}